\documentclass[12pt]{article}
\pdfoutput=1

\usepackage{xr-hyper}
\usepackage[margin=1in]{geometry}
\usepackage{amsfonts,amsmath,amssymb,amsthm,graphicx,bm}
\usepackage{mathtools}
\usepackage{mathrsfs}
\usepackage{ifthen}
\usepackage{changepage}
\usepackage{enumerate}
\usepackage{scalerel}
\usepackage{accents}
\usepackage{enumitem}
\usepackage{float}      
\usepackage[utf8]{inputenc} 
\usepackage[T1]{fontenc} 
\usepackage{csquotes}
\usepackage{comment}
\usepackage{bbm}
\usepackage{array}
\usepackage{makecell}
\usepackage{booktabs}
\usepackage{xfrac}
\usepackage{thm-restate}
\usepackage{tikz}
\usepackage{pgfplots}
\usepackage{subcaption}
\usepackage{graphicx}
\usepackage{multirow}
\usetikzlibrary{patterns}
\usepgfplotslibrary{fillbetween}
\usetikzlibrary{intersections}
\usepackage{wrapfig}
\usetikzlibrary{positioning,fit,calc,arrows.meta}
\usepackage{fix-cm}
\usepackage[ruled,vlined,linesnumbered]{algorithm2e}
\usepackage{xcolor}
\usepackage{natbib}
\bibpunct{(}{)}{;}{a}{,}{,}

\usepackage{standalone}  
\usetikzlibrary{patterns, decorations.pathreplacing} 
\usetikzlibrary{arrows.meta, calc}

\usepackage[suppress]{color-edits}
\addauthor{zs}{blue}
\addauthor{wt}{purple}
\addauthor{wtres}{red}

\newcommand{\xhdr}[1]{\vspace{6pt}\noindent{\bf {#1.}}}

\allowdisplaybreaks
\theoremstyle{plain}
\newtheorem{theorem}{Theorem}[section]
\newtheorem{lemma}[theorem]{Lemma}

\newtheorem{proposition}[theorem]{Proposition}

\theoremstyle{plain}
\newtheorem{definition}{Definition}[section] 

\theoremstyle{plain}


\usepackage[colorlinks=true, linkcolor=blue!70!black, citecolor=blue!70!black, urlcolor=blue, breaklinks=true]{hyperref}
\usepackage{cleveref}
\Crefname{algocf}{Algorithm}{Algorithms}
\crefname{claim}{claim}{claims}

\newcommand{\Payoff}[2][]{\text{\bf Payoff}\ifthenelse{\not\equal{}{#1}}{_{#1}}{}\!\left[{\def\givenn{\middle|}#2}\right]}
\newcommand{\Reg}[2][]{\text{\bf REG}\ifthenelse{\not\equal{}{#1}}{_{#1}}{}\!\left[{\def\givenn{\middle|}#2}\right]}

\newcommand{\Rev}[2][]{\text{\bf Rev}\ifthenelse{\not\equal{}{#1}}{_{#1}}{}\!\left[{\def\givenn{\middle|}#2}\right]}
\newcommand{\optRev}[2][]{\text{\bf Rev}^\dagger\ifthenelse{\not\equal{}{#1}}{_{#1}}{}\!\left[{\def\givenn{\middle|}#2}\right]}

\newcommand{\dd}{\mathrm{d}}
\newcommand{\biddernum}{N}
\newcommand{\val}{v}
\renewcommand{\val}{v}
\newcommand{\valprofile}{v}
\newcommand{\valUB}{\bar{\val}}
\newcommand{\valLB}{\underline{\val}}

\newcommand{\signal}{s}
\newcommand{\signalprofile}{s}
\newcommand{\signalspace}{S}
\newcommand{\signalprofilespace}{S}
\newcommand{\valspace}{V}
\newcommand{\valprofilespace}{\valspace^\biddernum}
\newcommand{\valprior}{\lambda}

\newcommand{\inforstructure}{\mathcal{I}}
\newcommand{\signalscheme}{\pi}
\newcommand{\margsignalscheme}{\pi_{\signalprofilespace}}
\newcommand{\bidstrategy}{\sigma}
\newcommand{\bidstrategyprofile}{\sigma}
\newcommand{\bid}{b}
\newcommand{\bidprofile}{b}
\newcommand{\bidspace}{B}
\newcommand{\bidstrategyspace}{\Sigma}
\newcommand{\middlepoint}{\kappa}
\newcommand{\inforstructureGeneral}{\inforstructure}

\newcommand{\biddersurplus}{\textsc{BS}}
\newcommand{\winprob}{\omega}
\newcommand{\secmax}{\textsc{secmax}}
\renewcommand{\Rev}{\textsc{Rev}}
\newcommand{\totalsurplus}{\textsc{WEL}}
\newcommand{\optwelfare}{\overline{\totalsurplus}}
\newcommand{\worstwelfare}{\underline{\totalsurplus}}
\newcommand{\inforstructureUnique}{\mathcal{I}^{\dagger}}

\newcommand{\winnerNum}{m}
\newcommand{\modifiedpriorCDF}{\widehat{\valprior}}
\newcommand{\lowerSignal}{\underline{\signal}}
\newcommand{\upperSignal}{\bar{\signal}}

\newcommand{\eps}{\varepsilon}

\newcommand{\posteriorBeliefDist}{H}

\newcommand{\permutation}{\xi}
\newcommand{\construcrule}{{\text{Winner-Dominance}}}
\newcommand{\zerorentcondi}{{\text{No-Rent Winner}}}

\newcommand{\symsignalprofilespace}{\signalprofilespace^{\biddernum}}

\definecolor{efficiencyColor}{HTML}{000000}
\definecolor{inefficiencyColor}{HTML}{404040}
\definecolor{regionFill}{HTML}{D9D9D9}

\newcommand{\winnerMean}{\hat{\val}}
\newcommand{\winnerExpectedvalProb}{g}
\newcommand{\minVal}{\underline{\val}}
\newcommand{\maxVal}{\valUB}

\newcommand{\constructsignalspace}{\bar{\signalspace}}
\newcommand{\constructsignalscheme}{\hat{\signalscheme}}
\newcommand{\winnerpriorCDF}{F}
\newcommand{\winnerposteriorCDF}{G}
\newcommand{\secmaxsignal}{h}
\newcommand{\jointdist}{f}
\newcommand{\Length}{L}
\newcommand{\strictsignalspace}{\hat{\signalspace}}
\newcommand{\signalProfileDist}{\phi}
\newcommand{\MPCfunction}{\Psi}
\newcommand{\conditionalDist}{Q}

\newcommand{\uniqueExpectedval}{y}
\newcommand{\constant}{C}
\newcommand{\positiveProb}{\omega}
\newcommand{\quantileBid}{\bid^\star}

\newcommand{\Devbidstrategyprofile}{\bidstrategyprofile^\dagger}

\newcommand{\selectvalprofile}{\hat{\val}}
\newcommand{\selectset}{B}
\newcommand{\maxbidder}{i^\star}
\newcommand{\deviationbid}{\bid\primed}
\newcommand{\mapping}{\gamma}
\newcommand{\uniform}{\mathrm{Unif}}
\newcommand{\lsignal}{\signal^{\cc{low}}}
\newcommand{\hsignal}{\signal^{\cc{high}}}

\newcommand{\expostbiddersurplus}{\Delta}
\newcommand{\event}{\mathcal{E}}
\newcommand{\setsize}[1]{{\left|#1\right|}}

\newcommand{\condition}{\,\mid\,}

\newcommand{\prob}[2][]{\text{Pr}\ifthenelse{\not\equal{}{#1}}{_{#1}}{}\!\left[{\def\givenn{\middle|}#2}\right]}
\newcommand{\indicator}[2][]{\mathbf{1}\ifthenelse{\not\equal{}{#1}}{_{#1}}{}\!\left\{{\def\givenn{\middle|}#2}\right\}}
\newcommand{\expect}[2][]{\mathbb{E}\ifthenelse{\not\equal{}{#1}}{_{#1}}{}\!\left[{\def\givenn{\middle|}#2}\right]}

\newcommand{\cc}[1]{\ensuremath{\mathsf{#1}}} 

\DeclareMathOperator*{\argmax}{arg\,max}

\newcommand{\R}{\mathbb{R}}
\newcommand{\reals}{\R}

\newcommand{\valCDF}{F}

\newcommand{\expectedVal}{x}

\newcommand{\mymax}{\max}

\newcommand{\supp}{\mathrm{supp}}

\newcommand{\primed}{^{\dagger}}
\newcommand{\doubleprimed}{^{\ddagger}}

\hypersetup{filecolor=blue!70!black}
\usepackage{setspace}
\title{Private Private Information in Second-Price Auction}
\author{
Boyu Liu\thanks{Renmin University of China. Email: {\tt liuboyu03@ruc.edu.cn}}
\and
Wei Tang\thanks{Chinese University of Hong Kong. Email: {\tt weitang@cuhk.edu.hk}}
\and
Zihe Wang\thanks{Renmin University of China. Email: {\tt wang.zihe@ruc.edu.cn}}
\and
Shuo Zhang\thanks{Renmin University of China. Email: {\tt zhangshuo1422@ruc.edu.cn}}
}
\date{}

\begin{document}

\maketitle

\begin{abstract}

Classic results show that even an arbitrarily small correlation across bidders' information can enable full surplus extraction in auctions and related mechanism design settings. 
Motivated by this fragility,
we study the information independence
in a second-price auction when the seller commits to a private private information structure, meaning bidders' signals are independent ex ante, while bidders share a symmetric and arbitrarily correlated prior distribution over their valuations.
We first show that the seller optimal efficient outcome with full surplus extraction can always be implemented by a private private information structure that admits a Bayes Nash equilibrium.
However, this equilibrium may not be stable. We then further construct a private private information structure that achieves revenue arbitrarily close to maximum welfare while admitting a strict equilibrium.
At the same time, we establish an impossibility result: under private private information, in general, bidder surplus cannot achieve maximal welfare exactly, and we characterize necessary and sufficient conditions on the prior distribution under which bidder surplus can be made arbitrarily close to maximal welfare. 
We finally explore which other efficient outcomes are achievable under private private information.
\end{abstract}


\newpage

\section{Introduction}
\label{sec:intro}

Early auction and mechanism design work emphasizes that, in interdependent value settings (including common value environments), 
bidders' information rents depend critically on whether their information remains private across bidders.
Indeed, \cite{MW-82} show by a common-value auction example that when one bidder's information is fully revealed by another's, the less informed bidder earns no surplus.
\cite{MR-92} generalize this insight by arguing that even arbitrarily small correlations let a designer use others' reports as ``cross-checks'' and drive each bidder’s interim informational rent to (almost) zero.
As they put it, ``introducing arbitrarily small amounts of correlation into the joint distribution of private information among the players is enough to render private information valueless.''\footnote{\cite{CM-88} also show that with correlated private values, a profit maximizing seller can design mechanisms that fully extract bidders' surplus by exploiting cross bidder correlation also through report-contingent transfers.}



Information independence, implicitly or explicitly, is also suggested in recent work to be crucial for sustaining bidders' surplus.
For example, in \cite{BBM-17}, the information structure used to attain the maximal bidder surplus in a first-price auction can be implemented with a particularly simple design in which bidders' signals are essentially independent and symmetric across bidders, while the information structure to minimize the bidders' surplus in an efficient allocation uses correlation across bidders' signals. 
Similarly, \cite{BBX-18} construct a revenue maximizing information structure for the second price auction that uses correlated, bidder specific signals to fine tune beliefs to extract nearly all surplus.
Recent work \citep{KFO-25} also mentions that ``compromising agents' privacy leads to undesirable outcomes (such as full surplus extraction).''
Although these results do not explicitly establish information or signal independence as a necessary condition for maximizing bidder surplus, they leave open a basic question: even in one of the most fundamental auction formats, the second price auction, does requiring bidders' additional information to be independent constrain the seller's ability to shape beliefs and extract surplus? In particular, under ex ante independent signals, can the seller still design an efficient information structure that approximately extracts bidders’ surplus, or does independence impose a fundamental barrier?

Motivated by these questions,\footnote{Beyond the surplus motivation, \wtresedit{information} independence is also practically motivated by privacy concerns in platform markets such as ad auctions, where the platform sends each bidder a bidder-specific score derived from sensitive user data. If a bidder can use its own signal to infer what rivals were shown, then ``private'' disclosure becomes a channel for cross-bidder inference and privacy leakage.} 
we study how \wtresedit{information} independence affects welfare outcomes, with particular emphasis on efficient allocations, in a second price auction where the seller commits to a private private information structure.
We adopt the notion of private private information introduced in \cite{KFO-25}, which requires bidders' information signals to be independent ex ante, \wtedit{so a bidder's signal realization does not reveal information about other bidders' signals}.\footnote{
We adopt the privacy notion proposed by \citet{KFO-25} since it admits a clean instantiation for imposing privacy across bidders  in our multi-bidder (i.e., multi-receiver) setting. 
See \Cref{sec:related-work} for more discussion about other privacy notions.}
Bidders nevertheless share a symmetric (exchangeable) prior over their valuations that may exhibit arbitrary correlation across bidders.
Because signals are independent ex ante, the seller cannot jointly manipulate bidders' beliefs in a coordinated way; 
this limits how aggressively the seller can raise the second-highest bid while keeping the allocation efficient, and it also constrains the extent to which the seller can lower the second-highest bid in the resulting equilibrium to transfer surplus to bidders.

Our first set of results shows that the private private information does not preclude full bidder surplus extraction, even under efficient allocation, and even within a particularly simple symmetric subclass of information structures.
We prove that there exists a private private information structure for which at least one induced Bayes Nash equilibrium is efficient and seller extracts full surplus, so the seller's expected revenue equals the maximum welfare (see \Cref{thm:pripri max rev}). 
In our construction, we restrict attention to a symmetric subclass of private private structures in which bidders share a common signal space, each bidder's signal has the same uniform marginal distribution on this space, and the information structure treats bidders symmetrically.
This restriction is deliberately simple, yet full surplus extraction remains achievable under an efficient allocation.

Our construction is guided by two principles that together guarantee full surplus extraction. 
The first principle is about efficient allocation. 
We design the joint relationship between values and signals so that whenever a bidder receives the highest signal, that bidder must also be the highest valuation bidder in every posterior state that the signal profile deems possible. 
Intuitively, the top signal serves as a certificate of being the efficient winner.
This ensures that the second-price auction allocates the item efficiently in the equilibrium we target, even though bidders do not observe their values.
The second principle aims to eliminate the winner's informational rent. 
Conditional on a bidder winning, we design beliefs so that the winner's posterior expected valuation equals the second highest signal in the realized signal profile. 
Equivalently, 
after observing the signal profile, the expected value advantage of holding the top signal is exactly offset, so the winner's posterior expected value is pinned down by the second highest signal rather than by the winner's own signal.
In a second-price auction, the payment is given by the second highest bid, so in the equilibrium, where we show that bidders would follow the pure strategy of bidding their signal, the winner pays exactly what the second highest signal suggests. Because we align the winner's expected value with that payment, the winner has no expected surplus left to capture. 
The seller therefore collects essentially the entire welfare as revenue.

A key intuition is that the private private constraint restricts only the ex ante dependence across bidders' signals, but it does not prevent the seller from coupling each bidder's signal with the underlying value profile in a way that creates strong competition in equilibrium. 
Our symmetric construction exploits this by choosing a carefully designed signal space and then selecting posterior mappings so that the induced distribution of winner posteriors is compatible with Bayes plausibility. At a high level, we first specify the distribution of the second highest signal conditional on winning, which determines the distribution of the winner's posterior expected value that we wish to implement. We then verify that this target distribution can indeed arise from the modified prior faced by the winner, and we build an information structure that realizes these posteriors while keeping signals independent and symmetric ex ante.


The equilibrium induced by the above private private information that extracts full surplus may not be stable (as there could be multiple best responses for bidders when fixing other bidders' bidding strategies). 
To address this issue, we show that there always exists a private private information that guarantees a {\em strict} Nash equilibrium while still achieving revenue arbitrarily close to the maximum welfare (see \Cref{prop:pripri approx max rev}). 
Under the strict equilibrium, each bidder has a unique best response to the others' equilibrium actions.
As in our full extraction result, we focus on the same symmetric private private class.
The key intuition behind equilibrium strictness is to eliminate flat regions in bidders' incentives. 
In the full extraction construction, some bidders may be indifferent among a range of bids at certain signals, which may make the resulting equilibrium unstable.
We therefore modify the information structure by adding a carefully designed slack that breaks these indifferences. 
Concretely, we perturb posteriors so that for every signal realization, deviating from the intended bid becomes strictly suboptimal, while the perturbation is small enough that the induced allocation remains efficient and the seller's revenue remains arbitrarily close to maximum welfare. 
This gives us a strict equilibrium that achieves near full surplus extraction, showing that ex ante independence does not prevent robust revenue guarantees even under a stringent stable equilibrium requirement.

Having characterized how the seller can extract bidders’ surplus under private private information, we next turn to the opposite extreme and ask whether the same ex ante signal-independence constraint can also support outcomes that are most favorable to bidders, specifically, whether bidder surplus can reach the level of maximal welfare. 
We establish an impossibility result: we show that, in general, under any private private information that admits an efficient equilibrium, bidder surplus cannot equal maximal welfare (\Cref{thm:pripri max bidder surplus}).
At the same time, we fully characterize the necessary and sufficient conditions on the prior value distribution under which bidder surplus can be made arbitrarily close to maximal welfare through some private private information structure. Specifically, we show that this is possible if and only if for every realized value profile, either bidders share a common value or at most one bidder has a positive value (\Cref{thm:pripri max bidder surplus iff}).

This limitation is specific to the private private requirement. 
The fundamental challenge is that ex ante signal independence limits the seller's ability to coordinate bidders' beliefs across bidders.
\wtresedit{In a second-price auction, maximizing bidder surplus requires both efficient allocation and an almost zero second-highest bid: the winner must be willing to win, while the losing bidders must not bid up the price. Under private private information, these two requirements conflict. If two bidders bid positively with positive probability, independence creates overlap events in which both bids are positive, generating positive revenue. If such overlap is ruled out, efficiency generally fails because different bidders must win in different states. As a result, except under the special priors characterized above, private private information cannot make bidder surplus arbitrarily close to maximal welfare, and exact equality is possible only in the common-value case.}
\wtresedit{By contrast, under general information structures, correlated signals can coordinate which bidder is induced to bid aggressively and which bidders are induced to submit low bids across states, so bidder surplus can achieve maximal welfare, as we show in the online appendix.}
This observation stands in contrast to the first-price auction setting studied in \citet{BBM-17}, where the bidder-surplus-maximizing outcome can already be implemented within a private private information structure. 
\wtresedit{The difference is that in first-price auctions the winner pays her own bid, whereas in second-price auctions the winner's surplus is pinned down by rivals' bids; hence bidder-optimal surplus requires precisely the cross-bidder coordination that private private information rules out.}

The above results also raise the following natural question: what other efficient outcomes are achievable under private private information?
We first revisit the prior conditions identified in \Cref{thm:pripri max bidder surplus iff} under which bidders' maximal surplus can approach maximal welfare, and we ask whether, under these priors, every efficient outcome is achievable by some private private information.
We show that for any such prior distribution, every efficient welfare outcome can be implemented by a private private information structure, except the endpoint in which bidder surplus equals maximal welfare (\Cref{thm:iff pri pri efficient}). 
We then turn to another well studied setting, the independent private values (IPV) model. In this setting, revealing full information is itself a private private information, because bidders' values are independent, and it admits an efficient equilibrium. 
We show that every efficient welfare outcome between the seller optimal outcome and the outcome induced by full information is implementable by some private private information (\Cref{thm:ipv pri pri efficient}).

\subsection{Related Work}
\label{sec:related-work}

Our work relates to the growing literature on  information design in auction settings.
The most closely related works are \cite{BBM-17,BBX-18,BHMSW-22}.
As mentioned, \cite{BBM-17} study how general information structures affect outcomes in first-price auctions and characterize the feasible region of (bidder surplus, revenue) outcomes.
In the same spirit, we characterize the feasible region in a second-price auction.
Like ours, \cite{BBX-18,BHMSW-22} study the information design in second-price auctions. 
However, \cite{BHMSW-22} restrict attention to independent private values and study exclusively independent and symmetric information policies, which correspond to a special case of private private structures under independent priors.
Thus, our work is a strict generalization of their work as we allow arbitrarily correlated prior valuation distributions. 
\cite{BBX-18} focus on seller-optimal information policies and construct correlated private signaling schemes that achieve near full surplus extraction.
We show that the seller can achieve exact full surplus extraction even within the private private information structures. 
Additional related work considers variants of (generalized) second-price auctions (see, e.g., \citealp{BS-12,EFGPT-14, BDPZ-22,DTWZ-25}) or other auction formats like all-pay auctions (see, e.g., \citealp{CLW-26}).
However, they either impose different disclosure policies (for instance, public disclosure where all bidders receive the same signal) or analyze settings that differ substantially from ours.

A complementary strand studies the joint design of the mechanism and the information policy, including information design in optimal auctions following \cite{M-81}.
See, for example, \citealp{BP-07,K-20,CY-23,BBDPXZ-25,BHM-26}.
Our work also relates, more broadly, to the literature on studying equilibria (including Bayes (coarse) correlated equilibria) in auction settings, including first-price auctions (see, e.g., \citealp{BS-03,FLN-16,AB-25}), in autobidding (see, e.g., \citealp{DMMZZ-24,B+26}), and in generalized second-price auctions (see, e.g., \citealp{GS-09,CH-13}).
After receiving their private signals, bidders' valuations are interdependent, and thus, our work also relates to auctions with interdependent values (see, e.g., \citealp{MP-04,MP-15,EFFG-18,EFGMM-23}).

Privacy has been extensively studied in both computer science and economics literature. 
For example, the theoretical computer science community has formalized the influential notion of differential privacy (see, e.g., \citealp{D-06,DMNS-06}), which has since been extensively studied in many contexts (see the book \citealp{DR-14} for a comprehensive treatment). 
In economics literature, a variety of privacy notions have been proposed to suit different environments (see, e.g., \citealp{I-20,LM-20,EEM-21,SY-22,AMMO-22,SY-24,KFO-25,SY-25}).
In this work, we adopt the notion of private private information, first formalized and systematically studied in \cite{KFO-25}, because it provides a clean way to model information independence across receivers in a multi-agent environment.\footnote{\cite{ABST-21} also raise the question of which belief distributions are implementable under private private information structures.}
\cite{KFO-25} focus on characterizing the Blackwell-Pareto frontier within the class of private private information structures. By contrast, we study how private private information shapes welfare outcomes in a mechanism design setting, focusing on second-price auctions.
Another closely related privacy notion is introduced in \cite{SY-24}. As they mentioned, a key difference is that in their framework the object that must remain private is exogenously specified (for example, a protected attribute), whereas under private private information the relevant privacy constraint is endogenous, arising from the requirement that receivers' signals be independent of one another.


\section{Preliminary}
\label{sec:pre}

There are $\biddernum$ bidders who compete for an indivisible good in a second-price auction. 
Bidder $i$'s value is denoted by $\val_i\in\valspace\subseteq [0, 1]$ where $\valspace$ is a common discrete value space shared by all bidders.\footnote{Our results and analysis can be readily generalized to continuous valuation.}
The bidders' values are jointly distributed according to a prior probability measure $\valprior\in\Delta(\valprofilespace)$.
Throughout this work, we consider symmetric (i.e., exchangeable) prior $\valprior$ for the bidders' values.\footnote{A distribution over $\reals^\biddernum$ is symmetric if, for all $\valprofile\in \reals^\biddernum$, the probability it assigns to $\valprofile$ is equal to the probability
it assigns to any permutation of $\valprofile$.}


The seller can choose how much information each bidder will have about her own value $\val_i$. 
An information structure is specified by a tuple $\inforstructure = ((\signalspace_i)_{i\in[\biddernum]}, \signalscheme)$ where each $\signalspace_i$ is a measurable signal space for the bidder $i$ 
and $\signalscheme\in\Delta(\valspace^\biddernum \times \signalprofilespace)$ is a joint
probability distribution over the space $\valprofilespace$ of value profiles and the space $\signalprofilespace \triangleq (\signalspace_i)_{i\in[\biddernum]}$ of signal profiles. 

For a fixed information structure $\inforstructure$, the second-price auction is a game of incomplete information, in which bidders' strategies are measurable mappings $\bidstrategy_i: \signalspace_i \rightarrow \Delta(\bidspace)$ from the received private signal to probability measures over bids where $\bidspace = [0, \valUB]$ is the bid space and $\valUB = \max \valspace$. 
We denote by $\bidstrategyspace_i$ the set of strategies for the bidder $i$.
Given a bid profile $\bidprofile = (\bid_i)_{i\in[\biddernum]}$, we shall consider a uniform tie-breaking rule where the allocation rule will be defined as follows
\begin{align*}
    \winprob_i(\bidprofile)
    =
    \begin{cases}
    \dfrac{1}{\left|\arg\max_{j\in[\biddernum]} \bid_j\right|}, & i \in \arg\max_{j\in[\biddernum]} \bid_j,\\[6pt]
    0, & \text{else.}
    \end{cases}
\end{align*}
We use $\secmax(\bidprofile)$ to denote the second highest bid in the bid profile $\bidprofile$. 
Fixing a profile of strategies $\bidstrategyprofile \in\bidstrategyspace = (\bidstrategyspace_i)_{i\in[\biddernum]}$, 
each bidder $i$'s (ex ante) surplus and seller revenue from the auction can be computed as\footnote{\wtedit{In non-confusing contexts, we write ``revenue'' for ex ante expected revenue, ``bidder surplus'' for ex ante expected surplus.}}
\begin{align*}
    \biddersurplus_i(\inforstructure, \bidstrategyprofile)
    & \triangleq
    \mathbb{E}_{(\valprofile, \signalprofile)\sim \signalscheme}
    \expect[\bidprofile\sim \bidstrategyprofile(\cdot \mid \signalprofile)]{\winprob_i(\bidprofile)
    \cdot (\val_i - \secmax(\bidprofile))}~;\\
    \Rev(\inforstructure, \bidstrategyprofile)
    & \triangleq
    \mathbb{E}_{(\valprofile, \signalprofile)\sim \signalscheme}
    \expect[\bidprofile\sim \bidstrategyprofile(\cdot \mid \signalprofile)]{\secmax(\bidprofile)}~.
\end{align*}
We adopt the Bayes Nash equilibrium as our solution concept for the bidders' incomplete information game.\footnote{When the strategy $\bidstrategy_i$ is a
pure strategy, we slightly abuse notation and write $\bidstrategy_i(\signal_i)$ for the bid played with probability $1$.}
\begin{definition}[Bayes Nash equilibrium]
\label{defn:bne}
We say a bidding strategy profile $\bidstrategyprofile$ is a {\em Bayes Nash equilibrium}, or equilibrium for short, under the information structure $\inforstructure$ if and only if the following holds\footnote{\wtedit{When a mixed strategy profile $\bidstrategyprofile$ is used, bids are drawn according to $\bigotimes_{i\in[\biddernum]}\bidstrategy_i(\cdot\mid \signal_i)$ conditional on the signal profile $\signalprofile$. Thus mixed strategies are implemented through independent private randomizations across bidders; no public or correlated randomization device is part of the Bayes Nash equilibrium concept used here.}}
\begin{align*}
    \biddersurplus_i(\inforstructure, \bidstrategyprofile)
    \; \ge \; 
    \biddersurplus_i(\inforstructure, \bidstrategy_i\primed, \bidstrategyprofile_{-i}), \quad 
    \forall \bidstrategy_i\primed \in \bidstrategyspace_i~,\; 
    \forall i\in[\biddernum]~.
\end{align*}
\end{definition}
We say an equilibrium is a strict equilibrium if all the inequalities in \Cref{defn:bne} are the strict inequalities.\footnote{Strict Nash equilibria are known to be strategically stable in the sense of \citet{KM-86}, since each player's equilibrium action is a unique best response, so the equilibrium is robust to small perturbations \citep{KM-86,B-88}.}
Given an information structure $\inforstructure$ and an equilibrium strategy profile $\bidstrategyprofile$, we define the total surplus, i.e., the welfare, under $\inforstructure, \bidstrategyprofile$ as follows:
\begin{align*}
    \totalsurplus(\inforstructure, \bidstrategyprofile)
    \triangleq 
    \mathbb{E}_{(\valprofile, \signalprofile)\sim \signalscheme}
    \expect[\bidprofile\sim \bidstrategyprofile(\cdot \mid \signalprofile)]{\sum\nolimits_{i\in[\biddernum]}\winprob_i(\bidprofile) \cdot \val_i}.
\end{align*}
Note that $\totalsurplus(\inforstructure, \bidstrategyprofile)$ is always upper bounded by the following quantity
\begin{align*}
    \optwelfare \triangleq \expect[\valprofile\sim\valprior]{\max(\valprofile)}~.
\end{align*}
We say that $\inforstructure, \bidstrategyprofile$ is efficient if it achieves maximum welfare $\totalsurplus(\inforstructure, \bidstrategyprofile) = \optwelfare$.

Following \cite{BBM-17}, throughout our analysis, we restrict attention to strategies in which each bidder never bids an amount which is sure to be strictly greater than her value. 
Equivalently, after any signal realization, bidder $i$'s bid is supported on the interval from $0$ up to the highest valuation in bidder $i$'s posterior belief.
This condition is closely related to, though slightly weaker than, excluding weakly dominated strategies: it rules out bids that would strictly overstate value with certainty, while still allowing a bidder to bid exactly their value with positive probability.


One important class of information structures considered in this work is the private private information structures, which we define as follows:
\begin{definition}[Private private information structure \citealp{KFO-25}]
An information structure $\inforstructure$ is {\em private private} if for any $\signalprofile = (\signal_1, \ldots, \signal_\biddernum) \in \signalprofilespace$, we have
\begin{align*}
     \prob{\signalprofile} = \prod\nolimits_{i\in[\biddernum]}\prob{\signal_i}~.
\end{align*}
Namely, the joint law of $(\signal_1, \ldots, \signal_\biddernum)$ is a product measure over each bidder's signal space.
\end{definition}

\section{Maximum Revenue}
\label{sec:pri pri rev}
\label{subsec:pripri maximal rev}


We show that there exists a private private information and an efficient equilibrium in which the seller is able to fully extract bidders' surplus. 
Our results apply to general symmetric priors. 
To streamline the exposition, we focus in this section on the generic case where the maximum value $\max_i \val_i$ variable is nondegenerate. 
The only exception is the boundary case in which $\max_i \val_i$ is a constant value across all value profile realizations in $\supp(\valprior)$, which we defer this case in the online appendix. 
\footnote{This nondegeneracy assumption is used only to rule out the boundary case in which $\max_i \val_i$ is a constant value on all value profile realizations in $\supp(\valprior)$. 
This condition is much weaker than full support and is satisfied by essentially any nontrivial symmetric prior.
In that case, exact full surplus extraction is impossible under {\em any} private private information (while maintaining efficiency),  but for every $\varepsilon>0$ we can construct a private private information that preserves efficiency and achieves revenue at least $\optwelfare-\varepsilon$; \wtedit{see the discussions in \Cref{remark: prior remark} in online appendix.}}
\begin{theorem}
\label{thm:pripri max rev}
There exists an information structure $\inforstructure$ and an efficient equilibrium $\bidstrategyprofile$ such that (i) the information $\inforstructure$ is private private; 
(ii) the seller extracts full surplus, that is, 
$\Rev(\inforstructure, \bidstrategyprofile) = \optwelfare$.
\end{theorem}

\wtresedit{It is useful to compare above results with \citet{BHMSW-22}, who also study seller-optimal information disclosure in a second-price auction. Their analysis focuses on an environment with independent-private values (IPV) and restricts each bidder's signal to depend only on her own valuation. Under these restrictions, bidders' signals are independent (and thus are private private).
However, in their setting, full surplus extraction under an efficient allocation is not generally attainable. 
\Cref{thm:pripri max rev} shows that the source of this limitation is not signal independence per se. In our setting, bidders' signals must be independent ex ante, but each signal may be coupled with the entire valuation profile. 
This additional flexibility is sufficient to align the efficient allocation with zero bidder surplus. 
Thus, signal independence alone does not guarantee positive bidder surplus; what matters is whether independence is combined with the stronger restriction that each bidder's information depends only on her own valuation.

In \citet{BBM-17} who considered the first-price auction, independent signals appear in a construction that minimizes winning bids and maximizes bidder surplus; the key force is that a bidder's own payment is her own bid, so information can be designed to make bidders willing to bid low while remaining indifferent to upward deviations. In the second-price auction, by contrast, the winner's payment is determined by rivals' bids. This changes the role of signal independence. The seller does not need to make the winner bid aggressively in order to extract surplus; instead, she needs to align the rival bid that sets the price with the winner's interim valuation. Our construction does precisely this through the \ref{cond:winner} condition defined shortly.}

The equilibrium induced by the information structure constructed in \Cref{thm:pripri max rev} may not be stable. 
We therefore show that there always exists a private private information that admits a strict equilibrium and guarantees seller revenue within $\eps$ of the maximal welfare.


\begin{theorem}
\label{prop:pripri approx max rev}
For any $\varepsilon > 0$, 
there exists an information structure $\inforstructure$ that admits a {\em strict} equilibrium $\bidstrategyprofile$ such that
(i) the information structure $\inforstructure$ is private private; 
(ii) the equilibrium is efficient;
(iii) the seller's revenue is at least $\optwelfare-\varepsilon$, that is, $\Rev(\inforstructure, \bidstrategyprofile) \ge \optwelfare-\varepsilon$.
\end{theorem}

We next introduce a structured subclass of private
private information structures that will serve as the basis for our constructions. Since bidders are
ex ante identical, it is natural to focus on information structures that treat bidders symmetrically.
In particular, we restrict attention to private private structures with a common signal space and
identical uniform marginals, together with invariance to relabeling bidders. 
We now formalize this notion.
\begin{definition}[Symmetric private private information structure]
\label{def: symmetric private private information}
We say a private private information structure $\inforstructure = ((\signalspace_i)_{i\in[\biddernum]}, \signalscheme)$ is {\em symmetric} if it satisfies the following properties: 
\begin{enumerate}
    \item {\em \textbf{Common signal space.}} For all $i, j\in[\biddernum]$, $\signalspace_i \equiv\signalspace_j$.
    \item {\em \textbf{Uniform marginals.}} For each bidder $i\in[\biddernum]$, the marginal distribution of bidder $i$'s signal $\signal_i$ under $\signalscheme$ is a uniform distribution.
    \item {\em \textbf{Permutation invariant.} }
    For every permutation $\permutation$ of $[\biddernum]$ and every $(\valprofile, \signalprofile) \in \valprofilespace \times \symsignalprofilespace$,\footnote{
    Slightly abusing the notation, when the bidders share a common signal space $\signalprofilespace$, we use  $\symsignalprofilespace=\signalprofilespace\times \cdots \times \signalprofilespace$ to denote the signal profile space.}
    $\signalscheme(\signalprofile, \valprofile) = \signalscheme(\permutation(\signalprofile), \permutation(\valprofile))$
    where $\permutation$ acts componentwise: $(\permutation(\valprofile))_i = \val_{\permutation(i)}$ and $(\permutation(\signalprofile))_i = \signal_{\permutation(i)}$.
\end{enumerate}
\end{definition}
This class of information structures greatly simplifies the design problem in a concrete way. Once the common signal space is fixed, identical uniform marginals together with ex ante independence pin down the distribution of signal profiles, and permutation invariance ensures that the information structure treats bidders symmetrically.\footnote{
\wtedit{Private private information with uniform marginals also appears in \citet{KFO-25}, who show that, for the purpose of studying Blackwell-Pareto optimality, one can reparameterize signals so that each bidder's signal is uniformly distributed (up to Blackwell equivalence). 
However, in our setting, they do not by themselves ensure efficiency or implement any particular welfare split. Achieving a desired outcome requires a carefully designed coupling between values and signals.}}
This focus is also natural in our environment, since bidders are ex ante identical and there is no reason for the seller to differentiate among them. Importantly, this parsimonious class remains rich enough for our purposes: as our subsequent construction shows, the seller can still achieve full surplus extraction within it.

With this symmetric private private class in place, we now introduce a structural condition that will serve as a common building block in our main constructions. 
The condition relates the signal ranking to the value ranking in the following way: whenever a bidder receives the highest signal, that bidder must be the highest value bidder in every valuation profile that remains possible under the corresponding posterior. 
This requirement will be imposed in both our full surplus extraction construction and our approximate full surplus extraction construction, as it guarantees efficiency of
the induced allocation and provides a disciplined starting point for controlling bidders' interim
values and surplus. 
We formalize this condition next.
\begin{definition}[{\construcrule} rule]
    \label{defn:construction rule}
    Let $\posteriorBeliefDist: \signalprofilespace \to \Delta(\valprofilespace)$ denote the posterior belief given a realized signal profile. 
    We define the following {\em \construcrule} rule:
    For any realized signal profile $\signalprofile$ and any bidder~$i$, if bidder~$i$ receives the highest signal (i.e., $\signal_i = \mymax(\signalprofile)$), then bidder~$i$ must hold the highest valuation in any state within the posterior support.
    That is,
    \begin{align}
        \label{eq:construction rule}
        \text{if } 
        \signal_i = \mymax(\signalprofile),
        \text{ then }
        \val_i = \mymax(\valprofile), 
        \quad \text{for all } 
        \valprofile \in \supp(\posteriorBeliefDist(\signalprofile)), \signalprofile\in\signalprofilespace, i\in[\biddernum]~.
        \tag{\construcrule}
    \end{align}
\end{definition}
When an information structure satisfies the condition \ref{eq:construction rule}, it implies that when the bidders adopt the pure strategy of bidding their signal, then the seller can guarantee the efficient allocation as the item will be always allocated to the bidder who has the highest value.

Our next requirement in our construction is the following \ref{cond:winner} rule. 
For any realized signal profile $\signalprofile$, let $\expectedVal_i(\signalprofile) \triangleq \expect{\val_i \mid \signalprofile}$ denote bidder~$i$'s expected valuation given the signal profile and $\expectedVal(\signalprofile) = (\expectedVal_i(\signalprofile))_{i\in[\biddernum]}$ denote the corresponding vector of expected valuations. 
\begin{definition}[{\zerorentcondi} rule]
We define the following {\em {\zerorentcondi}} rule:
For any realized signal profile $\signalprofile$ and any bidder~$i$, if bidder~$i$ receives the highest signal (i.e., $\signal_i = \mymax(\signalprofile)$), then bidder~$i$'s posterior expected value must equal the second-highest bid.
That is,
\begin{align}
    \label{cond:winner}
    \tag{\zerorentcondi}
    \text{if }
    \signal_i = \mymax(\signalprofile), 
    \text{ then }
    \expectedVal_i(\signalprofile) = \secmax(\signalprofile)~,
    \quad \text{for all } \signalprofile\in\signalprofilespace, i\in[\biddernum]~.
\end{align}
\end{definition}
In the following \Cref{subsec:an illustrative example}, we present a product-prior example and also intuitions to illustrate how we construct a symmetric private private information that satisfies \ref{eq:construction rule} and \ref{cond:winner}, so that the seller can extract the full surplus. 
We then formalize these intuitions in \Cref{subsec:maximum rev construction} by giving our construction for general prior distributions. 
Finally, in \Cref{subsec:approx maximum rev construction}, we generalize these intuitions to design a private private information that admits a \wtedit{strict} efficient equilibrium to allow the seller to extract almost the full surplus.

\subsection{An Illustrative Example}
\label{subsec:an illustrative example}

\begin{figure}[htbp]
\centering
\resizebox{0.75\textwidth}{!}{


\begin{subfigure}[t]{0.40\linewidth}
\centering

\begin{tikzpicture}[
    x=0.82cm,
    y=0.82cm,
    line join=round,
    panel/.style={draw=black!75,line width=0.9pt},
    diag/.style={draw=black!75,line width=1.0pt}
]

\def\L{5.0}

\def\yAxisX{-0.55}
\def\xAxisY{-0.48}

\fill[gray!15] (0,0)--(0,\L)--(\L,\L)--cycle;
\fill[gray!50] (0,0)--(\L,0)--(\L,\L)--cycle;

\draw[panel] (0,0) rectangle (\L,\L);
\draw[diag] (0,0)--(\L,\L);

\node at (\yAxisX,\L) {$1$};
\node at (\yAxisX,{\L/2}) {$\val_2$};

\node at (\yAxisX,\xAxisY) {$0$};
\node at ({\L/2},\xAxisY-0.1) {$\val_1$};
\node at (\L,\xAxisY) {$1$};

\node at (1.45,4.15) {$\val_2>\val_1$};
\node at (1.35,3.28) {local piece};
\node at (1.65,2.55) {$\pi^{(2)}$};

\node at (3.65,2.05) {$\val_1>\val_2$};
\node at (3.60,1.22) {local piece};
\node at (3.35,0.52) {$\pi^{(1)}$};

\end{tikzpicture}

\caption{Value space}
\label{fig:local-coupling-value}
\end{subfigure}
\hspace{4pt}
\begin{subfigure}[t]{0.40\linewidth}
\centering
\begin{tikzpicture}[
    x=0.82cm,
    y=0.82cm,
    line join=round,
    panel/.style={draw=black!75,line width=0.9pt},
    diag/.style={draw=black!75,line width=1.0pt}
]

\def\L{5.0}

\def\yAxisX{-0.55}
\def\xAxisY{-0.48}

\fill[gray!15] (0,0)--(0,\L)--(\L,\L)--cycle;
\fill[gray!50] (0,0)--(\L,0)--(\L,\L)--cycle;

\draw[panel] (0,0) rectangle (\L,\L);
\draw[diag] (0,0)--(\L,\L);

\node at (\yAxisX,\L) {$1$};
\node at (\yAxisX,{\L/2}) {$\signal_2$};

\node at (\yAxisX,\xAxisY) {$0.5$};
\node at ({\L/2},\xAxisY-0.1) {$\signal_1$};
\node at (\L,\xAxisY) {$1$};

\node at (1.75,4.20) {$\signal_2>\signal_1$};
\node[align=center] at (1.55,3.18)
  {bidder 2 has\\ top signal};

\node at (3.60,2.05) {$\signal_1>\signal_2$};
\node[align=center] at (3.45,1.05)
  {bidder 1 has\\ top signal};

\end{tikzpicture}
\caption{Signal space}
\label{fig:local-coupling-signal}
\end{subfigure}


}
\caption{
Local coupling by rank. Panel (a) partitions the value space according to which bidder has the higher value. Panel (b) gives the corresponding partition of the signal space. The local piece $\pi^{(1)}$ matches the region $\{\val_1>\val_2\}$ with $\{\signal_1>\signal_2\}$, while $\pi^{(2)}$ matches $\{\val_2>\val_1\}$ with $\{\signal_2>\signal_1\}$. Thus, the coupling satisfies \ref{eq:construction rule}.
}
\label{fig:local-coupling}
\end{figure}
\zsedit{Although our primary analysis focuses on a discrete valuation space, 
our results
extend naturally to continuous settings. 
We begin by presenting an illustrative example with continuous valuations. For the remainder of the discussion, however, we retain the discrete assumption for expositional simplicity.}

Consider two bidders whose valuations are drawn independently from a uniform distribution, i.e., $(\val_1, \val_2) \sim \uniform([0,1]^2)$. 
By definition, we know that $\optwelfare = \expect{\val_1 \vee \val_2} = \sfrac{2}{3}$.
Consider the following private private information structure $\inforstructure = ((\signalspace_i)_{i=1,2}, \signalscheme)$, where the signal support is $\signalspace_1 = \signalspace_2 = [1/2, 1]$. The joint probability distribution $\signalscheme(\valprofile, \signalprofile)$ is defined as follows:
\begin{align*}
    \signalscheme(\val_1, \val_2, \signal_1, \signal_2) = 
    \begin{cases} 
        \displaystyle \frac{2}{\val_1(1-\signal_2)} \cdot \indicator{\frac{1}{2} \le \signal_2 \le \signal_1 \le 1, \enskip \signal_2 \le \frac{\val_1+1}{2}} & \text{ if } \val_1 > \val_2~; \\[15pt]
        \displaystyle \frac{2}{\val_2(1-\signal_1)} \cdot \indicator{\frac{1}{2} \le \signal_1 \le \signal_2 \le 1, \enskip \signal_1 \le \frac{\val_2+1}{2}} & \text{ if } \val_2 > \val_1.
    \end{cases}
\end{align*}
The case where $\val_1 = \val_2$ occurs with probability zero and can be defined arbitrarily. 
Next, we verify that $\inforstructure$ is indeed a symmetric private private information structure. 
For any $\signal_1 > \signal_2$, 
we have:
\begin{align*}
    \int_{\valprofile}\signalscheme(\val_1, \val_2, \signal_1, \signal_2) \; \dd \valprofile
    & = \int_{2\signal_2 -1}^{1} \int_{0}^{\val_1} \frac{2}{\val_1(1-\signal_2)} \, \dd \val_2 \, \dd \val_1 
    = 4\\ 
    & = 
    \int_{\signal_1}\int_{\valprofile}\signalscheme(\val_1, \val_2, \signal_1, \signal_2) 
    \; \dd \valprofile \dd \signal_1
    \cdot
    \int_{\signal_2}\int_{\valprofile}\signalscheme(\val_1, \val_2, \signal_1, \signal_2) 
    \; \dd \valprofile \dd \signal_2~.
\end{align*}
By symmetry the same holds when $\signal_2 > \signal_1$. Thus, it is a private private information.
As both $\signal_1, \signal_2$ have uniform marginals, the constructed information structure is also symmetric.

\zsedit{As illustrated in \Cref{fig:local-coupling}, the construction defines the coupling locally on the two rank regions: $\pi^{(1)}$ couples $\{\val_1>\val_2\}$ with $\{\signal_1>\signal_2\}$, while $\pi^{(2)}$ couples $\{\val_2>\val_1\}$ with $\{\signal_2>\signal_1\}$. By construction, the information structure $\inforstructure$ satisfies \ref{eq:construction rule}. Observe that $\signal_1 > \signal_2$ implies $\val_1 > \val_2$; otherwise, $\signalscheme(\val_1, \val_2, \signal_1, \signal_2) = 0$.}
Also $\inforstructure$ satisfies \ref{cond:winner}. Given a signal profile where $\signal_1 > \signal_2$, we obtain:
\begin{align*}
    \expect[\signalscheme]{ \val_1 \mid (\signal_1, \signal_2)}  
    = \int_{2\signal_2 - 1}^{1} \val_1 \int_{0}^{\val_1} \frac{1}{4} \cdot \frac{2}{\val_1(1-\signal_2)} \, \dd \val_2 \, \dd \val_1 = \signal_2~, \quad 
    \expect[\signalscheme]{ \val_2 \mid (\signal_1, \signal_2)} 
    = \frac{\signal_2}{2}~.
\end{align*}
By symmetry, when $\signal_2 > \signal_1$, we obtain $\expect[\signalscheme]{ \val_1 \mid (\signal_1, \signal_2)} = \signal_1 / 2$ and $\expect[\signalscheme]{ \val_2 \mid (\signal_1, \signal_2)} = \signal_1$. Since the event $\signal_1 = \signal_2$ has measure zero, condition \ref{eq:construction rule} is satisfied almost surely.
In the next subsection, we demonstrate that for any information structure satisfying \ref{cond:winner} and \ref{eq:construction rule}, $\bidstrategy_i(\signal_i) = \signal_i$ is an equilibrium. In this equilibrium, bidders have zero surplus. Furthermore, because condition \ref{eq:construction rule} ensures that the bidder with the highest signal also holds the highest valuation, the allocation is efficient. Consequently, the seller extracts the full surplus, achieving $\Rev(\inforstructure, \bidstrategy) = \optwelfare$.

We briefly discuss the intuition derived from this example, which serves as a foundation for the subsequent analysis. 
First, observe that the information structure $\signalscheme$ can be decomposed into two ``local'' probability distributions, $\signalscheme^{(1)}$ and $\signalscheme^{(2)}$, corresponding to the disjoint events $\{\val_1 > \val_2\}$ and $\{\val_2 > \val_1\}$, respectively. 
On each event, the structure is designed so that the bidder with the higher valuation also receives the higher signal and therefore wins.
This piecewise construction is the key to satisfy condition \ref{eq:construction rule}, which guarantees allocation efficiency. 
Consequently, the design of the local distribution $\signalscheme^{(1)}$ 
amounts to specifying a joint distribution supported on the product set where bidder $1$ is the winner in both dimensions:
$\{\valprofile: \val_1 = \mymax(\valprofile)\} \times \{\signalprofile: \signal_1 = \mymax(\signalprofile)\}$ (and analogously for $\signalscheme^{(2)}$).
The central challenge, however, lies in designing these local distributions to ensure that the higher-valuation bidder consistently receives the higher signal, while simultaneously making the bidders' expected surplus to zero. 
To achieve this, we impose the binding constraint \ref{cond:winner}. This condition ensures that for any realized signal profile, the winner's expected valuation coincides exactly with the second-highest signal (the payment). 
Finally, our analysis relies on a fundamental observation: when a symmetric private private information structure satisfies \ref{cond:winner}, once its signal space is fixed, the distribution of bidder $i$'s expected valuation conditional on winning (i.e., receiving the highest signal) will be then determined.

\subsection{Towards Maximum Revenue}
\label{subsec:maximum rev construction}

In this subsection, we prove \Cref{thm:pripri max rev}.
In particular, we first argue that under any information structure satisfying \ref{eq:construction rule} and \ref{cond:winner}, there exists an efficient equilibrium in which bidders' surplus is $0$ (see \Cref{lem:suff condi for full suplus extraction}). 
We then show that there exists a symmetric private private information structure satisfying these two conditions (see \Cref{prop:sym pri-pri rev max}), and this existence result is established via two construction lemmas (\Cref{lem:local-coupling} and \Cref{lem:global-assembly}).

\begin{lemma}
\label{lem:suff condi for full suplus extraction}
\wtedit{Consider an information structure whose signal-profile marginal is a product of atomless marginal distributions. Suppose that \ref{eq:construction rule} and \ref{cond:winner} hold for almost every signal profile with a unique highest signal.}
Then the bidders following the pure strategy of bidding their signals ($\bidstrategy_i(\signal_i) = \signal_i$) is an efficient equilibrium and the bidders' surplus is zero.
\end{lemma}
The proof of the above lemma, as well as the remaining omitted proofs, is deferred to \wtedit{the appendix}.
Intuitively, \ref{cond:winner} binds the winner's expected valuation at every realized signal profile to exactly match the second-price payment.
So the winner earns zero surplus and the seller extracts the full surplus.
As we show in the proof, \ref{eq:construction rule} and  \ref{cond:winner} can jointly imply the following condition: 
\begin{align}
    \label{cond:loser}
    \text{if }
    \signal_i < \mymax(\signalprofile), 
    \text{ then }
    \expectedVal_i(\signalprofile) \le \secmax(\signalprofile)~,
    \quad \text{for all } \signalprofile\in\signalprofilespace, i\in[\biddernum]~.
\end{align}
\wtedit{\ref{cond:winner}, together with Eqn.~\eqref{cond:loser}, rules out profitable upward deviations by any losing bidder: to overtake the current winner, a bidder must bid above $\mymax(\signalprofile)$ and would then pay (at least) $\mymax(\signalprofile)$, while her expected valuation is bounded above by $\secmax(\signalprofile)$.}
Downward deviations are similarly unattractive: 
\wtedit{since the equilibrium payoff is already zero, shading one's bid can only reduce the chance of winning without generating profitable gain.}
Although \Cref{lem:suff condi for full suplus extraction} imposes strong pointwise constraints on the expected valuations of all bidders across the entire signal profile space, we later show that a private private information structure satisfying these conditions does exist and can be explicitly constructed.

With \Cref{lem:suff condi for full suplus extraction}, our design problem reduces to constructing an information structure that simultaneously satisfies the binding equalities in \ref{cond:winner} and \ref{eq:construction rule}. 

\begin{proposition}
\label{prop:sym pri-pri rev max}
There exists a symmetric private private information structure such that it satisfies the conditions \ref{eq:construction rule} and \ref{cond:winner}.
\end{proposition}

\Cref{prop:sym pri-pri rev max}, together with  \Cref{lem:suff condi for full suplus extraction}, completes the proof of \Cref{thm:pripri max rev}.
\begin{proof}[Proof of \Cref{thm:pripri max rev}]
Let the information structure be the one identified in \Cref{prop:sym pri-pri rev max}. 
Then we know that it is symmetric private private, and it satisfies the conditions \ref{eq:construction rule} and \ref{cond:winner}. 
By \Cref{lem:suff condi for full suplus extraction}, the condition in Eqn.~\eqref{cond:loser} also holds.
Thus, under this information structure, truth bidding is an equilibrium, and moreover, under this equilibrium the bidders' surplus is $0$.
We also know that this information structure achieves efficient allocation under the truthful bidding.
Thus, the seller achieves revenue that equals the maximum welfare. 
\end{proof}

In the sequel, we prove \Cref{prop:sym pri-pri rev max} by explicitly constructing a symmetric private private information that satisfies \ref{eq:construction rule} and \ref{cond:winner}.
Our argument is organized around two lemmas. We begin by restricting attention to symmetric private private information structures (see \Cref{def: symmetric private private information}), under which fixing a common signal space $\constructsignalspace$ (will be defined shortly) for each bidder  pins down the distribution of signal profiles.
To satisfy the condition \ref{cond:winner}, we must let the winning bidder's interim value equal the second-highest signal, so we know that fixing the signal space also pins down the distribution of the winner's interim value. 
We then enforce \ref{eq:construction rule}, by conditioning the
prior on the event that the winner indeed has the highest realized valuation, which leads to a modified (tie-weighted) prior. 
We achieve these two steps using the following two lemmas: 
\Cref{lem:local-coupling} chooses a construction of $\constructsignalspace$ and constructs, for each bidder $i$, a local coupling between value profiles and signal profiles on the event that $i$ is selected as the winner, matching the corresponding modified prior and satisfying $\expect{\val_i\mid \signalprofile} = \secmax(\signalprofile)$. 
This would guarantee the desired properties on signal profiles where the winner is uniquely identified by having a strictly highest signal. 
\Cref{lem:global-assembly} assembles these  local couplings into a global joint distribution over values and signals, verifies that the resulting signal marginals are uniform and independent across bidders, and concludes the existence of the desired symmetric private private information structure.

\xhdr{An interval signal space and its winner's interim value distribution}
We consider the following interval signal space, parameterized by $\lowerSignal, \upperSignal \in \wtedit{[0, \mymax_\val \valspace]}$ where $\upperSignal > \lowerSignal$ and derive the resulting distribution of the winner's interim value:
\begin{align*}
    \constructsignalspace
    \; \triangleq \;
    (\lowerSignal\;, \;  \upperSignal)~.
\end{align*}
An important implication of \ref{cond:winner} is that the distribution of the winner's expected valuation is fully determined by the distribution of the second-highest signal if the considered information structure satisfies the condition \ref{cond:winner}.
Indeed, conditional on bidder~$i$ winning the item, 
\ref{cond:winner} implies that $\expectedVal_i(\signalprofile) = \secmax(\signalprofile)$, so the law of $\expectedVal_i(\signalprofile)$  is exactly the law of $\secmax(\signalprofile)$ under the conditional distribution of signal profiles given that $i$ wins.
Thus, once the signal space is fixed, the distribution of $\secmax(\signalprofile)$ conditional on bidder $i$ winning is pinned down, and so is the distribution of the winning bidder's interim value.
By symmetry of both the prior and the information structure, this conditional distribution is the same for all bidders. We denote it by $\winnerposteriorCDF \in \Delta(\constructsignalspace)$ as its CDF. 
Note that since the signal space is continuous, ties occur with probability zero. 
We denote $\signalProfileDist$ as the distribution of signal profiles and $\signalProfileDist_i$ as the marginal distribution.
Then, for each signal $\expectedVal \in \constructsignalspace$,
\zsedit{
\begin{align}
    \label{eq:winner vale dist}
    \winnerposteriorCDF(\expectedVal) 
    \triangleq 
    \prob{\expectedVal_i(\signalprofile) \leq \expectedVal \mid  i \text{ wins}}   
    =
    \frac{\int_{\signalprofile \in \constructsignalspace^\biddernum} \indicator{\mymax(\signalprofile) = \signal_i \text{ and } \secmax(\signalprofile) \leq \expectedVal}\cdot \signalProfileDist(\signalprofile) \, \dd \signalprofile}{\int_{\signalprofile \in \constructsignalspace^\biddernum}\indicator{\mymax(\signalprofile) = \signal_i}\cdot \signalProfileDist(\signalprofile) \, \dd \signalprofile}
    ~.
\end{align}
We denote by $\winnerExpectedvalProb$ the density function of distribution $\winnerposteriorCDF$.
}
Here $\indicator{\cdot}$ is the indicator function. 
Moreover, when the marginal distribution over signal is uniform, we have $\signalProfileDist(\signalprofile) = \signalProfileDist(\signalprofile\primed)$ for all signal profiles $\signalprofile, \signalprofile\primed\in\constructsignalspace^\biddernum$, and thus it
cancels from the numerator and denominator. 
Intuitively, distribution $\winnerExpectedvalProb$ summarizes how often each
signal $\expectedVal$ appears as the winning bidder's expected valuation when the underlying information structure is symmetric private private with the signal space $\constructsignalspace$.

\xhdr{Winner-Dominance and the modified prior}
However, not every candidate signal space is feasible for constructing a symmetric private private information structure that satisfies both \ref{eq:construction rule} and \ref{cond:winner}. 
\ref{eq:construction rule} requires that whenever bidder $i$ wins, she must have the highest realized value.
This leads us to work under the prior conditional on ``$i$ is the (tie-broken) highest-value bidder.'' We write
$\winnerNum(x)$ for the number of coordinates attaining the maximum in a profile $x$ (either a value
profile $\valprofile$ or a signal profile $\signalprofile$). 
Under a uniform tie-breaking rule, we define bidder $i$’s ``modified'' prior $\modifiedpriorCDF_i$ as the
distribution of bidder $i$'s value $\val_i$ conditional on $i$ winning, and its probability mass function is given as follows:
\begin{align}
    \label{eq:modified prior dist}
    \modifiedpriorCDF_i(a)
    \triangleq
    \prob{\val_i = a \mid i \text{ wins}} 
    &= \frac{\sum_{\valprofile} \indicator{\val_i = \mymax(\valprofile) \text{ and } \val_i = a} \cdot \valprior(\valprofile) \cdot \frac{1}{\winnerNum(\valprofile)}} {\sum_{\valprofile} \indicator{\val_i = \mymax(\valprofile)} \cdot \valprior(\valprofile)\cdot\frac{1}{\winnerNum(\valprofile)}}~, \quad
    a\in\valspace~.
\end{align}
Since the original prior is symmetric (exchangeable across bidders), the modified prior is identical across bidders. We nevertheless write it as $\modifiedpriorCDF_i$ as the subsequent construction is carried out from bidder $i$'s perspective.

\xhdr{Winning profiles and tie weighted conditioning}
Fix the interval signal space $\constructsignalspace = (\lowerSignal, \upperSignal)$. For each bidder $i\in[\biddernum]$, we define the following quantities
\begin{align*}
    \valspace^{(i)} 
    & \triangleq 
    \left\{\valprofile\in \valspace^\biddernum: \val_i=\max(\valprofile)\right\},
    \\
    \constructsignalspace^{(i)} 
    & \triangleq 
    \left\{\signalprofile\in \constructsignalspace^\biddernum: \signal_i=\max(\signalprofile)\right\}~;
    \;\;
    \strictsignalspace^{(i)} 
    \triangleq \left\{\signalprofile\in \constructsignalspace^{\biddernum}: \signal_i>\max_{j\ne i}\signal_j\right\}~.
\end{align*}
By definition, the set $\strictsignalspace^{(i)}$ contains all signal profiles where bidder $i$ has the unique highest signal, whereas
$\constructsignalspace^{(i)}\setminus \strictsignalspace^{(i)}$ consists of signal profiles where bidder $i$ receives the highest signal and another bidder $j\neq i$ also receives the same highest signal.  Since the signal space is continuous, the set of ties $\constructsignalspace^{(i)} \setminus \strictsignalspace^{(i)}$ has measure zero. 
We thus restrict our attention to $\strictsignalspace^{(i)}$.

The role $\valspace^{(i)}$ in our construction is the following. 
For $\signalprofile\in\strictsignalspace^{(i)}$, bidder $i$ is the
\emph{only} bidder with the highest signal, so \ref{eq:construction rule} requires that bidder $i$
has the highest realized value in her posterior support. 
Hence, restricting value profiles to $\valspace^{(i)}$ is already sufficient to guarantee \ref{eq:construction rule} on $\strictsignalspace^{(i)}$. 

To better describe our construction, we introduce the following definition. 
Under a uniform tie breaking rule, we define the tie weighted distribution
$\modifiedpriorCDF^{(i)}\in\Delta(\valspace^{(i)})$ by
\begin{equation}\label{eq:hat-vartheta-profile}
    \modifiedpriorCDF^{(i)}(\valprofile)\ 
    \triangleq\
    \frac{1}{\winnerNum(\valprofile)}
    \cdot \frac{\valprior(\valprofile)}{\sum\nolimits_{\valprofile'\in \valspace^{(i)}} \valprior(\valprofile')\cdot \frac{1}{\winnerNum(\valprofile')}}\,,
    \quad \valprofile\in \valspace^{(i)}~.
\end{equation}
By definition, distribution $\modifiedpriorCDF_i$ that we defined in Eqn.~\eqref{eq:modified prior dist} is the marginal distribution of $\val_i$ in $\modifiedpriorCDF^{(i)}$, and we next define its mean value $\winnerMean \triangleq \expect[\modifiedpriorCDF_i]{\val_i}$.
Similarly, under uniform signals on $\constructsignalspace^\biddernum$, we define the distribution
$\signalProfileDist^{(i)}\in\Delta(\constructsignalspace^{(i)})$ as the probability measure on $\constructsignalspace^{(i)}$.
With these definitions, we now present our two construction lemmas: the first establishes a local coupling that designates bidder $i$ as the winner, and the second builds the global construction by combining these local couplings.
\begin{lemma}[Local coupling for a designated winner]
\label{lem:local-coupling}
There exists a choice of parameters $(\lowerSignal, \upperSignal)$ such that for every
$i\in[\biddernum]$ there exists a joint distribution $\signalscheme^{(i)}\in\Delta(\valspace^{(i)}\times \constructsignalspace^{(i)})$
satisfying: (1) the marginal of $\signalscheme^{(i)}$ on $\valspace^{(i)}$ is $\modifiedpriorCDF^{(i)}$; 
(2) the marginal of $\signalscheme^{(i)}$ on $\constructsignalspace^{(i)}$ is $\signalProfileDist^{(i)}$;
(3) for every $\signalprofile\in \constructsignalspace^{(i)}$, 
$\expect[\signalscheme^{(i)}]{\val_i \mid \signalprofile} = \secmax(\signalprofile)$.
In particular, for all $\signalprofile \in \strictsignalspace^{(i)}$, \ref{cond:winner} and \ref{eq:construction rule} hold.
\end{lemma}

\begin{lemma}[Global construction from local couplings]
\label{lem:global-assembly}
Suppose for each bidder $i\in[\biddernum]$, we have a distribution
$\constructsignalscheme^{(i)}\in\Delta(\valspace^{(i)}\times \constructsignalspace^{(i)})$ satisfying the three properties in
\Cref{lem:local-coupling}. Define a joint distribution $\constructsignalscheme\in\Delta(\valspace^\biddernum\times \constructsignalspace^\biddernum)$ by
\begin{equation}\label{eq:global-mixture}
    \constructsignalscheme(\valprofile,\signalprofile)\ \triangleq\ \frac{1}{\biddernum}\sum\nolimits_{i\in[\biddernum]} \constructsignalscheme^{(i)}(\valprofile,\signalprofile)~.
\end{equation}
Then $\constructsignalscheme$ defines a symmetric private private information structure $\hat{\inforstructure} = (\constructsignalspace^\biddernum, \constructsignalscheme)$.
Moreover, under $\constructsignalscheme$, \ref{cond:winner} and \ref{eq:construction rule} hold for
all signal profiles $\signalprofile \in \cup_{i \in [\biddernum]}\strictsignalspace^{(i)}$.
\end{lemma}
\wtresedit{Since, by construction, the product signal distribution is atomless, the set of signal profiles with ties has measure zero. Thus, \Cref{lem:global-assembly} immediately implies \Cref{prop:sym pri-pri rev max}, and we defer the formal verification to \Cref{apx:proofs in pri pri rev}.}


\subsection{Towards Approximate Maximum Revenue with Strict BNE}
\label{subsec:approx maximum rev construction}

To show the existence of a private private information that admits a \zsedit{strict} efficient equilibrium, we modify the \ref{cond:winner} condition as a strengthened set of sufficient conditions: whenever an information structure satisfies these conditions, it admits truthful bidding i.e., bidding one's signal $\bidstrategy_i(\signal_i) = \signal_i$, as a \zsedit{strict} equilibrium. 
Moreover, these conditions guarantee that upon winning, a bidder’s expected payment is within $\varepsilon$ of her expected valuation, so her expected surplus is strictly positive while remaining uniformly bounded by $\varepsilon$, i.e., $0 < \biddersurplus(\inforstructure, \bidstrategyprofile) \le \varepsilon$.


\begin{lemma}
\label{lem:suff condi for approx full suplus extraction}
Fix an arbitrarily small $\varepsilon > 0$. Consider a symmetric private private information structure $\inforstructure$ that satisfies the following conditions:
\begin{enumerate}
    \item 
    \em{\textbf{Signal space:}} 
    Every bidder shares a common signal space $\constructsignalspace = (\lowerSignal, \upperSignal)$ with sufficiently small length, i.e.,  $0 < \upperSignal - \lowerSignal \le \varepsilon$. 
    \item 
    \em{\textbf{Expected valuation:}} 
    For any signal profile $\signalprofile$ \zsedit{with a unique highest signal}, bidder $i$'s expected valuation 
    $\expectedVal_i(\signalprofile)$ satisfies:
    \begin{enumerate}
        \item 
        \label{it:highest-signal}
        If bidder $i$ holds the highest signal, that is $\signal_i = \mymax(\signalprofile)$, her expected valuation is given by:
        \begin{align*}
        \expectedVal_i(\signalprofile) = 
        \begin{cases} 
        \secmax(\signalprofile) - \displaystyle\frac{(\secmax(\signalprofile) -  \lowerSignal) \cdot \eps}{2\biddernum} & \text{if } \lowerSignal< \secmax(\signalprofile) < \displaystyle\frac{\signal_i + \lowerSignal}{2}, \\[10pt]
        \secmax(\signalprofile) + 
        \displaystyle \frac{(\signal_i - \secmax(\signalprofile))\cdot\eps}{\biddernum} & \text{if } 
        \displaystyle \frac{\signal_i + \lowerSignal}{2} \leq \secmax(\signalprofile) < \upperSignal.
        \end{cases}
        \end{align*}
        \item 
        \label{it:not-highest-signal}
        If bidder $i$ does not hold the highest signal, that is, $\signal_i < \mymax(\signalprofile)$, then for any $j$ such that $\signal_j = \mymax(\signalprofile)$, it holds that: $\expectedVal_i(\signalprofile) \leq \expectedVal_j(\signalprofile)$.
    \end{enumerate}
\end{enumerate}
Under these conditions, 
bidders following the pure strategy of bidding their signal ($\bidstrategy_i(\signal_i) = \signal_i$) is \zsedit{a strict} equilibrium.
Moreover, the aggregate bidders' surplus satisfies $0 < \biddersurplus(\inforstructure, \bidstrategyprofile) \leq \varepsilon$.
\end{lemma}
We next give a brief intuition for how \Cref{lem:suff condi for approx full suplus extraction} helps to ensure the truthful bidding as a \zsedit{strict}  equilibrium. \zsedit{Consider a symmetric private private information satisfying the conditions in \Cref{lem:suff condi for approx full suplus extraction}.} 
Let us fix a bidder $i$ with her signal $\signal_i$, and let $y \triangleq \max_{k\neq i} \signal_k$ denote the highest competing signal, i.e., the price bidder $i$ pays if she wins under truthful bidding by all other bidders. \zsedit{Since the signal space is continuous, ties occur with zero probability. Thus, we focus only on signal profiles with a unique highest signal. For all signal profiles $\signalprofile$ where bidder $i$ has the unique highest signal $\signal_i$, under the construction in \Cref{lem:suff condi for approx full suplus extraction}, bidder $i$'s interim expected valuation depends on $\signal_i$ and $y$. With a slight abuse of notation, we use $\expectedVal_{i}(\signal_i,y)$ to denote bidder $i$'s expected value conditional on winning with a payment of $y$. Note that $\expectedVal_{i}(\signal_i,y)$ is a piecewise-linear function of $y$.}
\Cref{fig:val structure} plots this function (solid black curve)
together with the break-even line $x=y$ (dashed diagonal). The vertical distance between the
two curves equals bidder $i$'s ex-post net utility, $\expostbiddersurplus(\signal_i,y) \triangleq  \expectedVal_i(\signal_i,y)-y$.

The key feature in this plot is the ``kink'' at $\middlepoint(\signal_i) \triangleq \frac{\signal_i+\lowerSignal}{2}$
which, together with $\signal_i$, partitions the payment axis into three regions. 
In \emph{Region 1} ($y<\middlepoint(\signal_i)$), the
solid curve lies strictly below the break-even line, so winning at such a price yields a strict loss:
$\expostbiddersurplus(\signal_i,y)<0$. 
In \emph{Region 2} ($\middlepoint(\signal_i)\le y<\signal_i$), the solid curve lies strictly above the break-even line, so winning yields a strictly positive but small gain, $\expostbiddersurplus(\signal_i,y)>0$,
of magnitude $O(\varepsilon)$ (by our construction). 
Finally, \emph{Region 3} ($y>\signal_i$) corresponds to prices that bidder $i$ can face only if she overbids above her signal to overtake a higher-signal competing bidder; in this
region, condition \eqref{it:not-highest-signal} in \Cref{lem:suff condi for approx full suplus extraction} implies that her expected valuation is bounded by that of the true
highest-signal bidder, which lies below the break-even line, so winning would also generate a strict loss.

\begin{figure}
    \centering
    \resizebox{0.40\textwidth}{!}{%
        \begin{tikzpicture}[scale=1.5, >=Stealth]
    \tikzset{
        axis/.style={->, line width=1.2pt, >=Stealth},
        tick/.style={line width=1.2pt},
        proj/.style={dotted, line width=0.8pt, gray},
        refdiag/.style={dashed, line width=1.2pt, gray},
        midguide/.style={dotted, line width=1.2pt, black},
        boundary/.style={line width=1.2pt, black},
        dot/.style={circle, fill=black, inner sep=0pt, minimum size=5pt}, 
        axlabel/.style={font=\LARGE\bfseries}, 
        ticklabel/.style={font=\LARGE},             
        ptlabel/.style={font=\LARGE\bfseries},      
    }

    \def\lowS{1.0}    
    \def\highS{7.5}   
    \def\midS{4.25}   
    \def\overS{9.5}   
    \def\lossGap{1.8}
    \def\gainGap{2.8}
    \def\extY{8.5}
    
    \def\AxOff{-0.5}

    \draw[axis] (\AxOff, \AxOff) -- (9.8, \AxOff);
        
    \draw[axis] (\AxOff, \AxOff) -- (\AxOff, 9.2);

    \node[axlabel] at (5, \AxOff - 1.3) {Payment};

    \node[axlabel, rotate=90] at (\AxOff - 2.0, 5) {Bidder $i$'s Expected Value};

    \draw[refdiag] (\AxOff, \AxOff) -- (9.2, 9.2);
    
    \draw[midguide] (\midS, \AxOff) -- (\midS, 9.2);


    \draw[tick] (\lowS, \AxOff) -- (\lowS, \AxOff-0.15)
        node[below, ticklabel] {$\underline{s}$};
    \draw[proj] (\lowS, \AxOff) -- (\lowS, \lowS); 

    \draw[tick] (\midS, \AxOff) -- (\midS, \AxOff-0.15)
        node[below, ticklabel] {$\kappa(s_i)$};
        
    \draw[tick] (\highS, \AxOff) -- (\highS, \AxOff-0.15)
        node[below, ticklabel] {$s_i$};
    \draw[proj] (\highS, \AxOff) -- (\highS, \highS);

    \draw[tick] (\AxOff, \lowS) -- (\AxOff-0.15, \lowS)
        node[left, ticklabel] {$\underline{s}$};
    \draw[proj] (\AxOff, \lowS) -- (\lowS, \lowS);

    \draw[tick] (\AxOff, \midS) -- (\AxOff-0.15, \midS)
        node[left, ticklabel] {$\kappa(s_i)$};
    \draw[proj] (\AxOff, \midS) -- (\midS, \midS);

    \draw[tick] (\AxOff, \highS) -- (\AxOff-0.15, \highS)
        node[left, ticklabel] {$s_i$};
    \draw[proj] (\AxOff, \highS) -- (\highS, \highS);

    \fill[gray!20] (\lowS, \lowS) -- (\midS, \midS - \lossGap) -- (\midS, \midS) -- cycle;
    \draw[boundary] (\lowS, \lowS) -- (\midS, \midS - \lossGap);

    \fill[gray!60] (\midS, \midS) -- (\midS, \midS + \gainGap) -- (\highS, \highS) -- cycle;
    \draw[boundary] (\midS, \midS + \gainGap) -- (\highS, \highS);

    \fill[gray!20] (\highS, \highS) -- (\overS, \extY) -- (\overS, \overS) -- cycle;
    \draw[boundary] (\highS, \highS) -- (\overS, \extY);

    \node[ptlabel, rotate=45] at (3.2, 2.6) {Region 1};
    \node[ptlabel, rotate=45] at (5.2, 6.2) {Region 2};
    \node[ptlabel, rotate=35] at (9.0, 8.6) {Region 3};

    \node[dot] at (\midS, \midS) {};
    \node[dot] at (\midS, \midS - \lossGap) {};
    \node[dot] at (\midS, \midS + \gainGap) {};
    \node[dot] at (\highS, \highS) {};
    \node[dot] at (\lowS, \lowS) {};
\end{tikzpicture}
    }
    \caption{Geometric structure of incentives in \Cref{lem:suff condi for approx full suplus extraction}. The horizontal axis represents the payment price $y$, and the vertical axis denotes bidder~$i$'s interim expected valuation. The solid black lines depict bidder~$i$'s expected value given signal $s_i$, conditional on a price $y$. The dashed diagonal line represents the break-even point where bidder~$i$'s expected value equals payment. }
    \label{fig:val structure}
    \vspace{-10pt}
\end{figure}

This geometry makes truthful bidding the unique best response when all other bidders bid truthfully. 
If bidder $i$ underbids to some $\middlepoint(\signal_i) \leq \bid <\signal_i$, she can only shrink the set of $y$ for which she wins: she gives us some part of Region~2, and thus her expected payoff strictly decreases. Bidding $\signal_i$ yields a strictly positive expected payoff, which strictly exceeds the payoff from deviating to any $\bid < \middlepoint(\signal_i)$.
If instead she overbids to $b>\signal_i$, she expands the winning set into Region~3, where the winning is strictly unprofitable, so her expected payoff again strictly decreases. 
Since the signal distribution is continuous under the information structure considered in \Cref{lem:suff condi for approx full suplus extraction} (ties occur with
probability zero), these losses occur with positive probability whenever $\bid\neq \signal_i$, yielding
strict suboptimality of any deviation and hence uniqueness of truthful bidding.
At the same time, the positive wedge in Region~2 can be arbitrarily small by choosing $\upperSignal-\lowerSignal\le \varepsilon$, which ensures that bidders' surplus is strictly positive but uniformly bounded by~$\varepsilon$.\footnote{\wtedit{The proofs of \Cref{lem:suff condi for approx full suplus extraction} and \Cref{lem: existence of infor structure with unique equilibrium} are provided in the online appendix.}}

\begin{proposition}
\label{lem: existence of infor structure with unique equilibrium}
    There exists a symmetric private private information $\inforstructureUnique$ satisfying the conditions of \Cref{lem:suff condi for approx full suplus extraction} such that \ref{eq:construction rule} holds for all signal profiles with a unique highest signal.
\end{proposition}

\begin{proof}[Proof of \Cref{prop:pripri approx max rev}]
Let $\inforstructure = (\constructsignalspace^\biddernum, \signalscheme)$ be the symmetric private private information identified in \Cref{lem: existence of infor structure with unique equilibrium}.
Then we know it satisfies conditions in \Cref{lem:suff condi for approx full suplus extraction} and \ref{eq:construction rule} for all $\signalprofile \in \strictsignalspace$. Since the signal space is continuous, the set of ties $\constructsignalspace^{(i)} \setminus \strictsignalspace^{(i)}$ has measure zero. Consequently, conditions in \Cref{lem:suff condi for approx full suplus extraction} and \ref{eq:construction rule} hold with probability one. This implies that the constructed information structure achieves efficiency and extracts almost full surplus 
and it admits a \zsedit{strict} equilibrium. 
\end{proof}

\section{Maximum Bidder Surplus}
\label{sec:pri pri bidders surplus}

In this section, we study the theoretical limits of bidder surplus in the context of private private information. 
We begin by establishing a negative result: in second-price auctions, no private private information can admit an efficient equilibrium in which bidders' surplus equals the maximum welfare. 
We then characterize exactly when the gap can be made arbitrarily small. 
Specifically, we provide necessary and sufficient conditions on the prior value distribution under which there exists a private private information and an efficient equilibrium such that bidders' surplus is arbitrarily close to maximal welfare. 
These conditions hold only in special cases: for every valuation profile $\valprofile\in\supp(\valprior)$, either all bidders share a common value or there exists at most one bidder that has a positive value.

\begin{theorem}
\label{thm:pripri max bidder surplus}
There exists {\em no} private private information structure $\inforstructure$ that admits an efficient equilibrium 
such that the resulting bidders' surplus equals the maximal welfare, as long as there exists a value profile $\valprofile\in\supp(\valprior)$ 
where there are two bidders have different values. 
\end{theorem}
The key tension behind \Cref{thm:pripri max bidder surplus} is that making bidders capture the entire efficient welfare forces the seller's revenue to be exactly zero. 
In a second-price auction, zero revenue means the second-highest bid must be zero almost surely in the equilibrium. 
Under a private private information structure, however, bidders' signals are independent ex ante, 
so any equilibrium in which two or more bidders sometimes submit a strictly positive bid necessarily generates a positive probability that multiple bidders bid positively at the same time.
This immediately implies a strictly positive second highest bid and hence positive revenue, contradicting zero revenue.
The only way to avoid this ``overlap'' under independence is to ensure that at most one bidder can ever bid positively, or that all bidders always bid zero. 
But either way breaks efficiency whenever valuations are not purely common values, 
as efficiency requires that different bidders win in different states, which is incompatible with suppressing competition in this way.

\Cref{thm:pripri max bidder surplus iff} complements \Cref{thm:pripri max bidder surplus} by identifying exactly when this tension disappears, so that bidder surplus can be arbitrarily close to maximal welfare under private private information.


\begin{proposition}
\label{thm:pripri max bidder surplus iff}
For any $\eps > 0$, there exists a private private information $\inforstructure$ satisfying the following conditions simultaneously: 
(i) it admits an efficient equilibrium;
(ii) bidders' surplus satisfies $\biddersurplus(\inforstructure, \bidstrategyprofile) \ge \optwelfare-\varepsilon$;
if and only if the prior $\valprior$ satisfies: for every $\valprofile \in \supp(\valprior)$, either bidders share a common value (i.e., $\val_i = \val_j$ for all $i,j$), or there exists some bidder $i \in [\biddernum]$ such that $\val_i > 0, \val_{-i} \equiv 0$.
\end{proposition}
To push bidders' surplus close to maximal welfare requires the seller's expected revenue to be arbitrarily close to zero, which means that competitive bidding must occur with vanishing probability. 
With independent signals, however, any scheme that reliably selects the highest value bidder typically creates overlap events in which more than one bidder submits a positive bid, which pushes the second highest bid away from zero.

The exceptional priors in \Cref{thm:pripri max bidder surplus iff} are those that avoid meaningful competition in the prior support.
First, if all bidders always have the same value,  then efficiency does not require distinguishing among bidders and competition is irrelevant.
Second, if at most one bidder has a positive value in the realized value profile, an efficient allocation can be implemented while keeping all other bidders' bids at zero and the second highest bid near zero.
In all other priors, the presence of at least two bidders with distinct positive values on an event of positive probability forces overlap to occur under independent signaling, which makes revenue bounded away from zero and prevents bidders' surplus from approaching maximal welfare.

\section{Private Private Efficient Frontier}
\label{sec:general}
\label{subsec:pri pri efficiency}
\newcommand{\revselleropt}{\Rev_{\textsc{S}}}
\newcommand{\BSselleropt}{\biddersurplus_{\textsc{S}}}
\newcommand{\revfull}{\Rev_{\textsc{F}}}
\newcommand{\BSfull}{\biddersurplus_{\textsc{F}}}

In \Cref{sec:pri pri rev}, we showed that there exists a private private information that admits an efficient equilibrium in which the seller extracts full surplus.
In \Cref{sec:pri pri bidders surplus}, we also proved that, under an efficient equilibrium, no private private information can make bidders' surplus equal to maximal welfare unless the conditions in \Cref{thm:pripri max bidder surplus iff} hold.
One natural next question is: 
what other efficient outcomes are achievable under private private information?
To describe our results, we denote by $\worstwelfare \triangleq \expect[\valprofile\sim\valprior]{\min\nolimits_{i\in[\biddernum]} \val_i}$ the maximally inefficient welfare the minimum welfare achievable under any information structures.

We first revisit the prior conditions identified in \Cref{thm:pripri max bidder surplus iff} under which bidders' maximal surplus can approach maximal welfare arbitrarily closely, and we ask whether, under these priors, every efficient outcome can be achieved by some private private information structure.
We then study the efficient frontier under private private information in another well studied setting, the independent private values setting (i.e., the prior $\valprior$ is a product measure).\footnote{\wtresedit{In online appendix, we fully characterize the set of welfare outcomes that can be implemented by general information structures. In particular, we show that  some simple general information structures exist to achieve the extreme welfare outcomes presented in \Cref{fig:feasible_region}.}}

Because the set of private private information is nonconvex, one cannot simply mix two private private information to obtain a welfare outcome that is a convex combination of the two. 
Indeed, such mixing typically induces correlation across bidders' signals, violating the mutual independence requirement in private private information.
As a result, it is not a priori clear whether the set of achievable efficient outcomes is path connected.
Nonetheless, we show that a full path-connected segment of the efficient frontier is attainable,
both for priors that satisfy the conditions in \Cref{thm:pripri max bidder surplus iff}  and for the standard independent private values setting (see \Cref{fig:feasible_region} for an illustration).

We first show that when the prior satisfies the conditions identified in \Cref{thm:pripri max bidder surplus iff} (including the pure-common-value prior distributions), the private private information can achieve any efficient welfare outcome.

\begin{figure}[t]
    \centering
    \resizebox{0.45\textwidth}{!}{
        \begin{tikzpicture}[scale=1.2, >=Stealth]
    \tikzset{
        axis/.style={->, line width=1.2pt, >=Stealth}, 
        tick/.style={line width=1.2pt},
        proj/.style={dotted, line width=0.8pt, gray},
        frontier/.style={line width=1.2pt, dashed},
        dot/.style={circle, fill=black, inner sep=0pt, minimum size=5pt}, 
        axlabel/.style={font=\normalsize\bfseries}, 
        ticklabel/.style={font=\normalsize},             
        ptlabel/.style={font=\normalsize\bfseries},      
    }

    \def\WmaxVal{6}
    \def\WminVal{2}
    \def\AxOff{-0.5}

    \coordinate (A) at (\WmaxVal, 0);
    \coordinate (B) at (0, \WmaxVal);
    \coordinate (C) at (0, \WminVal);
    \coordinate (D) at (\WminVal, 0);

    \coordinate (E) at (3.5, 2.5);

    \coordinate (F) at (A);

    \fill[gray!20] (D) -- (A) -- (B) -- (C) -- cycle;

    \draw[axis] (\AxOff, \AxOff) -- (7, \AxOff);
    \draw[axis] (\AxOff, \AxOff) -- (\AxOff, 7);

    \node[axlabel] at (3.8, \AxOff - 0.67) {Bidder surplus};

    \node[axlabel, rotate=90] at (\AxOff - 0.7, 3.5) {Revenue};

    \begin{scope}[proj]
        \draw (A) -- (\WmaxVal, \AxOff);
        \draw (D) -- (\WminVal, \AxOff);
        \draw (B) -- (\AxOff, \WmaxVal);
        \draw (C) -- (\AxOff, \WminVal);
        \draw (A) -- (\AxOff, 0);
        \draw (B) -- (0, \AxOff);
    \end{scope}

    \draw[frontier, color=efficiencyColor] (B) -- (A);
    \draw[frontier, color=inefficiencyColor] (C) -- (D);
    \draw[frontier, color=efficiencyColor] (B) -- (C);
    \draw[frontier, color=efficiencyColor] (D) -- (A);

    
    \draw[blue, line width=5pt, opacity=0.4] (B) -- (F);

    \draw[red, line width=1.5pt] (B) -- (E);

    
    \node[dot] at (A) {}; 
    \node[ptlabel, above right] at (A) {A(\textcolor{blue}{F})};

    \node[dot] at (B) {}; \node[ptlabel, above right] at (B) {B};
    \node[dot] at (C) {}; \node[ptlabel, left]        at (C) {C};
    \node[dot] at (D) {}; \node[ptlabel, below]       at (D) {D};

    \node[dot, red] at (E) {}; 
    \node[ptlabel, red, above right] at (E) {E};

    \draw[tick] (0, \AxOff) -- (0, \AxOff-0.15) node[below, ticklabel] {0}; 
    \draw[tick] (\WminVal, \AxOff) -- (\WminVal, \AxOff-0.15) node[below, ticklabel] {$\underline{\textsc{Wel}}$};
    \draw[tick] (\WmaxVal, \AxOff) -- (\WmaxVal, \AxOff-0.15) node[below, ticklabel] {$\overline{\textsc{Wel}}$};

    \draw[tick] (\AxOff, 0) -- (\AxOff-0.15, 0) node[left, ticklabel] {0};
    \draw[tick] (\AxOff, \WminVal) -- (\AxOff-0.15, \WminVal) node[left, ticklabel] {$\underline{\textsc{Wel}}$};
    \draw[tick] (\AxOff, \WmaxVal) -- (\AxOff-0.15, \WmaxVal) node[left, ticklabel] {$\overline{\textsc{Wel}}$};

\end{tikzpicture}
    }
    \caption{
    The feasible surplus region under general information structures (see \wtedit{online appendix}) is a trapezoid with vertices $A: (\optwelfare, 0)$, $B: (0, \optwelfare)$, $C: (0, \worstwelfare)$, and $D: (\worstwelfare, 0)$.
    Under the condition in \Cref{thm:pripri max bidder surplus iff}, 
    \Cref{thm:iff pri pri efficient} characterizes an efficient frontier segment $BF$ with $F$ arbitrarily close to $A$.
    Under IPV,
    \Cref{thm:ipv pri pri efficient} characterizes an efficient frontier segment $BE$ connecting $B$ to the full information outcome $E$.}
    \label{fig:feasible_region}
\end{figure}

\begin{theorem}
\label{thm:iff pri pri efficient}
Suppose the prior distribution $\valprior$ satisfies the necessary and sufficient conditions in \Cref{thm:pripri max bidder surplus iff}. Then given any $\varepsilon > 0$, for any pair $(B, R)\in\reals_+^2$ such that $B + R = \optwelfare$ and $B \le \optwelfare - \varepsilon$, there exists a private private information $\inforstructure$ and an efficient equilibrium $\bidstrategyprofile$ such that $\Rev(\inforstructure, \bidstrategyprofile) = R$ and $\biddersurplus(\inforstructure, \bidstrategyprofile) = B$.
\end{theorem}

Another well studied setting in the literature is the independent private values (IPV) model. In this setting, revealing full information is itself a private private information.
Because bidders' values are independent, the truthful bidding is an efficient equilibrium. 
Let $(\BSfull, \revfull)$ be the welfare outcome under full information revelation, then $\BSfull = \expect[\valprofile\sim\valprior]{\mymax(\valprofile) - \secmax(\valprofile)}$. 
We show that every efficient welfare outcome between the seller optimal outcome and the outcome induced by full information can be implemented by some private private information structure.

\begin{theorem}
\label{thm:ipv pri pri efficient}
For the IPV setting (i.e., the prior $\valprior$ is a product measure), given any pair $(B, R)\in\reals_+^2$ where $B+R = \optwelfare$ and $B \le \BSfull$, there exists a private private information $\inforstructure$ and an equilibrium $\bidstrategyprofile$ such that $\Rev(\inforstructure, \bidstrategyprofile) = R$ and $\biddersurplus(\inforstructure, \bidstrategyprofile) = B$.
\end{theorem}
\begin{figure}[t]
    \centering
    \resizebox{0.40\textwidth}{!}{%

\begin{tikzpicture}[scale=1.2, >=Stealth]
    \tikzset{
        axis/.style={->, line width=2.2pt, >=Stealth},
        tick/.style={line width=2.2pt},
        dashedline/.style={dashed, line width=2.2pt, black}, 
        boundary/.style={line width=2.2pt, black},
        axlabel/.style={font=\Huge},
        ticklabel/.style={font=\huge},
        ptlabel/.style={font=\large, align=center},
    }

    \def\Vmax{9}   
    \def\Tval{5}

    \coordinate (O) at (0,0);
    \coordinate (V1Max) at (\Vmax, 0);
    \coordinate (V2Max) at (0, \Vmax);
    \coordinate (TopRight) at (\Vmax, \Vmax);
    
    \coordinate (Center) at (\Tval, \Tval);

    
    \fill[gray!10] (O) rectangle (\Tval, \Tval);

    \fill[gray!25] (\Tval, 0) rectangle (\Vmax, \Tval);

    \fill[gray!25] (0, \Tval) rectangle (\Tval, \Vmax);

    \fill[gray!50] (\Tval, \Tval) rectangle (\Vmax, \Vmax);

    \draw[axis] (0, 0) -- ({\Vmax + 0.8}, 0) 
        node[right, axlabel] {$v_1$};
        
    \draw[axis] (0, 0) -- (0, {\Vmax + 0.8}) 
        node[above, axlabel] {$v_2$};

    \draw[dashedline] (\Tval, 0) -- (\Tval, \Vmax);
    
    \draw[dashedline] (0, \Tval) -- (\Vmax, \Tval);

    \draw[boundary] (O) rectangle (\Vmax, \Vmax);

    
    \draw[tick] (0, 0) -- (0, -0.15) 
        node[below, ticklabel] {0};
        
    \draw[tick] (\Tval, 0) -- (\Tval, -0.15)
        node[below, ticklabel] {$t$};
        
    \draw[tick] (\Vmax, 0) -- (\Vmax, -0.15)
        node[below, ticklabel] {$\bar{v}$};

    \draw[tick] (0, 0) -- (-0.15, 0) 
        node[left, ticklabel] {0};

    \draw[tick] (0, \Tval) -- (-0.15, \Tval)
        node[left, ticklabel] {$t$};
        
    \draw[tick] (0, \Vmax) -- (-0.15, \Vmax)
        node[left, ticklabel] {$\bar{v}$};

    
    \node[ptlabel] at ({0.5*\Tval}, {0.5*\Tval}) {Low-Low Region:\\Rev-Optimal Coupling\\ $x_i(s) = \textsc{secmax}(s)$ \\ (Zero Surplus)};
    
    \node[ptlabel] at ({0.5*(\Tval+\Vmax)}, {0.5*\Tval}) {High-Low Region:\\
    $s_1  = v_1$; $s_2 \sim \mathrm{Unif}(S)$\\(High Rent)};
    
    \node[ptlabel] at ({0.5*\Tval}, {0.5*(\Tval+\Vmax)}) {Low-High Region:\\
    $s_1 \sim \mathrm{Unif}(S)$; $s_2 = v_2$\\(High Rent)};
    
    \node[ptlabel] at ({0.5*(\Tval+\Vmax)}, {0.5*(\Tval+\Vmax)}) {High-High Region:\\
Full Info Coupling\\$s_i = v_i$\\(High Surplus)};

\end{tikzpicture}
}
    \caption{Characterizing a continuum of efficient outcomes via a threshold based hybrid information structure (2 bidders). 
    For valuations below the threshold $t$ (Low-Low Region), the mechanism defaults to the seller-optimal private private structure with zero bidder surplus. Conversely, in the High-High Region ($v > t$), it converges to full information coupling. The intermediate regions (Low-High/High-Low) provide partial information rents by revealing the true value to only one bidder.}
    \label{fig:illustration of mix bidder optimal and seller optimal}
    \vspace{-10pt}
\end{figure}
As implied by \Cref{thm:ipv pri pri efficient}, 
for the IPV setting, we are able to characterize the continuum of efficient outcomes connecting the seller optimal outcome, denoted by $(\BSselleropt, \revselleropt)$, and the full information outcome$(\BSfull, \revfull)$.
To do so, we construct a family of private private information structures indexed by a threshold $t$  that partitions the valuation space into a ``low region'' and a ``high region.'' 
The key feature is that the signaling rule depends on the realized value profile.
If all bidders' valuations lie in the low region, we use the full-surplus extraction information structure constructed in  \Cref{subsec:maximum rev construction}. 

In contrast, when the realized profile contains at least one bidder in the high region, we reveal the values of those high region bidders so the efficient winner is identified,
and the main subtlety is how to define signals for the remaining bidders whose valuations lie in the low region.
We cannot simply apply the low region full-surplus extraction information structure conditional on ``everyone is low,'' because the signal generation rule in that structure is defined exclusively for valuation profiles where all components are low. Consequently, a profile containing a high-value bidder falls outside the domain of this specific information structure.
Instead, for each low region bidder we add an independent private randomization and remap their low value into a signal that has the same target marginal distribution as in the low only construction. 
This remapping is feasibility driven: it restores product structure across bidders' signals even though the overall signaling rule depends on the full profile. 
At the same time, it preserves incentives, because these low region signals always stay below the threshold and satisfy the same local indifference properties as before, so a low bidder cannot profit by mimicking a high bid, and truthful bidding remains an equilibrium while efficiency is maintained.
By continuously varying the threshold $t$, we obtain a continuous path of the efficient welfare outcomes induced by the private private information, which forms the efficient frontier between $(\BSselleropt, \revselleropt)$ and $(\BSfull, \revfull)$. \Cref{fig:illustration of mix bidder optimal and seller optimal} illustrates this threshold-based construction and the resulting partition of the valuation space for 2 bidders.

\section{Conclusions}
\label{sec:conclusions}

We study how private private information shapes welfare outcomes in second-price auctions. We show that ex ante signal independence does not prevent seller-optimal efficient outcomes: the seller can fully extract surplus under a private private information structure, and approximately so with a strict equilibrium. In contrast, private private information sharply limits bidder-optimal outcomes. Bidder surplus generally cannot equal maximal welfare and can be made arbitrarily close to maximal welfare only under the special prior distributions characterized in the paper. 
We also characterize the efficient surplus splits attainable under private private information. 
A natural direction for future research is to endogenize the private private requirement. In particular, it would be useful to understand under which market environments, privacy concerns, platform objectives, or downstream mechanism-design constraints private private information arises as an optimal disclosure policy rather than as an exogenous restriction.

\appendix
\section{Proofs in \texorpdfstring{\Cref{sec:pri pri rev}}{rev}}
\label{apx:proofs in pri pri rev}
\label{apx:subsec maximum rev construction}


\begin{proof}[Proof of \Cref{lem:suff condi for full suplus extraction}]
We first prove that Eqn.~\eqref{cond:loser} can be implied jointly by \ref{eq:construction rule} and \ref{cond:winner}. We then establish the efficient equilibrium.  

\xhdr{Eqn.~\eqref{cond:loser} holds due to \ref{eq:construction rule} and \ref{cond:winner}} Fix an arbitrary signal profile $\signalprofile$. Let $i$ denote a bidder with the highest signal, i.e., $\signal_i = \mymax(\signalprofile)$. (If multiple bidders hold the maximal signal, the argument applies symmetrically to any such bidder.) By \ref{eq:construction rule}, bidder $i$ holds the highest valuation in every state $\valprofile$ in the support of the posterior belief generated by $\signalprofile$. That is, $\val_i \geq \val_j$  for all $j \in [\biddernum]$.
Taking expectations over valuations conditional on $\signalprofile$ preserves this inequality 
$\expectedVal_i(\signalprofile) \geq \expectedVal_j(\signalprofile)$ for all $j \neq i$.
By \ref{cond:winner}, we have $\expectedVal_i(\signalprofile) = \secmax(\signalprofile)$. Substituting this into the inequality above yields $\expectedVal_j(\signalprofile) \leq \secmax(\signalprofile)$ for all $j$ with $\signal_j < \mymax(\signalprofile)$, which establishes Eqn.~\eqref{cond:loser}.

\xhdr{Truthful bidding is an efficient equilibrium}
Let $\biddersurplus_i(\bid_i \mid \signal_i)$ denote the interim expected payoff of bidder~$i$ when observing signal $\signal_i$ and submitting a bid $\bid_i$, assuming all other bidders bid truthfully (i.e., $\bid_j = \signal_j$ for all $j \neq i$). For any signal profile, let $y_i \equiv \max_{j \neq i} \signal_j$ denote the highest competing bid.

Consider the bidder $i$'s interim payoff $\biddersurplus_i(\bid_i \mid \signal_i)$ under the truthful bidding:
\begin{itemize}
    \item Case (i): $\signal_i \ge y_i$. Bidder $i$ wins and pays $y_i$. In this case, $\signal_i = \mymax(\signalprofile)$ and $y_i = \secmax(\signalprofile)$. 
    By \ref{cond:winner}, $\expectedVal_i(\signalprofile) = \secmax(\signalprofile)$. Thus, the bidder $i$'s realized payoff under this signal profile $\signalprofile$ is:
    \begin{align*}
    \expectedVal_i(\signalprofile) - y_i = \secmax(\signalprofile) - \secmax(\signalprofile) = 0~.
    \end{align*}
    Thus, taking the expectation of all such signaling profiles will give bidder~$i$'s interim payoff of $0$. 
    \item Case (ii): $\signal_i < y_i$. Bidder $i$ loses, resulting in a payoff of $0$.
\end{itemize}
Thus, the expected payoff from truthful bidding is identically zero: $\biddersurplus_i(\signal_i \mid \signal_i) = 0$.

Next, we show that any deviation is not profitable. Let $\deviationbid_i \neq \signal_i$ be the deviation bid.
\begin{itemize}
    \item 
    Upward deviations: 
    Consider a deviation to a higher bid $\deviationbid_i > \signal_i$. This deviation affects the allocation outcome only when $\signal_i < y_i < \deviationbid_i$.
    In this scenario, bidder $i$ turns a loss into a win. However, note that since $\signal_i < y_i$, bidder $i$ does not hold the highest signal, so $\mymax(\signalprofile) = y_i$. By Eqn.~\eqref{cond:loser}, we have $\expectedVal_i(\signalprofile) \leq \secmax(\signalprofile)$.
    The price paid is $y_i = \mymax(\signalprofile)$. Since $\secmax(\signalprofile) \leq \mymax(\signalprofile)$, the payoff is:
    $\expectedVal_i(\signalprofile) - y_i \leq \secmax(\signalprofile) - \mymax(\signalprofile) \leq 0$.
    Thus, the deviation yields a non-positive payoff, making it not profitable compared to the zero payoff from truthful bidding.
    \item 
    Downward deviations: 
    Consider a deviation to a lower bid $\deviationbid_i < \signal_i$. This affects the allocation outcome only when $\deviationbid_i < y_i < \signal_i$.
    In this scenario, bidder $i$ turns a win (with zero payoff) into a loss (with zero payoff). The bidder is indifferent, so the deviation is not strictly profitable.
\end{itemize}
Since no deviation yields a strictly positive payoff, truthful bidding $\bid_i = \signal_i$  constitutes an equilibrium. Furthermore, since $\biddersurplus_i(\signal_i \mid \signal_i) = 0$, the bidders' surplus in this equilibrium is zero.

We next show that the truthful bidding is an efficient equilibrium. 
Under truthful bidding (i.e., $\bid_j(\signal_j) = \signal_j$), the auction allocates the item to the bidder $i$ with the highest signal, i.e., $\signal_i = \mymax(\signalprofile)$. By condition \ref{eq:construction rule}, the event $\signal_i = \mymax(\signalprofile)$ implies that $\val_i = \mymax(\valprofile)$ for all $\valprofile$ in the support of the posterior. 
Consequently, the winning bidder is invariably the highest-valuation bidder, ensuring ex post efficiency.
\end{proof}

\begin{proof}[Proof of \Cref{lem:local-coupling}]
Fix a bidder $i\in[\biddernum]$. We fix a common signal space $\constructsignalspace$ and work within the class of
symmetric private private information structures. In this class, once $\constructsignalspace$ is fixed, the conditional distribution of signal profiles given that bidder $i$ is selected as the winner, denoted
by $\signalProfileDist^{(i)}$, is uniquely determined by uniform marginals and permutation invariance.
Given $\signalProfileDist^{(i)}$, the second-highest signal $\secmax(\signalprofile)$ induces a distribution
$\winnerExpectedvalProb\in\Delta(\constructsignalspace)$ via
$\winnerExpectedvalProb(\signal_\ell)\ \triangleq\ \prob[\signalprofile\sim \signalProfileDist^{(i)}]{\secmax(\signalprofile)=\signal_\ell}$.
For the signal space $\constructsignalspace$, the probability measure function of $\winnerExpectedvalProb$ is given by \eqref{eq:winner vale dist}.
Our goal is to construct, conditional on bidder $i$ winning, a joint distribution over values and
signals such that bidder $i$'s posterior mean equals the second-highest signal at every realized
signal profile, namely
\[
    \expect{\val_i\mid \signalprofile}=\secmax(\signalprofile)\quad\text{for all } s \text{ in the winning signal event.}
\]
This is exactly the requirement needed to satisfy \ref{cond:winner} for bidder $i$ on that event.
Equivalently, it suffices to show that the induced distribution $g$ is Bayes-plausible under bidder
$i$'s modified prior $\modifiedpriorCDF_i$ (the distribution of $\val_i$ conditional on $i$ winning under uniform
tie-breaking). 
We establish this Bayes-plausibility by choosing the signal space parameters $(\lowerSignal,\upperSignal)$ so that $\winnerExpectedvalProb$ and $\modifiedpriorCDF_i$ have the same mean and $\modifiedpriorCDF_i$  is dominated by  $\winnerExpectedvalProb$ in the second-order stochastic dominance.
Let $\winnerpriorCDF$ and $\winnerposteriorCDF$ denote the CDFs of $\modifiedpriorCDF_i$ and $\winnerExpectedvalProb$, respectively. 
Let $\winnerMean = \expect[\modifiedpriorCDF_{i}]{\val_i}$. We first show that there exist parameters $(\lowerSignal, \upperSignal) \in \mathbb{R}^+ \times \mathbb{R}^+$ such that 
$\expect[\winnerExpectedvalProb]{\signal_\ell} = \winnerMean$.
We have 
$\expect[\winnerExpectedvalProb]{\signal_\ell} = \int_{\lowerSignal}^{\upperSignal} 
\signal_\ell \cdot \winnerExpectedvalProb(\signal_\ell)
\, \dd \signal_\ell$.
Since the density $\winnerExpectedvalProb$ is continuous in $(\lowerSignal, \upperSignal)$, $\expect[\winnerExpectedvalProb]{\signal_\ell}$ is a continuous function of the parameters $(\lowerSignal, \upperSignal)$. 
Fix the interval length $\Length = \upperSignal - \lowerSignal$ to be a constant such that $0 < \Length$. We perform a change of variables with $\delta_\ell = \signal_\ell - \lowerSignal$, where $\delta_\ell \in (0, \Length)$. 
The probability density is defined as follows
$\gamma(\delta_\ell; \Length)  
\triangleq \biddernum(\biddernum-1)\left(\frac{\delta_\ell}{\Length}\right)^{\biddernum-2}\cdot \frac{\Length-\delta_\ell}{\Length^2}$.
The expectation then becomes
$\expect[\winnerExpectedvalProb]{\signal_\ell} 
= \lowerSignal + \int_{0}^{\Length} \delta_\ell \cdot \gamma(\delta_\ell; \Length) \, \dd \delta_\ell 
= \lowerSignal + K(\Length)$
where $K(\Length) \triangleq \int_{0}^{\Length} \delta_\ell \cdot \gamma(\delta_\ell; \Length) \, \dd \delta_\ell$ is a constant determined strictly by the interval length $\Length$.
This yields the linear form $\expect[\winnerExpectedvalProb]{\signal_\ell} = \lowerSignal + K(\Length)$, where $0 < K(\Length) < \Length$ is a constant independent of $\lowerSignal$ and is positively correlated with $\Length$. As $\lowerSignal$ varies from $\minVal$ to $\maxVal$ and $\Length \to 0$, the continuous function $\expect[\winnerExpectedvalProb]{\signal_\ell}$ spans the entire interval of $(\minVal, \maxVal)$. Since $\winnerMean \in (\minVal, \maxVal)$, by the Intermediate Value Theorem, there exists a $\lowerSignal$ such that $\expect[\winnerExpectedvalProb]{\signal_\ell} = \winnerMean$.

Next, we verify the second-order stochastic dominance condition. We require
$\MPCfunction(x) \triangleq \int_{\minVal}^{x} \bigl(\winnerpriorCDF(t) - \winnerposteriorCDF(t)\bigr) \,\dd t \ge 0$ for all $x \in [\minVal, \maxVal]$.
Note that for any $x \in [\minVal, \maxVal]$, the term $\int_{\minVal}^{x} \winnerpriorCDF(t) \,\dd t = \expect[\modifiedpriorCDF_i]{\max\{0, x-\val_i\}}$ is a constant determined solely by the prior $\modifiedpriorCDF_i$.

Now, consider the second term. Having fixed $\lowerSignal = \winnerMean - K(\Length)$ to match the mean, the integral of the posterior CDF can be expressed as
$\int_{\minVal}^{x} \winnerposteriorCDF(t) \,\dd t = \expect[\winnerExpectedvalProb]{\max\{0, x-\signal_\ell\}}$.
This expression is a continuous function of the interval length $\Length$. As $\Length \to 0$, the support of $\winnerExpectedvalProb$ shrinks to the singleton $\{\winnerMean\}$, meaning $\winnerExpectedvalProb$ converges weakly to the Dirac measure $\delta_{\winnerMean}$. Consequently, for any $x < \winnerMean$:
\begin{align*}
    \lim_{\Length \to 0} \expect[\winnerExpectedvalProb]{\max\{0, x-\signal_\ell\}} = \max\{0, x-\winnerMean\} = 0~.
\end{align*}
Since $\winnerMean$ is the mean of the prior $\modifiedpriorCDF_i$, the point mass $\delta_{\winnerMean}$ constitutes a Mean-Preserving Contraction (MPC) of $\modifiedpriorCDF_i$. By the strict convexity of the max function, for any non-degenerate prior, the inequality holds strictly for small $\Length$. Therefore, by continuity, there exists a sufficiently small $\Length > 0$ such that the SOSD condition is satisfied, implying that $\winnerExpectedvalProb$ can be induced by $\modifiedpriorCDF_i$.

We are now ready to construct the local information structure $\signalscheme^{(i)}$.
As the distribution $\winnerExpectedvalProb$ is Bayes-plausible under the modified prior $\modifiedpriorCDF_i$, there exists a martingale coupling between $\val_i\sim\modifiedpriorCDF_i$ and $\expectedVal_i \sim \winnerExpectedvalProb$:
equivalently, there are conditional distributions $\{\conditionalDist_{t}\}_{t\in\constructsignalspace} \in \Delta(\valspace)$ such that
\begin{align*}
    \sum\nolimits_{t\in\constructsignalspace} \winnerExpectedvalProb(t)\,\conditionalDist_{t}(\val_i) = \modifiedpriorCDF_i(\val_i)
    \quad\text{and}\quad
    \expect[\conditionalDist_{t}]{\val_i} =t\ \ \text{for all }t\in \constructsignalspace~.
\end{align*}

We construct the joint distribution $\signalscheme^{(i)}$ as follows:
We first sample $\signalprofile\sim \signalProfileDist^{(i)}$ and set $t=\secmax(\signalprofile)$, then we sample $\val_i\sim \conditionalDist_t$, and finally, we sample $\val_{-i}$ from the conditional distribution of $\modifiedpriorCDF^{(i)}(\cdot\mid \val_i)$ independently of $\signalprofile$ given $v_i$. 
That is, for any $\valprofile \in \valspace^{(i)}$ and $\signalprofile \in \constructsignalspace^{(i)}$:
\begin{align*}
    \signalscheme^{(i)}(\valprofile, \signalprofile) & = \prob{\valprofile \mid \signalprofile} \cdot \signalProfileDist^{(i)}(\signalprofile) = \prob{(\val_i, \val_{-i}) \mid \signalprofile} \cdot \signalProfileDist^{(i)}(\signalprofile) 
    = \prob{\val_i \mid \signalprofile} \cdot \prob{\val_{-i} \mid \signalprofile, \val_i}\cdot\signalProfileDist^{(i)}(\signalprofile)~,
\end{align*}
By the construction we have that $\val_{-i}$ is independent of $\signalprofile$ given $\val_i$. Thus,
\begin{align*}
    \signalscheme^{(i)}(\valprofile, \signalprofile) 
    = \prob{\val_i \mid \signalprofile} \cdot \prob{\val_{-i} \mid \signalprofile, \val_i}\cdot\signalProfileDist^{(i)}(\signalprofile) 
    & = \prob{\val_i \mid \signalprofile} \cdot \prob{\val_{-i} \mid \val_i}\cdot\signalProfileDist^{(i)}(\signalprofile)\\
    & = \conditionalDist_{\secmax(\signalprofile)}(\val_i) \cdot \modifiedpriorCDF^{(i)}(\val_{-i} \mid \val_i) \cdot \signalProfileDist^{(i)}(s)~.
\end{align*}
By construction, the $\signalprofile$ marginal is $\signalProfileDist^{(i)}$, the $\valprofile$ marginal is $\modifiedpriorCDF^{(i)}$, and for all $\signalprofile$, 
\begin{align*}
    \expect[\signalscheme^{(i)}]{\val_i \mid \signalprofile} &= \sum\nolimits_{\val_i} \val_i \cdot \frac{\sum\nolimits_{\val_{-i}} \signalscheme^{(i)}((\val_i, \val_{-i}), \signalprofile)}{\signalProfileDist^{(i)}(\signalprofile)}\\
    & = \sum\nolimits_{\val_i} \val_i \cdot \frac{\sum\nolimits_{\val_{-i}}\conditionalDist_{\secmax(\signalprofile)}(\val_i) \cdot \modifiedpriorCDF^{(i)}(\val_{-i} \mid \val_i) \cdot \signalProfileDist^{(i)}(s)}{\signalProfileDist^{(i)}(\signalprofile)}\\
    & = \sum\nolimits_{\val_i} \val_{i} \cdot  \conditionalDist_{\secmax(\signalprofile)}(\val_i)  = \expect[\conditionalDist_{\secmax(\signalprofile)}]{\val_i} =\secmax(\signalprofile)~.
\end{align*}
which satisfies condition \wtresedit{(3)}.
Moreover, for all $\signalprofile \in \strictsignalspace^{(i)}$, since $\valprofile \in \valspace^{(i)}$ and only bidder~$i$ has the highest signal and 
\wtresedit{condition (3)}
holds, we have that the signal profile $\signalprofile$ satisfies conditions \ref{cond:winner} and \ref{eq:construction rule}.
\end{proof}

\begin{proof}[Proof of \Cref{lem:global-assembly}]
We first verify the marginals.
For any value profile $\valprofile\in \valspace^\biddernum$,
\begin{align*}
    \int_{\signalprofile\in \constructsignalspace^\biddernum}\constructsignalscheme(\valprofile,\signalprofile)\, \dd \signalprofile
    =\frac{1}{\biddernum}\sum\nolimits_{i\in[\biddernum]} \int_{\signalprofile\in\constructsignalspace^{(i)}} \constructsignalscheme^{(i)}(\valprofile,\signalprofile) \, \dd \signalprofile
    =\frac{1}{\biddernum}\sum\nolimits_{i\in[\biddernum]} \modifiedpriorCDF^{(i)}(\valprofile)~.
\end{align*}
By definition of $\modifiedpriorCDF^{(i)}$, we have $\modifiedpriorCDF^{(i)}(\valprofile)=0$ unless $\valprofile\in \valspace^{(i)}$,
and when $\valprofile\in \valspace^{(i)}$ the tie weighted normalization implies
$\modifiedpriorCDF^{(i)}(\valprofile)=\biddernum\cdot \valprior(\valprofile)/\winnerNum(\valprofile)$. Therefore
\begin{align*}
    \frac{1}{\biddernum}\sum\nolimits_{i\in[\biddernum]} \modifiedpriorCDF^{(i)}(\valprofile)
    =\frac{1}{\biddernum}\sum\nolimits_{i\in\arg\max(v)} \biddernum\cdot \frac{\valprior(\valprofile)}{\winnerNum(\valprofile)}
    =\valprior(\valprofile)~,
\end{align*}
so the $\valprofile$ marginal of $\constructsignalscheme$ is exactly $\valprior$.
Similarly, for any signal profile $\signalprofile\in\constructsignalspace^\biddernum$,
\begin{align*}
    \sum\nolimits_{\valprofile\in \valspace^\biddernum}\constructsignalscheme(\valprofile,\signalprofile)
    =\frac{1}{\biddernum}\sum\nolimits_{i\in[\biddernum]} \sum\nolimits_{\valprofile\in \valspace^{(i)}} \constructsignalscheme^{(i)}(\valprofile,\signalprofile)
    =\frac{1}{\biddernum}\sum\nolimits_{i\in[\biddernum]} \signalProfileDist^{(i)}(\signalprofile)~.
\end{align*}
$\signalProfileDist^{(i)}(\signalprofile) = 0$ for all $\signalprofile \notin \strictsignalspace^{(i)}$, and for
$\signalprofile\in\strictsignalspace^{(i)}$ we have $\signalProfileDist^{(i)}(\signalprofile)=\biddernum\cdot \signalProfileDist(\signalprofile)$. 
Thus, we have 
$\frac{1}{\biddernum}\sum\nolimits_{i\in[\biddernum]} \signalProfileDist^{(i)}(\signalprofile)
=\frac{1}{\biddernum} \cdot \biddernum \signalProfileDist(\signalprofile)
=\signalProfileDist(\signalprofile)$.
Thus, the $\signalprofile$ marginal is uniform on $\constructsignalspace^\biddernum$, which implies the signals are independent and
uniform across bidders, so the information structure is private private. Permutation invariance
follows from symmetry (construct $\constructsignalscheme^{(i)}$ by permuting $\constructsignalscheme^{(1)}$). Thus, $\hat{\inforstructure}$ is a symmetric private private information structure.

We define $\strictsignalspace = \cup_{i} \strictsignalspace^{(i)} $. We verify the conditions \ref{cond:winner} and \ref{eq:construction rule} for all $\signalprofile \in \strictsignalspace$.
Fix $\signalprofile\in\strictsignalspace$ and a bidder $j$ with $\signal_j=\mymax(\signalprofile)$. By Eqn.~\eqref{eq:global-mixture},
$\constructsignalscheme(\cdot\mid \signalprofile)$ is a convex combination of $\constructsignalscheme^{(i)}(\cdot\mid \signalprofile)$ over indices $i$ with
$\signalprofile\in\strictsignalspace^{(i)}$. For each such $i$, $\valspace^{(i)}$ implies that under
$\constructsignalscheme^{(i)}(\cdot\mid \signalprofile)$ we have $\val_i =\max(\valprofile)$. Hence the same holds under the
mixture $\constructsignalscheme(\cdot\mid \signalprofile)$, we thus prove \ref{eq:construction rule}.
Moreover, 
\wtresedit{condition (3)}
gives $\expect[\constructsignalscheme^{(i)}]{\val_i\mid \signalprofile}=\secmax(\signalprofile)$ for any $\signalprofile \in \strictsignalspace^{(i)}$, we thus prove \ref{cond:winner}.
\end{proof}

\begin{proof}[Proof of \Cref{prop:sym pri-pri rev max}]
In \Cref{lem:global-assembly}, we construct the symmetric private private information
structure $\hat{\inforstructure}
=
(\constructsignalspace^\biddernum,\constructsignalscheme)$
satisfying conditions \ref{cond:winner} and \ref{eq:construction rule} for all signal
profiles $\signalprofile \in \bigcup_{i\in[\biddernum]}\strictsignalspace^{(i)}$.
Moreover, by construction, the marginal distribution of the signal profile is the uniform
product distribution on the continuous signal space \(\constructsignalspace^\biddernum\).
Thus, the signal distribution is atomless, and the set of signal profiles with ties, i.e.,
$\constructsignalspace^\biddernum
\setminus
\bigcup\nolimits_{i\in[\biddernum]}\strictsignalspace^{(i)}$
has measure zero. Consequently, conditions \ref{cond:winner} and
\ref{eq:construction rule} hold with probability one, and with probability one there is a unique highest signal.

Therefore, \Cref{lem:truthful bidding} (in online appendix) applies: under \(\hat{\inforstructure}\), truthful bidding is a Bayes Nash equilibrium, the induced allocation is efficient, and bidders' surplus is zero. 
Thus the constructed information structure achieves efficiency and extracts full surplus. This completes the proof.
\end{proof}

\section{Proofs in \texorpdfstring{\Cref{sec:pri pri bidders surplus}}{biddersurplus}}
\label{apx:proofs in pri pri biddersurplus}

\begin{proof}[Proof of \Cref{thm:pripri max bidder surplus}]
\wtedit{With \Cref{lem:truthful bidding}, it suffices to consider truthful bidding as the equilibrium for all information structures. }
    We proceed by contradiction. Suppose there exists a private private information structure $\inforstructure$ that admits an equilibrium $\bidstrategyprofile$ (where $\bidstrategyprofile_i(\signal_i) = \signal_i$) such that the bidders' aggregate surplus equals the efficient surplus $\optwelfare$. Since the total social surplus is the sum of bidder surplus and seller revenue, the assumption that bidders capture the entire efficient surplus implies that the seller's expected revenue must be zero.
    Let $\signalspace_i$ denote the support of the marginal signal distribution for bidder $i$. We consider three exhaustive cases regarding the signal structures:
    \begin{itemize}
    \item 
    
    \textbf{Case 1: Two or more bidders have positive signal support.}
    Suppose there exist distinct bidders $i, j$ whose marginal distributions assign positive mass to strictly positive signals. Let $P_k \triangleq \prob{\signal_k > 0} > 0$ for $k \in \{i, j\}$.
    By the independence, the joint distribution is the product of marginals. Consequently, the event where both bidders simultaneously receive positive signals, $\event \triangleq \{ \signalprofile : \signal_i > 0 \text{ and } \signal_j > 0 \}$, occurs with strictly positive probability, namely,
    $\prob{\event} = P_i \cdot P_j > 0$
    For any realization $\signalprofile \in \event$, we have $\secmax(\signalprofile) \ge \min\{\signal_i, \signal_j\} > 0$.
    The seller's expected revenue is thus strictly positive, namely, 
    $\Rev \ge \int_{\event} \secmax(\signalprofile) \, \dd \signalProfileDist(\signalprofile) > 0$.
    This contradicts the requirement for $\biddersurplus(\inforstructure, \bidstrategyprofile) = \optwelfare$.


   \item 
    \xhdr{Case 2: Exactly one bidder possesses non-zero signals}
Suppose there exists exactly one bidder $i$ such that $\supp(\signalspace_i) \not\subseteq \{0\}$, while for all other bidders $j \neq i$, the signal distribution is degenerate at zero (i.e., $\prob{\signal_j = 0} = 1$). 
Since valuations are not purely common values and the prior is symmetric, there exists a set of valuation profiles with positive measure where bidder $i$ does not hold the highest valuation. Specifically, let $\mathcal{E} = \{ \valprofile : \val_i < \max_{k \neq i} \val_k \}$. By symmetry and the non-common value assumption, $\prob{\mathcal{E}} > 0$.
Consider any profile $\valprofile \in \mathcal{E}$. For all bidders $j \neq i$, the signal is deterministically $\signal_j = 0$. Bidder $i$, however, receives a signal $\signal_i \ge 0$. Consequently, bidder $i$ always receives a signal that is at least as high as any other bidder's signal ($\signal_i \ge \max_{j \neq i} \signal_j = 0$).
Thus, bidder $i$ wins the item with strictly positive probability, which violates efficiency. 

    \item 
    \textbf{Case 3: All bidders possess only the zero signal.}
    Suppose $\signalspace_i = \{0\}$ for all $i$. The allocation is necessarily state-independent. Since the efficient winner varies with $\valprofile$ (due to non-common values), efficiency is trivially violated.
    \end{itemize}
    In summary, any private private information structure either yields positive seller revenue (Case 1) or fails to implement the efficient allocation(Cases 2 and 3). Thus, no such private private information structure exists. 
\end{proof}

\section{Proofs in \texorpdfstring{\Cref{sec:general}}{general}}
\label{apx:proofs in general}


\begin{proof}[Proof of \Cref{thm:iff pri pri efficient}]

We explicitly construct a parametric class of private private information structures, denoted by $\inforstructure^{\alpha}$. We demonstrate that by varying the parameter $\alpha$ (and the associated signal space), the seller's revenue $\Rev(\inforstructure^{\alpha}, \bidstrategyprofile)$ continuously spans the entire interval $(0, \optwelfare)$. Crucially, this construction strictly preserves allocative efficiency throughout the process, thereby fully characterizing the efficient frontier.

For any value profile $\valprofile$, we refer to the bidder with a positive valuation under this $\valprofile$ as the active bidder. Let $\maxVal = \max\{\mymax(\valprofile) \text{ for all } \valprofile \}$. We construct the information structure as follows. Define $Q_1 = \sum_{\valprofile: \val_1 = \cdots = \val_\biddernum} \valprior(\valprofile)$ and $Q_2 = \sum_{\valprofile: \valprofile \text{ has a single active bidder}} \valprior(\valprofile) = 1-Q_1$. 
We further define 
\begin{align*}
    \winnerMean 
    \triangleq \frac{1}{Q_1} \cdot \sum_{\valprofile: \val_1 = \cdots = \val_\biddernum} \max_i v_i \cdot \valprior(\valprofile)~,\;\;
    \val^\dagger 
    \triangleq \frac{1}{Q_2} \cdot \sum_{\valprofile: \valprofile \text{ has a single active bidder}} \max_i \val_i \valprior(\valprofile)~.
\end{align*}

For bidder 1, the signal space is $\signalspace_1 = \{\maxVal, \lsignal_{\alpha}\}$ with marginal signal distribution $\signalProfileDist_1 = ((1-\alpha)\cdot Q_1 + \frac{1}{\biddernum}Q_2)\delta_{(\maxVal)} + (\alpha\cdot Q_1 + \frac{(\biddernum-1)Q_2}{\biddernum}) \delta_{(\lsignal_{\alpha})}$.
For each bidder $i \in [2:\biddernum-1]$, the signal space is $\signalspace_i = \{\hsignal_{\alpha, i}, \lsignal_{\alpha}\}$ with distribution $\signalProfileDist_i = \frac{1}{\biddernum+1-i} \cdot \delta_{(\hsignal_{\alpha, i})} + \frac{\biddernum-i}{\biddernum+1-i} \cdot \delta_{(\lsignal_{\alpha})}$. For bidder $\biddernum$, let $\signalspace_\biddernum = \{\hsignal_{\alpha, \biddernum}\}$. We impose the ordering $\maxVal > \hsignal_{\alpha, 2} > \hsignal_{\alpha, 3} > \cdots > \hsignal_{\alpha, \biddernum} > \lsignal_{\alpha}$, and define $\hsignal_{\alpha, 1} = \maxVal$ for all $\alpha$.

Given any value profile $\valprofile \in \supp(\valprior)$, signals are generated as follows:
\begin{itemize}
    \item If $\valprofile$ features a single active bidder $i \in [1: \biddernum]$ (i.e., $\val_i > 0$): We set $\signal_j = \lsignal_{\alpha}$ for all $j < i$ and $\signal_i = \hsignal_{\alpha, i}$. For $j > i$, the signal $\signal_j$ is drawn independently from $\signalProfileDist_j$.
    \item If $\valprofile$ corresponds to the perfectly correlated case: We sample $\signal_1$ from $(1-\alpha) \cdot \delta_{(\hsignal_{\alpha, 1})} + \alpha\cdot \delta_{(\lsignal_{\alpha})}$ and draw each $\signal_i$ (for $i > 1$) from $\signalProfileDist_i$.
\end{itemize}
By construction, the signals are mutually independent, ensuring that the structure is a private private information structure.
Next, we analyze the incentives for truthful bidding. For any $\ell > 1$, define:
\begin{align*}
    \theta(\alpha) 
    \triangleq \prob{ \valprofile \text{ is perfectly correlated } \mid \signal_l = \mymax(\signalprofile)}
    =
    \frac{Q_1\cdot \frac{\alpha}{\biddernum-1}}{Q_1\cdot \frac{\alpha}{\biddernum-1} + Q_2 \cdot \frac{1}{\biddernum}}~.
\end{align*}
Given $i, \ell > 1$, let $\val^{\cc{inact}} \triangleq \expect{\val_i \mid \signal_\ell = \mymax(\signalprofile), i \neq \ell} = \winnerMean \cdot \theta(\alpha)$ denote bidder~$i$'s expected valuation when $\signal_\ell = \mymax(\signalprofile)$ (where $i \neq \ell$),
For any $\ell > 1$, let $\val^{\cc{win}}_{\ge 2} \triangleq \expect{\val_\ell \mid \signal_\ell = \mymax(\signalprofile)} = (1-\theta(\alpha)) \cdot \val^\dagger + \theta(\alpha) \cdot \winnerMean$ denote bidder~$\ell$'s expected valuation when $\signal_\ell = \mymax(\signalprofile)$.
Similarly, let $\val^{\cc{win}}_1 \triangleq \frac{Q_2/\biddernum}{Q_1(1-\alpha) + Q_2/\biddernum} \cdot \val^\dagger + \frac{Q_1(1-\alpha)}{Q_1(1-\alpha) + Q_2/\biddernum} \cdot \winnerMean$ denote bidder~$1$'s expected valuation when $\signal_1 = \maxVal$.
Let $\biddersurplus_i(\bid \mid \signal_i)$ be bidder~$i$'s expected payoff given signal $\signal_i$ and bid $b$, assuming other bidders bid truthfully. For each bidder $i > 1$:
\begin{align*}
    \biddersurplus_i(\hsignal_{\alpha, i} \mid \hsignal_{\alpha, i}) \ge \prob{\hsignal_{\alpha, i} = \mymax(\signalprofile)} (\val^{\cc{win}}_{\ge 2} - \hsignal_{\alpha, i+1})~,
\end{align*}
and $\biddersurplus_i(\lsignal_{\alpha} \mid \lsignal_{\alpha}) = 0$.
Consider a deviating bid $\bid^\dagger$. For any bidder $i>1$, the gain from deviation is bounded by:
\begin{align*}
 \biddersurplus_i(\bid^\dagger \mid \signal_i) - \biddersurplus_i(\signal_i \mid \signal_i) \leq \max \{ \prob{b^\dagger > \mymax(\signalprofile)} \cdot (\val^{\cc{inact}} - \lsignal_{\alpha}), 0 \}~.
\end{align*}
Therefore, if we ensure $\val^{\cc{win}}_{\ge 2} - \hsignal_{\alpha, 2} \ge 0$ and $\val^{\cc{inact}} - \lsignal_{\alpha} \leq 0$, no bidder $i > 1$ has an incentive to deviate.
For bidder 1, we have $\biddersurplus_1(\lsignal_\alpha \mid \lsignal_\alpha) = 0$, and 
$\biddersurplus_1(\maxVal \mid \maxVal) \ge  \val^{\cc{win}}_1 - \hsignal_{\alpha, 2}$.
We restrict attention to $\val^{\cc{win}}_1 - \hsignal_{\alpha, 2} \ge 0$. Under this condition, for any deviation $\bid^\dagger$, we have $\biddersurplus_1(\signal_1 \mid \signal_1) - \biddersurplus_1(\bid^\dagger \mid \signal_1) \leq 0$.
Thus, given any $\signal_{\alpha, 2} \leq \min(\val^{\cc{win}}_{\ge 2}, \val^{\cc{win}}_{1})$ and $\lsignal_\alpha \ge \val^{\cc{inact}}$, truthful bidding constitutes an equilibrium.

Observe that $\val^{\cc{win}}_{\ge 2}$ is strictly increasing and $\val^{\cc{win}}_{1}$ is strictly decreasing in $\alpha$. Equality holds at $\alpha = \frac{\biddernum-1}{\biddernum}$. Since $\val^{\cc{win}}_{\ge 2} > \val^{\cc{inact}}$ for all $\alpha$, it follows that for any $\alpha \in [0, \frac{\biddernum-1}{\biddernum}]$, there exist signal values $\hsignal_{\alpha, 2}$ and $\lsignal_{\alpha}$ such that $\val^{\cc{win}}_{\ge 2} \ge \signal_{\alpha, 2} > \lsignal_\alpha \ge  \val^{\cc{inact}}$. These signals ensure that truthful bidding forms an equilibrium satisfying allocative efficiency.

We now proceed to analyze the seller's revenue as a function of $\alpha$ and the associated signal space. Given any $\alpha$ and an arbitrarily small $\eps > 0$, we set $\hsignal_{\alpha, 2} = \lsignal_\alpha + \eps$ and maintain the order $\maxVal > \hsignal_{\alpha, 2} > \hsignal_{\alpha, 3} > \cdots > \hsignal_{\alpha, \biddernum} > \lsignal_{\alpha}$. Consequently, for any $\alpha \in [0, \frac{\biddernum-1}{\biddernum}]$, the revenue $\Rev(\inforstructure^\alpha, \bidstrategyprofile)$ lies within $(\val^{\cc{inact}}(\alpha), \val^{\cc{win}}_{\ge 2}(\alpha))$. Moreover, since $\eps$ can be chosen arbitrarily small, any value strictly within this interval is attainable.
When $\alpha = 0$, $\val^{\cc{inact}}(\alpha) = 0$. When $\alpha = \frac{\biddernum-1}{\biddernum}$, we have $\val^{\cc{win}}_{\ge 2}(\alpha) = Q_1 \cdot \winnerMean + Q_2 \cdot \val^\dagger = \optwelfare$.
Since $\val^{\cc{win}}_{\ge 2}(\alpha) > \val^{\cc{inact}}(\alpha)$ for all $\alpha \in [0, \frac{\biddernum-1}{\biddernum}]$ and both $\val^{\cc{win}}_{\ge 2}(\alpha)$ and $\val^{\cc{inact}}(\alpha)$ are increasing continuously in $\alpha$, we can construct a structure $\inforstructure^\alpha$ to achieve any revenue $\Rev \in (0, \optwelfare)$ while maintaining efficiency. This fully characterizes the efficient frontier.
\qedhere
\end{proof}

\begin{proof}[Proof of \Cref{thm:ipv pri pri efficient}]
    We construct the desired outcome by a carefully designed mixing between the full information structure and the full surplus extraction information. 
    To handle potential atoms in the distribution and ensure continuity, we parameterize this mixture by a threshold $t \in [0, \valUB]$ and a mass allocation parameter $q \in [0, 1]$ for the boundary value $t$.
    Specifically, we partition the value space into a ``high region'' and a ``low region.'' The high region includes all values strictly greater than $t$, plus a fraction $q$ of the probability mass at $t$. The low region includes all values strictly less than $t$, plus the remaining fraction $1-q$ of the probability mass at $t$. 
    

Let $\valprior_{low}$ denote the distribution of value profiles conditional on all components $\val_i$ falling within the low region. Let $\winnerMean(t, q) \triangleq \expect[\valprofile \sim \valprior_{low}]{\max_i \val_i}$ represent the expected efficient welfare within this region. Invoking \Cref{prop:sym pri-pri rev max} and the local coupling construction from \Cref{lem:local-coupling}, there exists a symmetric private private information structure restricted to the low region with a signal space $\signalspace_{low} \subset \reals_+$ that achieves full surplus extraction for the conditional prior $\valprior_{low}$.
Importantly, as demonstrated in the proof of \Cref{lem:local-coupling}, we can construct the signal space $\signalspace_{low}$ as an interval $[\underline{\signal}, \bar{\signal}]$ of length $\eps$ around $\winnerMean(t, q)$. We choose $\varepsilon$ sufficiently small such that $\bar{\signal} < t$.
Define the restricted valuation space as $\valspace_i^{(t)} = \{\val_i \in \valspace_i: \val_i \leq t\}$. For a fixed threshold $t$, let $\valprior_i^{(t)}$ denote the prior distribution of bidder $i$'s valuation conditional on $\val_i \le t$. We define a mapping $\mapping_i: \valspace_i^{(t)} \times [0, 1] \to \signalspace_{low}$. Let $r_i \sim \uniform([0,1])$ be an auxiliary random variable drawn independently for each bidder. The mapping $\mapping_i$ is constructed such that the random variable $\mapping_i(\val_i, r_i)$, induced by $\val_i \sim \valprior_i^{(t)}$ and $r_i \sim \uniform([0,1])$, is uniformly distributed on the interval $\signalspace_{low} = (\lowerSignal, \upperSignal)$.

Let $\signalscheme^{t}_{low}$ denote the joint signal distribution constructed in \Cref{prop:sym pri-pri rev max}. Given a valuation profile $\valprofile$, we construct the signal profile $\signalprofile$ as follows. To handle the potential probability mass at the threshold $t$, we employ a randomized classification rule parameterized by $q \in [0,1]$. Specifically, any realization $\val_i = t$ is treated as belonging to the ``high'' regime (conceptually $\val_i > t$) with probability $q$, and to the ``low'' regime (conceptually $\val_i < t$) with probability $1-q$. Consequently, for the remainder of this construction, we consider every valuation as falling strictly into either the high or low regime:
\begin{enumerate}
    \item If $\val_i < t$ for all bidders $i$, the signal profile is drawn according to the local construction: $\signalprofile \sim \signalscheme^{t}_{low}(\cdot \mid \valprofile)$. 
    \item If there exists at least one bidder with $\val_i > t$, the signals are determined deterministically (or independently) as follows:
    \begin{itemize}
        \item For every bidder $i$ with $\val_i > t$, we set $\signal_i = \val_i$.
        \item For every bidder $j$ with $\val_j < t$, we set $\signal_j = \mapping_j(\val_j, r_j)$, where $r_j \sim \uniform[0,1]$ is an auxiliary random variable drawn independently across bidders.
    \end{itemize}
\end{enumerate}
Since $\mapping_i(\val_i, r_i) \sim \uniform((\lowerSignal, \upperSignal))$, it is obvious that the signals are independent. We verify that the truthful strategy profile $\bidstrategy_i(\signal_i) = \signal_i$ for all $i$ is an equilibrium. 
Consider a bidder $i$ with signal $\signal_i$.
\begin{enumerate}
    \item \textbf{Case $\signal_i > t$ (or $\signal_i = t$ and with probability $q$)}: By construction, $\signal_i = \val_i$. Since truthful bidding is a weakly dominant strategy in second-price auctions, $\bidstrategyprofile(\signal_i) = \signal_i$ is a best response.
    \item \textbf{Case $\signal_i \in (\lowerSignal, \upperSignal)$}: Let $\Delta(b')$ denote the gain in expected payoff from deviating to a bid $b'$ compared to bidding truthfully.
    \begin{itemize}
        \item For any deviation $b' \le t$, $\Delta(b') = 0$ due to the local indifference property constructed in the low region.
        \item For any deviation $b' > t$, $\Delta(b') \leq \upperSignal - t < 0$.
    \end{itemize}
\end{enumerate}
In all cases, deviating yields no profit. Thus, truthful bidding is an equilibrium. We demonstrate that the equilibrium allocation is efficient. Let $i$ be the bidder with the highest signal, i.e., $\signal_i = \mymax(\signalprofile)$. If $\signal_i < t$, the entire profile falls within the low region, $\signalscheme_{low}^t$ (from \Cref{lem:local-coupling}) guarantees that bidder $i$ holds the highest valuation. If $\signal_i > t$, we have $\val_i = \signal_i$. Since any other bidder $j$ satisfies $\val_j \le \max\{\signal_j, t\} < \signal_i$, it follows immediately that $\val_i =  \mymax(\valprofile)$. Thus, the item is always allocated to the efficient winner.

We now calculate the total bidder surplus $\biddersurplus(t, q)$ generated by this equilibrium. The surplus depends on the realization of the highest valuation $\val_{(1)}$ and the second highest valuation $\val_{(2)}$ and their assignment to the high or low region.
    \begin{enumerate}
        \item Case 1: $\val_{(2)}$ is in the high region. Both the winner (with value $\val_{(1)}$) and the second highest bidder (with value $\val_{(2)}$) are in the high region (since $\val_{(1)} \ge \val_{(2)}$). Thus, their signals reveal their true values: $\signal_{(1)} = \val_{(1)}$ and $\signal_{(2)} = \val_{(2)}$. The winner pays the second highest bid $\signal_{(2)} = \val_{(2)}$. The realized surplus is $\val_{(1)} - \val_{(2)}$.
        \item Case 2: $\val_{(1)}$ is in the high region, and $\val_{(2)}$ is in the low region. The winner has signal $\signal_{(1)} = \val_{(1)}$. All other bidders are in the low region and have signals in $\signalspace_{low}$. The payment is determined by the highest competing signal $\max_{j \ne \text{winner}} \signal_j$, which lies in $\signalspace_{low} = [\omega(t, q)-\varepsilon, \omega(t, q)+\varepsilon]$. As we take $\varepsilon \to 0$, this payment converges to $\omega(t, q)$. The realized surplus converges to $\val_{(1)} - \omega(t, q)$.
        \item Case 3: $\val_{(1)}$ is in the low region. All bidders are in the low region. The mechanism behaves as a full surplus extraction mechanism for these types, so the bidder surplus is 0.
    \end{enumerate}
    Thus, taking expectations over valuations, the total bidder surplus is given by:
    \begin{equation}
        B(t, q) = \expect {(\val_{(1)} - \val_{(2)}) \cdot \indicator{\val_{(2)} \in \text{High}}}  + \expect {(\val_{(1)} - \omega(t, q)) \cdot \indicator{\val_{(1)} \in \text{High}, \val_{(2)} \in \text{Low}}}.
    \end{equation}
    Since the value distribution is an i.i.d. product prior with CDF $\valCDF$, the joint density measure of the order statistics on the region $x > y$ is $\biddernum(\biddernum-1) \valCDF(y)^{\biddernum-2} \, \dd \valCDF(x) \, \dd \valCDF(y)$.
    Let $\chi_{t,q}(v) \triangleq \indicator{v > t} + q\indicator{v=t}$ be the indicator (or weight) that value $v$ is in the high region. Similarly, $1-\chi_{t,q}(v) = \indicator{v < t} + (1-q)\indicator{v=t}$ is the weight for the low region.
    The expression becomes:
    \begin{align}
        B(t, q) &= \iint_{x > y} (x - y) \biddernum(\biddernum-1) \valCDF(y)^{\biddernum-2} \chi_{t,q}(y) \, \dd \valCDF(x) \, \dd \valCDF(y) \notag \\
        &+ \iint_{x > y} (x - \omega(t, q)) \biddernum(\biddernum-1) \valCDF(y)^{\biddernum-2} (1-\chi_{t,q}(y)) \chi_{t,q}(x) \, \dd \valCDF(x) \, \dd \valCDF(y).
    \end{align}

    The term $\omega(t, q)$ is the conditional expectation of maximum values in the low region and is continuous in $(t, q)$. Consequently, the bidder surplus $B(t, q)$ is continuous in $(t, q)$.
    Consider the trajectory where we vary $t$ from $\valUB$ to $\valLB$. At any point $t$ with positive probability mass (an atom), we vary $q$ from $0$ to $1$.
    When $t=\valUB$ (and $q=0$), the high region is empty, so the mechanism extracts full surplus and $B(\valUB, 0) = 0$.
    When $t=\valLB$ (and $q=1$), the low region is empty, so the mechanism is the standard second-price auction.
    By the Intermediate Value Theorem, for any feasible bidder surplus level, there exists a pair $(t, q)$ achieving it. This completes the proof that this family of mechanisms traces the efficient frontier.
\end{proof}

\newpage
\clearpage

\begin{center}
    {\Large Online Appendix for\par}
    \vspace{0.5em}
    {\Large ``Private Private Information in Second-Price Auction''\par}
\end{center}

\vspace{2em}

\section{Truthful Bidding}
\begin{lemma}
\label{lem:truthful bidding}
Fix any information structure $\inforstructure$ and a Bayes Nash equilibrium $\bidstrategyprofile$. There exists another information structure $\inforstructure\primed$
that admits truthful bidding, denoted by $\bidstrategyprofile\primed$ (i.e., $\bidstrategy_i\primed(\signal_i\primed)=\signal_i\primed$), as a Bayes Nash equilibrium such that
\begin{align*}
    \Rev(\inforstructure\primed,\bidstrategyprofile\primed)
    =
    \Rev(\inforstructure,\bidstrategyprofile),
    \quad
    \biddersurplus(\inforstructure\primed,\bidstrategyprofile\primed)
    =
    \biddersurplus(\inforstructure,\bidstrategyprofile)~.
\end{align*}
Moreover, if $\inforstructure$ is private private, then $\inforstructure\primed$ can be chosen to be private private.
\end{lemma}

\begin{proof}
Fix an information structure $\inforstructure$ and a Bayes Nash equilibrium $\bidstrategyprofile$. 
We construct a new information structure by folding the equilibrium mixed bids into the signals.

Let $\bidspace$ denote the original bid space. Define the new signal space of each bidder by $\signalspace_i\primed \equiv \bidspace$.
Under the new information structure $\inforstructure\primed$, conditional on a realized value
profile $\valprofile$, first draw the original signal profile $\signalprofile$ according to
$\signalscheme(\signalprofile\mid \valprofile)$, and then draw bids independently across
bidders according to the original equilibrium mixed strategy profile:
\begin{align*}
    \prob{\bidprofile\mid \signalprofile}
    =
    \prod_{i\in[\biddernum]}\bidstrategy_i(\bid_i\mid \signal_i)~.
\end{align*}
The new signal sent to bidder $i$ is $\signal_i\primed \triangleq \bid_i$.
Equivalently, the density/mass of the new information structure is
\begin{align*}
    \signalscheme\primed(\valprofile,\signalprofile\primed)
    =
    \valprior(\valprofile)
    \int_{\signalprofile}
    \signalscheme(\signalprofile\mid \valprofile)
    \prod_{i\in[\biddernum]}
    \bidstrategy_i(\signal_i\primed\mid \signal_i)
    \, \dd \signalprofile ~.
\end{align*}
When the original signal space is discrete, the integral above should be read as the
corresponding sum. Here the product form reflects the standard Bayes Nash interpretation
of mixed strategies: conditional on the signal profile, bidders' randomizations are independent
across bidders.

Under truthful bidding in $\inforstructure\primed$, i.e.,
$\bidstrategy_i\primed(\signal_i\primed)=\signal_i\primed$, the induced joint law of
$(\valprofile,\bidprofile)$ is exactly the same as the joint law of $(\valprofile,\bidprofile)$
under $(\inforstructure,\bidstrategyprofile)$. Since seller revenue and bidders' surplus depend
only on $(\valprofile,\bidprofile)$, it follows that
\begin{align*}
    \Rev(\inforstructure\primed,\bidstrategyprofile\primed)
    =
    \Rev(\inforstructure,\bidstrategyprofile)~,
    \quad
    \biddersurplus(\inforstructure\primed,\bidstrategyprofile\primed)
    =
    \biddersurplus(\inforstructure,\bidstrategyprofile)~.
\end{align*}

It remains to show that truthful bidding is a Bayes Nash equilibrium under $\inforstructure\primed$. 
Fix a bidder $i$ and consider an arbitrary deviation $\tau_i:\signalspace_i\primed\to\Delta(\bidspace)$ in the new game. We construct a corresponding deviation $\widehat{\bidstrategy}_i:\signalspace_i\to\Delta(\bidspace)$ in
the original game as follows:
\begin{align*}
    \widehat{\bidstrategy}_i(a\mid \signal_i)
    =
    \int_{t\in\bidspace}
    \tau_i(a\mid t)\,
    \bidstrategy_i(t\mid \signal_i)
    \, \dd t~.
\end{align*}
When the bid space is discrete, the integral is read as the corresponding sum. That is, after
observing the original signal $\signal_i$, bidder $i$ first draws $t$ according to the original
equilibrium mixed strategy $\bidstrategy_i(\cdot\mid \signal_i)$ and then applies the deviation
$\tau_i(\cdot\mid t)$. This defines a valid mixed strategy in the original game.

By construction, the joint distribution of values and final bids under $(\inforstructure\primed,(\tau_i,\bidstrategy_{-i}\primed))$ is identical to the joint distribution of values and final bids under $(\inforstructure,(\widehat{\bidstrategy}_i,\bidstrategy_{-i}))$.
Therefore,
\begin{align*}
    \biddersurplus_i(\inforstructure\primed,(\tau_i,\bidstrategy_{-i}\primed))
    =
    \biddersurplus_i(\inforstructure,(\widehat{\bidstrategy}_i,\bidstrategy_{-i}))~.
\end{align*}
Since $\bidstrategyprofile$ is a Bayes Nash equilibrium under $\inforstructure$, we have
\begin{align*}
    \biddersurplus_i(\inforstructure,\bidstrategyprofile)
    \ge
    \biddersurplus_i(\inforstructure,(\widehat{\bidstrategy}_i,\bidstrategy_{-i}))~.
\end{align*}
Using again the equality of the induced joint law under
$(\inforstructure\primed,\bidstrategyprofile\primed)$ and
$(\inforstructure,\bidstrategyprofile)$, we obtain
\begin{align*}
    \biddersurplus_i(\inforstructure\primed,\bidstrategyprofile\primed)
    \ge
    \biddersurplus_i(\inforstructure\primed,(\tau_i,\bidstrategy_{-i}\primed))~.
\end{align*}
Since bidder $i$ and the deviation $\tau_i$ were arbitrary, truthful bidding is a Bayes Nash equilibrium under $\inforstructure\primed$.

Finally, suppose that $\inforstructure$ is private private. Let $\varphi$ denote the marginal
density/mass of the original signal profile $\signalprofile$. Since $\inforstructure$ is private private, we have $\varphi(\signalprofile)
=
\prod_{i\in[\biddernum]}\varphi_i(\signal_i)$
where $\varphi_i$ is bidder $i$'s marginal signal density/mass. Let $\varphi\primed$ denote the
marginal density/mass of the new signal profile $\signalprofile\primed$. By construction,
\begin{align*}
    \varphi\primed(\signalprofile\primed)
    =
    \int_{\signalprofile}
    \prod_{i\in[\biddernum]}
    \bidstrategy_i(\signal_i\primed\mid \signal_i)
    \,
    \varphi(\signalprofile)
    \,\dd\signalprofile~.
\end{align*}
Using the product form of $\varphi$, we get
\begin{align*}
    \varphi\primed(\signalprofile\primed)
    =
    \int_{\signalprofile}
    \prod_{i\in[\biddernum]}
    \bidstrategy_i(\signal_i\primed\mid \signal_i)
    \varphi_i(\signal_i)
    \,\dd \signalprofile
    = 
    \prod_{i\in[\biddernum]}
    \int_{\signal_i}
    \bidstrategy_i(\signal_i\primed\mid \signal_i)
    \varphi_i(\signal_i)
    \,\dd\signal_i
    = 
    \prod_{i\in[\biddernum]}\varphi_i\primed(\signal_i\primed)~,
\end{align*}
where $\varphi_i\primed(\signal_i\primed) \equiv \int_{\signal_i} \bidstrategy_i(\signal_i\primed\mid \signal_i)
\varphi_i(\signal_i) \,\dd\signal_i$.
Thus the new signals are unconditionally independent. 
Thus, $\inforstructure\primed$ is private private whenever $\inforstructure$ is private private.
\end{proof}

\section{The Degenerate Priors}
\label{remark: prior remark}
Consider the class of priors where $\max_i \val_i \equiv \maxVal$ for all value profiles $\valprofile\sim \valprior$ in the support
(excluding the trivial case of identical valuations). 
Note that it is without loss to consider truthful bidding (i.e., bidding the received signal) as an equilibrium.\footnote{\wtedit{Throughout this work, for presentation simplicity, we refer to the pure strategy of bidding one's received signal as the truthful bidding. See \Cref{lem:truthful bidding} for the without loss of generality justification.}} 
Exact full surplus extraction requires the second-highest signal to equal the maximum valuation $\maxVal$ almost surely. 
Any private private information necessitates that at least $\biddernum-1$ bidders must deterministically signal $\maxVal$. Consequently, ties occur with probability one, and uniform tie-breaking inevitably results in inefficiency by assigning the item to a bidder with $\val_i < \maxVal$. 

Nevertheless, we can approximate full surplus extraction arbitrarily closely. We restrict bidder $i$'s signal space to a binary set $\{\signal_{i, l}, \signal_{i, h}\}$ (degenerate to a singleton $\{\signal_{\biddernum, h}\}$ for bidder $\biddernum$) strictly within the interval $(\maxVal - \varepsilon, \maxVal]$. 
We then impose a hierarchical structure on the high signals: $\signal_{1, h} > \dots > \signal_{\biddernum, h} > \max_{j} \signal_{j, l}$. 
For the lowest-indexed winner $i$, we assign $\signal_i = \signal_{i, h}$ and set $\signal_j = \signal_{j, l}$ for all predecessors $j < i$, while signals for successors $j > i$ are drawn independently according to specific marginals. 
This construction ensures that the winner holds the strictly highest signal and that the second-highest bid exceeds $\maxVal - \varepsilon$. Thus, as $\varepsilon \to 0$, revenue converges to the full surplus.

\section{Proofs in \texorpdfstring{\Cref{sec:pri pri bidders surplus}}{}}
\label{apx:proof pri pri bidders surplus}
\begin{proof}[Proof of \Cref{thm:pripri max bidder surplus iff}]
Note that it is without loss to focus on truthful bidding as the equilibrium. 
We first prove the sufficiency direction, and the necessity direction.

\xhdr{The $\Leftarrow$ direction}
    We first prove the sufficiency of the conditions by constructing explicit private private information structures for each case. We show that these structures achieve allocation efficiency while giving seller revenue arbitrarily close to zero (and thus bidders' surplus arbitrarily close to $\optwelfare$).

\textbf{Case 1: $\biddernum=2$ and there is no $\valprofile \in \supp(\valprior)$ such that $\val_i > \val_j > 0$.}
Fix an arbitrarily small $\varepsilon > 0$. Specifically, we require $\varepsilon < \min \{ \val_k : \valprofile \in \supp(\valprior), \val_k > 0 \}$, ensuring that $\varepsilon$ is below any positive valuation realization. We construct a private private information structure $\inforstructure$ defined by: given any value profile $\valprofile$
\[
    \signal_1 = \val_1 \quad \text{and} \quad \signal_2 \sim \uniform(0, \varepsilon) \quad (\text{independent of } \valprofile).
\]

Since $\signal_2$ is independent of $\val_1$ and $\signal_1 = \val_1$, $\signal_1$ and $\signal_2$ are independent.
We show that $\bidstrategy_i(\signal_i) = \signal_i$ is an equilibrium.
For bidder 1, since $\signal_1$ perfectly reveals $\val_1$, it is obvious that truthful bidding is weakly dominant.
For bidder 2, let $\biddersurplus_2(\bid \mid \signal_2)$ denote the ex-post payoff from bidding $\bid$ given signal $\signal_2$. Facing $\bid_1 = \val_1 \ge 0$, any deviation $b'$ is unprofitable compared to the truthful bid $\bid_2 = \signal_2$:
\begin{align*}
    \bid_1 = 0 \implies & \quad \biddersurplus_2(\signal_2 \mid \signal_2) = \val_2 \ge \biddersurplus_2(b' \mid \signal_2) = \begin{cases} \val_2 & \text{if } b' > 0 \\ 0 & \text{otherwise} \end{cases}; \\
    \bid_1 > 0 \implies & \quad \biddersurplus_2(\signal_2 \mid \signal_2) = 0 \ge \biddersurplus_2(b' \mid \signal_2) = \begin{cases} \val_2 - \val_1 \leq 0 & \text{if } b' > \val_1 \\ \frac{1}{2}(\val_2 - \val_1) \leq 0 & \bid' = \val_1  \\ 0 & \text{otherwise} \end{cases}.
\end{align*}
Thus $\bidstrategyprofile_i(\signal_i) = \signal_i$ constitutes an equilibrium.
Under truthful bidding, the allocation is ex-post efficient. Bidder 1 wins if and only if $\val_1 > 0$, in which case $\val_1 > \varepsilon > \bid_2$. Since the condition $\val_1 > 0$ implies $\val_1 \ge \val_2$, assigning the item to bidder 1 is efficient. Conversely, when $\val_1 = 0$, bidder 2 wins (since $0 < \bid_2$), which is also efficient as $\val_2 \ge 0$. Thus, the item is always awarded to the bidder with the highest valuation.
And the seller's revenue is bounded by $O(\eps)$. The payment is $\bid_2 = \signal_2$ when bidder 1 wins, and $\bid_1 = 0$ when bidder 2 wins. The expected revenue is bounded by:
\[
    \Rev = \prob{\val_1 >0} \cdot \expect[\uniform(0, \eps)]{\signal_2} + \prob{\val_1 = 0} \cdot 0 \le \prob{\val_1 = v} \cdot \varepsilon.
\]
By choosing an arbitrarily small $\varepsilon > 0$, seller's revenue approaches zero. Thus, $\biddersurplus(\inforstructure, \bidstrategy) \ge \optwelfare - O(\eps)$.

\textbf{Case 2: $\biddernum \ge 3$.}
\wtedit{Given a value profile $\valprofile$, we refer to the bidder under this $\valprofile$ that has a positive value as the active bidder.}
For any sufficiently small $\eps > 0$,  we define $t_\ell = \frac{\biddernum + 1 - \ell}{\biddernum} \cdot \eps$. Let $\maxVal = \max\{\mymax(\valprofile) \text{ for all } \valprofile \}$, we construct a private private information structure as follows. Let $Q_1 = \sum\nolimits_{\valprofile: \val_1 = \cdots = \val_\biddernum} \valprior(\valprofile)$, $Q_2 = \sum\nolimits_{\valprofile: \valprofile \text{ has a single active bidder}} \valprior(\valprofile) = 1-Q_1$. For bidder 1, the signal space $\signalspace_1 = \{\maxVal, 0\}$ and the marginal signal distribution is $\signalProfileDist_1 = (Q_1 + \frac{1}{\biddernum}Q_2)\delta_{(\maxVal)} + \frac{(\biddernum-1)Q_2}{\biddernum} \delta_{(0)}$; For bidder $i \in [2:\biddernum-1]$, $\signalspace_i = \{t_{i}, 0\}$ with the signal distribution $\signalProfileDist_i = \frac{1}{\biddernum+1-i} \cdot \delta_{(t_{\ell})} + \frac{\biddernum-i}{\biddernum+1-i} \cdot \delta_{(0)}$; For bidder $\biddernum$, $\signalspace_\biddernum = \{t_\biddernum\}$.  Given any signal profile $\valprofile \in \supp(\valprior)$:
\begin{itemize}
    \item If $\valprofile$ has a single active bidder:
    \begin{itemize}
        \item If $\val_1 > 0$. We fix $\signal_1 = \maxVal$ and $\signal_\biddernum = t_\biddernum$. We construct the distribution of signal profiles such that the  signals are independent, that is, $\signalscheme(\signalprofile \mid \valprofile) = \times_{i \in [2:\biddernum-1]} \signalProfileDist_i(\signal_i)$ for all $\signalprofile \in \{\maxVal\} \times \signalspace_2 \times \cdots \signalspace_{\biddernum-1} \times  \{t_\biddernum\}$.
        \item If $\val_i > 0$ where $i \in [2: \biddernum]$. We fix for all $j \in [1: i-1]$ such that $\signal_j = 0$ and $\signal_i = t_{i}$ and $\signal_\biddernum = t_\biddernum$. We construct signal profiles such that the signals are independent, that is, $\signalscheme(\signalprofile \mid \valprofile) = \times_{j \in [i+1:\biddernum-1]} \signalProfileDist_j(\signal_j)$. 
    \end{itemize}
    \item If $\valprofile$ is perfectly correlated: We fix $\signal_1 = \maxVal$ and $\signal_\biddernum = t_\biddernum$. We construct signal profiles such that the signals are independent, that is, $\signalscheme(\signalprofile \mid \valprofile) = \times_{i \in [2:\biddernum]} \signalProfileDist_i(\signal_i)$ for all $\signalprofile \in \{\maxVal\} \times \signalspace_2 \times \cdots \signalspace_{\biddernum-1} \times \{t_\biddernum\}$.
\end{itemize}
By construction, the signals are mutually independent, ensuring the structure is private private. We verify that the truthful strategy profile $\bidstrategy(\signal) = \signal$ constitutes an equilibrium and implements the efficient allocation. Furthermore, for any given $\eps > 0$, the construction guarantees that the second-highest signal satisfies $\secmax(\signalprofile) \le \eps$ for all $\signalprofile \in \supp(\signalProfileDist)$. Consequently, the seller's expected revenue is bounded above by $\eps$. That is $\biddersurplus(\inforstructure, \bidstrategyprofile) \ge \optwelfare - \eps$.

\xhdr{The $\Rightarrow$ direction}
We prove necessity in two cases. 

\textbf{Case 1: $\biddernum \ge 3$.}  Fix any $\biddernum \ge 3$, we show that if there exists $\valprofile \in \supp(\valprior)$ such that $\val_j \neq \val_j$, then  no private private information structure $\inforstructure$ simultaneously satisfies: (1)  truthful bidding $\bidstrategy_i(\signal_i) = \signal_i$ is an equilibrium; (2) ex-post efficiency; (3) approximate full surplus extraction: for any $\varepsilon > 0$, $\biddersurplus(\inforstructure, \bidstrategyprofile) \ge \optwelfare - \varepsilon$. 
\begin{enumerate}
    \item \textbf{Case 1a.} We first show that if there exists $\valprofile \in \supp(\valprior)$ such that $\val_j > \val_i > 0$ for distinct $i,j$, then  no private private information structure $\inforstructure$ simultaneously satisfies conditions (1)(2)(3).
We proceed by showing that conditions (1) and (2) impose a strict upper bound on bidder surplus. Specifically, we establish that for any such $\inforstructure$, there exists a constant $\constant > 0$ (dependent only on $\valprior$) such that
\(
    \biddersurplus(\inforstructure, \bidstrategyprofile) \leq \optwelfare - \constant.
\)
This implies that condition (3) cannot be satisfied. 

Suppose there exists a private private information structure $\inforstructure = ((\signalspace_i)_{i \in [\biddernum]}, \signalscheme)$ satisfying conditions (1) and (2). Let $\signalProfileDist$ denote the distribution of signal profiles. Since $\inforstructure$ is a private private information structure, the signals are  independent, thus $\signalProfileDist = \times_{i \in [\biddernum]} \signalProfileDist_i$. Fix an arbitrary bidder $i$, we define $\selectvalprofile \in  \supp(\valprofile)$ as any value profile satisfying there exists $k \neq i$ such that $\selectvalprofile_k = \mymax(\selectvalprofile) > \selectvalprofile_i > 0$. We define $\positiveProb \triangleq \valprior(\selectvalprofile)$.
Given the existence of $\valprofile \in \supp(\valprior)$ where $\val_j > \val_i > 0$, we have $\positiveProb > 0$. Note that $\positiveProb$ is a constant depending only on the prior $\valprior$. Then, for any $j\neq i$, we define 
$\quantileBid_j \triangleq \inf\; 
\left\{\bid: \prob{\signal_j > \bid} 
\leq \frac{\positiveProb}{2\biddernum}\right\} $.
We establish that $\max_{j \neq i} \quantileBid_j$ is bounded below by a positive constant. This follows from the incentive compatibility constraint: for  $\bidstrategy_i(\signal_i) = \signal_i$ to be a best response against truthful opponents, the competing bids cannot be arbitrarily low.
Suppose all bidders $j \neq i$ bid truthfully. Consider the following deviation strategy $\Devbidstrategyprofile_i$ for bidder~$i$: 
\begin{align*}
    \Devbidstrategyprofile_i(\signal_i) \triangleq \begin{cases}
        \max_{j \neq i} \quantileBid_j~, & \signal_i < \max_{j \neq i} \quantileBid_j~;\\
        \signal_i~, & \text{ otherwise}~.
    \end{cases}
\end{align*}
Let $\biddersurplus_i(\bidstrategyprofile_i)$ and $\biddersurplus_i(\Devbidstrategyprofile_i)$ denote bidder~$i$'s expected payoff under $\bidstrategyprofile_i$ and $\Devbidstrategyprofile_i$, respectively. The difference in payoffs is given by:
\begin{align*} \biddersurplus_i(\Devbidstrategyprofile_i) - \biddersurplus_i(\bidstrategyprofile_i)  = & \int_{\signalprofile: \signal_i < \mymax(\signalprofile) \leq \max_{j \neq i} \quantileBid_j} \bigg(\expectedVal_i(\signalprofile) - \mymax(\signalprofile)\bigg) \cdot \signalProfileDist(\signalprofile) \, \dd \signalprofile \\
& + \int_{\signalprofile: \signal_i = \mymax(\signalprofile) \leq \max_{j \neq i} \quantileBid_j} \frac{\winnerNum(\signalprofile) - 1}{\winnerNum(\signalprofile)} \cdot \bigg(\expectedVal_i(\signalprofile) - \mymax(\signalprofile)\bigg) \cdot \signalProfileDist(\signalprofile) \, \dd \signalprofile~,
\end{align*}
where $\expectedVal_i(\signalprofile)$ is bidder~$i$'s expected value given signal profile $\signalprofile$, and $\winnerNum(\signalprofile) = |\{\ell: \signal_{\ell} = \mymax(\signalprofile), \forall \ell \in [\biddernum]\}|$. Note that there exists $\signalprofile \in \supp(\signalProfileDist)$ such that $\signal_i < \mymax(\signalprofile) \leq \max_{j \neq i} \quantileBid_j$; Otherwise,  it contradicts the efficiency. Since truthful bidding $\bidstrategyprofile_i$ is a best response, the deviation cannot be profitable:
\begin{align*}
    0 
   & \ge \biddersurplus_i(\Devbidstrategyprofile_i) - \biddersurplus_i(\bidstrategyprofile_i) \\
   & \ge \int_{\signalprofile: \signal_i < \mymax(\signalprofile) \leq \max_{j \neq i} \quantileBid_j} 
    \expectedVal_i(\signalprofile)\cdot\signalProfileDist(\signalprofile)\,\dd \signalprofile  - \int_{\signalprofile: \signal_i 
    \leq \mymax(\signalprofile) \leq \max_{j \neq i} \quantileBid_j} \mymax(\signalprofile)\cdot\signalProfileDist(\signalprofile) \, \dd \signalprofile \\
    & \ge \int_{\signalprofile: \signal_i < \mymax(\signalprofile) \leq \max_{j \neq i} \quantileBid_j} \expectedVal_i(\signalprofile)\cdot\signalProfileDist(\signalprofile)\,\dd \signalprofile
    - \int_{\signalprofile: \signal_i \leq \mymax(\signalprofile) \leq \max_{j \neq i} \quantileBid_j} \max_{j\neq i}\quantileBid_j \cdot\signalProfileDist(\signalprofile) \, \dd \signalprofile \\
    & \ge \int_{\signalprofile: \signal_i < \mymax(\signalprofile) \leq \max_{j \neq i} \quantileBid_j} \expectedVal_i(\signalprofile)\cdot\signalProfileDist(\signalprofile)\,\dd \signalprofile - \max_{j\neq i}\quantileBid_j~.
\end{align*}
Substituting the definition of the expected value $\expectedVal_i(\signalprofile) = \frac{\int_{\valprofile} \val_i \cdot \signalscheme(\valprofile, \signalprofile) \, \dd \valprofile}{\signalProfileDist(\signalprofile)}$
we obtain
\begin{align*}
    0 
    & \ge \int_{\signalprofile: \signal_i < \mymax(\signalprofile) \leq \max_{j \neq i} \quantileBid_j} \expectedVal_i(\signalprofile)\cdot\signalProfileDist(\signalprofile)\,\dd \signalprofile - \max_{j\neq i}\quantileBid_j \\
    & = \int_{\signalprofile: \signal_i < \mymax(\signalprofile) \leq \max_{j \neq i} \quantileBid_j} \int_{\valprofile} \val_i \cdot \signalscheme(\valprofile, \signalprofile) \, \dd \valprofile\,\dd \signalprofile - \max_{j\neq i}\quantileBid_j\\
    & \ge \int_{\signalprofile: \signal_i < \mymax(\signalprofile) \leq \max_{j \neq i} \quantileBid_j}  \selectvalprofile_i \cdot \signalscheme(\selectvalprofile, \signalprofile) \,\dd \signalprofile - \max_{j\neq i}\quantileBid_j
\end{align*}
Since $\inforstructure$ achieves efficiency, for any $\signalprofile$ such that  $\signalscheme(\selectvalprofile, \signalprofile) > 0$, it must be that $\signal_i < \mymax(\signalprofile)$. Therefore, we have
$\int_{\signalprofile: \signal_i < \mymax(\signalprofile)}  \signalscheme(\selectvalprofile, \signalprofile) \,\dd \signalprofile = \positiveProb$.
Furthermore, we can decompose the integral as:
\begin{align*}
    \int_{\signalprofile: \signal_i < \mymax(\signalprofile) \leq \max_{j \neq i} \quantileBid_j}  \signalscheme(\selectvalprofile, \signalprofile) \,\dd \signalprofile  &=  \int_{\signalprofile: \signal_i < \mymax(\signalprofile)} \signalscheme(\selectvalprofile, \signalprofile) \,\dd \signalprofile  - \int_{\signalprofile: \signal_i < \mymax(\signalprofile) \text{ and } \mymax(\signalprofile) > \max_{j \neq i} \quantileBid_j} \signalscheme(\selectvalprofile, \signalprofile)\dd \signalprofile \\
    & \ge  \positiveProb - \int_{\signalprofile: \signal_i < \mymax(\signalprofile) \text{ and } \mymax(\signalprofile) > \max_{j \neq i} \quantileBid_j} \signalProfileDist(\signalprofile) \,\dd \signalprofile
\end{align*}
By the definition of $\quantileBid_j$,
\begin{align*}
    \int_{\signalprofile: \signal_i < \mymax(\signalprofile) \text{ and } \mymax(\signalprofile) > \max_{j \neq i} \quantileBid_j} \signalProfileDist(\signalprofile) \,\dd \signalprofile \leq \sum\nolimits_{j \neq i} \int_{\signal_j > \quantileBid_j} \signalProfileDist_j(\signal_j) \, \dd \signal_j \leq \frac{\positiveProb}{2} ~.
\end{align*}
This implies
$\int_{\signalprofile: \signal_i < \mymax(\signalprofile) \leq \max_{j \neq i} \quantileBid_j} \signalscheme(\selectvalprofile, \signalprofile) \, \dd \signalprofile  \ge \frac{\positiveProb}{2}$.
Combining these inequalities, we have
\begin{align*}
    0 \ge \selectvalprofile_i \cdot \int_{\signalprofile: \signal_i < \mymax(\signalprofile) \leq \max_{j \neq i} \quantileBid_j} \signalscheme(\selectvalprofile, \signalprofile)\,\dd \signalprofile - \max_{j\neq i}\quantileBid_j \ge \frac{\selectvalprofile \cdot \positiveProb}{2} - \max_{j\neq i}\quantileBid_j~.
\end{align*}
Consequently, there exists some $j \neq i$ such that $\quantileBid_j \ge \frac{\selectvalprofile \cdot \positiveProb}{2}$. 
Fix bidder $j$. By the symmetry of the prior $\valprior$, there exists $\valprofile \in \supp(\valprior)$ such that $\val_\ell > \val_j$. Applying the analogous derivation to bidder $j$ on this event, we establish that:
$\max_{\ell \neq j} \quantileBid_\ell \ge  \frac{\selectvalprofile \cdot \positiveProb}{2}$.
Thus, there exists bidder $\ell \neq j$ such that $\quantileBid_\ell \ge \frac{\selectvalprofile \cdot \positiveProb}{2}$.
Considering these bidders $j$ and $\ell$, the seller's expected revenue is bounded below by:
\begin{align*}
    \Rev(\inforstructure, \bidstrategyprofile) & \ge \min\{\quantileBid_j, \quantileBid_\ell\} \cdot \int_{\signal_j \ge \quantileBid_j} \signalProfileDist_j(\signal_j) \, \dd \signal_j \cdot \int_{\signal_\ell \ge \quantileBid_\ell} \signalProfileDist_\ell(\signal_\ell) \, \dd \signal_\ell \ge \frac{\selectvalprofile \cdot \positiveProb^3}{8 \biddernum^2}~.
\end{align*}
Thus, $\biddersurplus(\inforstructure, \bidstrategyprofile) \leq \optwelfare - \frac{\selectvalprofile \cdot \positiveProb^3}{8 \biddernum^2}$. Since $\frac{\selectvalprofile \cdot \positiveProb^3}{8 \biddernum^2}$ is a constant depending only on the prior $\valprior$, this contradicts condition (3). We note that the result holds for $\biddernum = 2$. The preceding analysis is independent of the specific value of $\biddernum$, requiring only $\biddernum \ge 2$ to ensure that $\quantileBid_j$ is well-defined.
\item \textbf{Case 1b.} We show that if there exists $\valprofile \in \supp(\valprior)$ such that $\val_i = \val_j > \val_k =  0$ for distinct $i,j,k \in [\biddernum]$, then no private private information structure $\inforstructure$ simultaneously satisfies conditions (1)-(3). The proof proceeds similarly to Case 1a. We demonstrate that conditions (1) and (2) impose a strict upper bound on bidder surplus.
Suppose there exists a private private information structure $\inforstructure = ((\signalspace)_{i \in [\biddernum]}, \signalscheme)$ satisfying conditions (1) and (2). We restrict our attention to priors where for any $\valprofile \in \supp(\valprior)$ and any bidder $i$, either $\val_i = \mymax(\valprofile)$ or $\val_i = 0$ (otherwise, the analysis of Case 1a applies). We define $\selectvalprofile \in \supp(\valprior)$ as a value profile such that there exist $i, j, k$ such that $\selectvalprofile_i = \selectvalprofile_j > \selectvalprofile_k  = 0$. We define $\positiveProb \triangleq\valprior(\selectvalprofile)$, and note that $\positiveProb > 0$. For any $i \in [\biddernum]$, we define 
$\quantileBid_i \triangleq \inf\; \left\{\bid: \prob{\signal_j > \bid} \leq \frac{\positiveProb}{2\biddernum}\right\}$.
By symmetry, there exists $\valprofile^\dagger \in \supp(\valprior)$ such that $\valprofile^\dagger_i = \valprofile^\dagger_j > \val_{\maxbidder} =  0$, let $\selectset = \{i: \valprofile^\dagger_i = \mymax(\valprofile^\dagger) \}.$ By efficiency, 
\begin{align*}
    \sum\nolimits_{i \in \selectset}\int_{\signalprofile: \signal_i = \mymax(\signalprofile)} \frac{1}{\winnerNum(\signalprofile)} \cdot \signalscheme(\valprofile^\dagger, \signalprofile) \, \dd \signalprofile = \valprior(\valprofile^\dagger) = \positiveProb~.
\end{align*}
Thus, there exists $i \in \selectset$ such that $\int_{\signalprofile: \signal_i = \mymax(\signalprofile)} \frac{1}{\winnerNum(\signalprofile)} \cdot \signalscheme(\valprofile^\dagger, \signalprofile) \, \dd \signalprofile \leq \frac{\positiveProb}{\winnerNum(\valprofile^\dagger)}$, we fix such $i$.
We define the deviation strategy $\Devbidstrategyprofile_i$ as:
\begin{align*}
    \Devbidstrategyprofile_i(\signal_i) \triangleq \begin{cases}
        \max_{j \neq i, \maxbidder} \quantileBid_j~, & \text{if } \signal_i \leq \max_{j \neq i} \quantileBid_j ~;\\
        \signal_i~, & \text{otherwise}~.
    \end{cases}
\end{align*}
Analogously to Case 1a, $\biddersurplus_i(\Devbidstrategyprofile_i) - \biddersurplus_i(\bidstrategyprofile_i) \leq 0$ implies a lower bound on the opponents' bids. 
Specifically, we derive that $\max_{j \neq i, \maxbidder} \quantileBid_j > C$, where $C$ is a positive constant depending only on the prior $\valprior$. Consequently, by symmetry, the seller's expected revenue is bounded away from zero, which implies that $\biddersurplus(\inforstructure, \bidstrategyprofile)$ is strictly bounded below $\optwelfare$. This contradicts condition (3).
\end{enumerate}

\textbf{Case 2:} $\biddernum = 2$. Fix $\biddernum =2$, we show that if there exists $\valprofile \in \supp(\valprior)$ such that $\val_j > \val_i > 0$ for distinct $i, j$, then no private private information structure $\inforstructure$ simultaneously satisfies conditions (1)-(3). This result follows directly from Case 1a, as the  analysis in Case 1a holds for any $\biddernum \ge 2$. 
\end{proof}

\section{Proofs in \texorpdfstring{\Cref{subsec:approx maximum rev construction}}{}}
\label{apx:subsec approx maximum rev construction}

\begin{proof}[Proof of \Cref{lem:suff condi for approx full suplus extraction}]
\zsedit{Consider a symmetric private private information structure $\inforstructure = (\constructsignalspace^\biddernum, \signalscheme)$ satisfying the conditions in \Cref{lem:suff condi for approx full suplus extraction}. }
We first prove if all bidders use truthful bidding strategy, then the aggregate bidders' surplus is bounded above by $\varepsilon$.

Let us fix a  bidder $i$, and fix her received signal $\signal_i \in (\lowerSignal, \upperSignal)$. 
Assume that all other bidders $j \neq i$ adopts  the truthful bidding strategy $\bidstrategy_j(\signal_j) = \signal_j$. 
With a slight abuse of notation, we let $\expectedVal_i(\signal_i, y)$ denote bidder $i$'s expected valuation given her signal $\signal_i$ and the second-highest signal is $y$. 
\zsedit{Since ties occur with zero probability, we focus on signal profiles with a unique highest signal.}

By the condition \eqref{it:highest-signal} in \Cref{lem:suff condi for approx full suplus extraction}, 
conditional on bidder $i$ winning (so $\mymax(\signalprofile)=\signal_i$), bidder $i$'s
expected valuation depends on the signal profile $\signalprofile$ only through $(\signal_i,y)$, where
$y=\secmax(\signalprofile)$. 
With a slight abuse of notation, we write this conditional expected valuation as
$\expectedVal_i(\signal_i,y)$. 
Recall that $\signalscheme\in\Delta(V^N\times S^N)$ is the joint distribution over $(\valprofile,\signalprofile)$ in the information structure, and let
$\margsignalscheme\in\Delta(\signalprofilespace^\biddernum)$ be its marginal on the signal profile space, i.e.,
$\margsignalscheme(\signalprofile) = \sum_{\valprofile} \signalscheme(\valprofile, \signalprofile)$.

For a bid $\bid_i$, let $\biddersurplus_i(\bid_i\mid \signal_i)$ denote bidder $i$'s interim expected payoff when she observes signal $\signal_i$ and submits $b=\bid_i$, assuming all other bidders bid truthfully. 
Under truthful bidding by bidder
$i$ as well, bidder $i$ obtains positive surplus only on profiles where she wins, and we can write
\begin{align*}
    \biddersurplus_i(\signal_i \mid \signal_i) 
    &=
    \int_{y \in (\lowerSignal, \signal_i)}\int_{\signalprofile:  \mymax(\signalprofile)= \signal_i, \secmax(\signalprofile) = y} \expectedVal_i(\signalprofile) - y\;\dd \margsignalscheme(\signalprofile)  \;\dd  y \\
    & = \int_{y \in (\lowerSignal, \signal_i)}\int_{\signalprofile: \mymax(\signalprofile) = \signal_i, \secmax(\signalprofile) = y } \expectedVal_i(\signal_i, y) - y\;\dd \margsignalscheme(\signalprofile)  \;\dd  y~.
    \tag{By definition of $\expectedVal_i(\signal_i, y)$}
\end{align*}
Observe that the term $(\expectedVal_i(\signal_i, y) - y)$ depends only on $\signal_i$ and $y$, and is therefore constant with respect to the inner integration over $\signalprofile$. 
Consequently, we have:
\begin{align*}
    \biddersurplus_i(\signal_i \mid \signal_i)
    & = \int_{y \in (\lowerSignal, \signal_i)} (\expectedVal_i(\signal_i, y) - y)\int_{\signalprofile: \mymax(\signalprofile) = \signal_i, \secmax(\signalprofile) = y } \;\dd \margsignalscheme(\signalprofile)  \;\dd  y  \\
    & =
    \int_{\lowerSignal}^{\signal_i} \left( \expectedVal_i(\signal_i, y) - y \right) g(y) \;\dd y~,
\end{align*}
where we have defined $\winnerExpectedvalProb(y)$ as the (tie-adjusted) conditional distribution of the second-highest signal $y$ in profiles where bidder $i$ wins under truthful bidding:
\begin{align*}
    \winnerExpectedvalProb(y) \triangleq \int_{\signalprofile: \mymax(\signalprofile) = \signal_i, \secmax(\signalprofile) = y } \;\dd \margsignalscheme(\signalprofile)
\end{align*}
for notation simplicity.
Also note that the information structure is symmetric, so we must have that 
$\winnerExpectedvalProb(y) \propto (y - \lowerSignal)^{\biddernum-2}$.

Let $\middlepoint(\signal_i) \triangleq \frac{\signal_i + \lowerSignal}{2}$. 
We define the bidder $i$'s \emph{ex-post net utility} as $\expostbiddersurplus(\signal_i, y) \triangleq \expectedVal_i(\signal_i, y) - y$. 
With the expected valuation defined in \Cref{lem:suff condi for approx full suplus extraction} (with $\mymax(\signalprofile) = \signal_i$ and $\secmax(\signalprofile) = y$):
\begin{align}
    \label{eq:expostbiddersurplus}
    \expostbiddersurplus(\signal_i, y) = 
    \begin{cases} 
        \displaystyle 
        -\frac{(y - \lowerSignal)\varepsilon}{2\biddernum} & \text{if } y \in (\lowerSignal, \middlepoint(\signal_i))~, \\[10pt]
        \displaystyle
        \frac{(\signal_i - y)\varepsilon}{\biddernum} & \text{if } y \in [\middlepoint(\signal_i), \signal_i)~.
    \end{cases}
\end{align}
Thus, we can lower bound the bidder's surplus as follows:
\begin{align*}
    \biddersurplus_i(\signal_i\mid \signal_i) 
    & = \int_{\lowerSignal}^{\middlepoint(\signal_i)}-\frac{(y-\lowerSignal)\varepsilon}{2\biddernum}\,g(y)\; \dd y
    \;+\;
    \int_{\middlepoint(\signal_i)}^{\signal_i}\frac{(\signal_i-y)\varepsilon}{\biddernum}\,g(y)\; \dd y \\
    & =\frac{\varepsilon}{\biddernum}\left(\int_{\lowerSignal}^{\middlepoint(\signal_i)}(t-\lowerSignal)\,g(\lowerSignal+\signal_i-t)\; \dd t-\frac12\int_{\lowerSignal}^{\middlepoint(\signal_i)}(t-\lowerSignal)g(t)\; \dd t\right)\\
    & \ge \frac{\varepsilon}{\biddernum}\left(\int_{\lowerSignal}^{\middlepoint(\signal_i)}(t-\lowerSignal)\,g(t)\; \dd t-\frac12\int_{\lowerSignal}^{\middlepoint(\signal_i)}(t-\lowerSignal)g(t)\; \dd t\right)
    \tag{Since $\lowerSignal+\signal_i-t\ge t$ and $g$ is non-decreasing }\\
    & =\frac{\varepsilon}{2N}\int_{\lowerSignal}^{\middlepoint(\signal_i)}(t-\lowerSignal)g(t)\; \dd t\\
    & > 0~.
    \tag{Since $\varepsilon > 0,\; t - \lowerSignal \ge 0,\; g(t) > 0$}
\end{align*}
We next show that the bidders' surplus is upper bounded by $\varepsilon$:
\begin{align*}
    \biddersurplus(\inforstructure, \bidstrategyprofile)
    & = \sum\nolimits_{i\in[\biddernum]} \int_{\lowerSignal}^{\upperSignal} \biddersurplus_i(\signal_i \mid \signal_i ) \cdot \frac{1}{\upperSignal - \lowerSignal}\; \dd \signal_i \\
    & = 
    \sum\nolimits_{i\in[\biddernum]} \int_{\lowerSignal}^{\upperSignal} 
    \frac{1}{\upperSignal - \lowerSignal} \cdot 
    \int_{\middlepoint(\signal_i)}^{\signal_i} \frac{(\signal_i - y) \varepsilon}{\biddernum} \cdot g(y) \; \dd y \; \dd \signal_i \\
    & = 
    \sum\nolimits_{i\in[\biddernum]} \int_{\lowerSignal}^{\upperSignal} 
    \frac{1}{\upperSignal - \lowerSignal} \cdot 
    \frac{\varepsilon^2}{\biddernum}
    \int_{\middlepoint(\signal_i)}^{\signal_i} g(y) \; \dd y \; \dd \signal_i 
    \tag{Due to $\signal_i - y \le \varepsilon$}\\
    & \le \varepsilon^2 < \varepsilon~.
\end{align*}
We now prove that the truthful bidding is indeed a strict equilibrium under the conditions defined in \Cref{lem:suff condi for approx full suplus extraction}. Let us fix a bidder $i\in[\biddernum]$, and fix her received private signal $\signal_i$.

\xhdr{Case 1: Under-bidding ($\bid < \signal_i$)}
Now consider a downward deviation in which bidder $i$ bids $\bid<\signal_i$. 
Let $y \;\triangleq\; \mymax_{k\neq i} \signal_k$ be the highest signal among the other bidders.
Under the continuous signal construction in \Cref{lem:suff condi for approx full suplus extraction}, ties occur with probability zero, so bidder $i$ wins iff $y\le \bid$ under
bid $\bid$, and wins iff $y\le \signal_i$ under truthful bidding. 
Thus, we can write the bidder $i$' interim expected surplus as follows,
\begin{equation}
    \label{eq:biddersurplus-defs}
    \biddersurplus_i(\bid\mid \signal_i)= \expect{\expostbiddersurplus(\signal_i,y)\mathbf{1}\{y\le \bid\} \mid \signal_i}~,
    \quad
    \biddersurplus_i(\signal_i\mid \signal_i)
    = \expect{\expostbiddersurplus(\signal_i,y)\mathbf{1}\{y\le \signal_i\}\mid \signal_i}~.
\end{equation}
Subtracting yields
\begin{equation}
    \label{eq:biddersurplus-diff}
    \biddersurplus_i(\signal_i\mid \signal_i)-\biddersurplus_i(\bid\mid \signal_i)
    = \expect{\expostbiddersurplus(\signal_i,y)\indicator{\bid<y\le \signal_i}\mid \signal_i}~.
\end{equation}
We now discuss two subcases.
\begin{itemize}
    \item 
    \emph{Case 1a:} $\bid\ge \middlepoint(\signal_i)$.
    On the event $\{\bid <y<\signal_i\}\subseteq\{\middlepoint(\signal_i)\le y<\signal_i\}$, 
    Eqn.~\eqref{eq:expostbiddersurplus} implies that
    $\expostbiddersurplus(\signal_i,y)
    =\frac{(\signal_i-y)\varepsilon}{\biddernum}>0$.
    Moreover, under the full-support continuous construction of signals in \Cref{lem:suff condi for approx full suplus extraction}, we have $\prob{\bid<y<\signal_i\mid \signal_i} > 0$ whenever $\bid <\signal_i$.
    Thus, the expectation in Eqn.~\eqref{eq:biddersurplus-diff} is strictly positive:
    \begin{align*}
    \biddersurplus_i(\signal_i\mid \signal_i)-\biddersurplus_i(\bid\mid \signal_i)
    = \expect{\expostbiddersurplus(\signal_i,y)\indicator{\bid <y\le \signal_i}\mid \signal_i}
    >0~,
    \end{align*}
    which gives us $\biddersurplus_i(\bid\mid \signal_i)<\biddersurplus_i(\signal_i\mid \signal_i)$.

    \item 
    \emph{Case 1b:} $\bid<\middlepoint(\signal_i)$.
    First, on the event $\{y\le \bid\}\subseteq\{y<\middlepoint(\signal_i)\}$, Eqn.~\eqref{eq:expostbiddersurplus} implies that
    $\expostbiddersurplus(\signal_i,y)\le 0$. 
    Thus,
    \begin{equation}
    \label{eq:biddersurplusb-nonpos}
        \biddersurplus_i(\bid\mid \signal_i)
        = \expect{\expostbiddersurplus(\signal_i,y)\indicator{y\le \bid}\mid \signal_i}
        \le 0~.
    \end{equation}
    Second, on the event $\{\middlepoint(\signal_i)<y<\signal_i\}$, Eqn.~\eqref{eq:expostbiddersurplus} implies
    $\expostbiddersurplus(\signal_i,y)>0$.
    By full support,
    we also have $\prob{\middlepoint(\signal_i)<y<\signal_i\mid \signal_i} > 0$.
    Thus,
    \begin{equation}\label{eq:biddersurplustruth-pos}
        \biddersurplus_i(\signal_i\mid \signal_i)
        =\expect{\expostbiddersurplus(\signal_i,y)\indicator{y\le \middlepoint(\signal_i)}\mid \signal_i} +
        \expect{\expostbiddersurplus(\signal_i,y)\indicator{\middlepoint(\signal_i)<y<\signal_i}\mid \signal_i}
        >0~.
    \end{equation}
    Combining \eqref{eq:biddersurplusb-nonpos} and \eqref{eq:biddersurplustruth-pos} gives us $\biddersurplus_i(\bid\mid \signal_i)<\biddersurplus_i(\signal_i\mid \signal_i)$.
\end{itemize}
Thus,  downward deviations are strictly suboptimal.


\xhdr{Case 2: Over-bidding $(\bid >\signal_i)$}
Fix bidder $i$ and condition on her signal $\signal_i\in(\lowerSignal,\upperSignal)$.
Suppose all other bidders bid truthfully, i.e., $\bid_j=\signal_j$ for all $j\neq i$.
Consider a deviation in which bidder $i$ submits a bid $\bid>\signal_i$.

Let $y \;\triangleq\; \mymax_{k\neq i} \signal_k$ be the highest signal among the other bidders. Under truthful bidding by bidder $i$, she wins if and only if $\signal_i\ge y$; under the deviation $\bid>\signal_i$, she wins if and only if $\bid\ge y$.
Thus, the deviation changes the allocation only on the event
\begin{align*}
    \event(\bid)\;\triangleq\;\{\,\signal_i<y\le \bid\,\}~.
\end{align*}
On $\event(\bid)$ the allocation (and payment) is unchanged, and thus the deviation has no effect on bidder $i$'s payoff. Thus, the interim payoff change of bidder $i$ satisfies
\begin{equation}
\label{eq:upward-dev-diff}
    \biddersurplus_i(\bid\mid \signal_i)-\biddersurplus_i(\signal_i\mid \signal_i)
    \;=\;
    \expect{(\expectedVal_i(\signalprofile)-y) \indicator{\event(\bid)} \mid  \signal_i}~,
\end{equation}
because on the event $\event(\bid)$, the bidder wins and pays $y$ (the highest competing bid), while under truthful
bidding she would lose and obtain payoff $0$.
We will show that the ex-post payoff $\expectedVal_i(\signalprofile)-y<0$ holds almost surely on $\event(\bid)$, which implies the expectation
in Eqn.~\eqref{eq:upward-dev-diff} is strictly negative whenever $\bid>\signal_i$.

Fix any signal profile $\signalprofile$ such that $\event(\bid)$ occurs, i.e., $\signal_i<y\le \bid$.
Let $j$ be a bidder holding the highest signal, so that $\signal_j=\mymax(\signalprofile)=y$.
Since $\signal_i<\mymax(\signalprofile)$, the condition \eqref{it:not-highest-signal} implies that  $\expectedVal_i(\signalprofile)\le \expectedVal_j(\signalprofile)$, and hence
\begin{equation}
\label{eq:xi-xj-compare}
    \expectedVal_i(\signalprofile)-y \;\le\; \expectedVal_j(\signalprofile)-y~.
\end{equation}
We next show $\expectedVal_j(\signalprofile)-y\le 0$, and it is strictly negative except on a tie event.
Let $z\triangleq \secmax(\signalprofile)$. We consider two cases.
\begin{itemize}
    \item 
    \emph{Case 2a:} $y = \mymax(\signalprofile)=\secmax(\signalprofile)$.
    Substituting $\secmax(\signalprofile)=y$ into the expected valuation construction in \eqref{it:highest-signal} gives us $\expectedVal_j(\signalprofile)=y$, and thus
    \begin{align*}
        \expectedVal_i(\signalprofile)-y
        \le 
        \expectedVal_j(\signalprofile)-y=0~.
        \tag{By Eqn.~\eqref{eq:xi-xj-compare}}
    \end{align*}
    
    \item 
    \emph{Case 2b:} $y = \mymax(\signalprofile)>\secmax(\signalprofile)$.
    We show $\expectedVal_j(\signalprofile)-y<0$ in both subcases of \eqref{it:highest-signal}:
    
    \begin{itemize}
        \item If $\lowerSignal<\secmax(\signalprofile)<\frac{\signal_j+\lowerSignal}{2}=\frac{y+\lowerSignal}{2}$, then \eqref{it:highest-signal} tells us
        \begin{align*}
            \expectedVal_j(\signalprofile)-y
            = 
            \secmax(\signalprofile) -\frac{(\secmax(\signalprofile) -\lowerSignal)\varepsilon}{2\biddernum}
            - y
            \;<\; \secmax(\signalprofile) - y \;<\; 0~.
        \end{align*}
    
        \item If $\frac{\signal_j+\lowerSignal}{2}=\frac{y+\lowerSignal}{2}\le \secmax(\signalprofile) <\upperSignal$, then \eqref{it:highest-signal} tells us
        \begin{align*}
            \expectedVal_j(\signalprofile)-y
            =
            \secmax(\signalprofile)-y+\frac{(y-\secmax(\signalprofile))\varepsilon}{\biddernum}
            =
            -(y-\secmax(\signalprofile))\left(1-\frac{\varepsilon}{\biddernum}\right)
            <0~,
        \end{align*}
        since $y-\secmax(\signalprofile)>0$ and $1-\varepsilon/N>0$ (for $N\ge 2$ and $\varepsilon\in(0,1]$).
    \end{itemize}
    Consequently, in this case we have $\expectedVal_j(\signalprofile)-y<0$, and thus by Eqn.~\eqref{eq:xi-xj-compare}, we have $\expectedVal_i(\signalprofile)-y<0$.
\end{itemize}
Combining the two cases, we obtain
\begin{equation}
    \label{eq:negativity-ae}
    \expectedVal_i(\signalprofile)-y\le 0 \text{ on } \event(\bid),
    \quad\text{and}\quad
    \expectedVal_i(\signalprofile)-y<0 \text{ on } \event(\bid)\cap\{y>\secmax(\signalprofile)\}~.
\end{equation}
Because the signal space is a continuous interval and the induced marginal distribution over
signals is atomless under the symmetric private private construction, ties occur with probability zero, that is, $\prob{y=\secmax(\signalprofile)\mid \signal_i} = 0$.
Thus, conditional on $\event(\bid)$ we have $y>\secmax(\signalprofile)$ almost surely, and thus by
Eqn.~\eqref{eq:negativity-ae},
\begin{align*}
(\expectedVal_i(\signalprofile)-y)\mathbf{1}\{\event(\bid)\}<0 \quad\text{almost surely on } \{\event(\bid)\}.
\end{align*}
Moreover, if $\bid>\signal_i$, then $\Pr(\event(\bid)\mid \signal_i)>0$ (there is positive probability that $y\in(\signal_i,b]$).
Therefore,
\begin{align*}
    \biddersurplus_i(\bid\mid \signal_i)-\biddersurplus_i(\signal_i\mid \signal_i)
    =
    \expect{(x_i(\signalprofile)-y)\,\mathbf{1}\{\event(\bid)\}\,\mid \, \signal_i} < 0~,
\end{align*}
which proves that any upward deviation $\bid>\signal_i$ is strictly suboptimal. 
\end{proof}

\begin{proof}[Proof of \Cref{lem: existence of infor structure with unique equilibrium}]
Fix a sufficiently small constant $\eps > 0$, we consider a symmetric private private information structure, and the common signal space $\constructsignalspace$ for each bidder  is a continuous interval $(\lowerSignal, \upperSignal)$ with $\upperSignal - \lowerSignal \leq \eps$.
For all bidders $i\in[\biddernum]$, we define
\begin{align*}
    \valspace^{(i)} 
    \triangleq 
    \left\{\valprofile\in \valspace^\biddernum: \val_i=\max(\valprofile)\right\},
    \quad
    \constructsignalspace^{(i)} 
    \triangleq 
    \left\{\signalprofile\in \constructsignalspace^\biddernum: \signal_i=\max(\signalprofile)\right\}~;
\end{align*}
Under a uniform tie breaking rule on values, we define the tie weighted distribution
$\modifiedpriorCDF^{(i)}\in\Delta(\valspace^{(i)})$ by
\begin{equation}\label{eq:hat-vartheta-profile for unique BNE}
    \modifiedpriorCDF^{(i)}(\valprofile)\ 
    \triangleq\
    \frac{1}{\winnerNum(\valprofile)}
    \cdot \frac{\valprior(\valprofile)}{\sum\nolimits_{\valprofile'\in \valspace^{(i)}} \valprior(\valprofile')\cdot \frac{1}{\winnerNum(\valprofile')}}\,,
    \quad \valprofile\in \valspace^{(i)}~.
\end{equation}
Similarly, since signal profiles are uniformly distributed on $\constructsignalspace^{\biddernum}$, the distribution on $\constructsignalspace^{(i)}$ is well defined, denoted as  $\signalProfileDist^{(i)}$. 

\xhdr{Step 1: Existence of feasible $(\lowerSignal, \upperSignal, \eps)$} 
We first show that there exists a choice of $(\lowerSignal, \upperSignal, \eps)$ such that for every $i \in [\biddernum]$ there exists a joint distribution $\signalscheme^{(i)} \in \Delta(\valspace^{(i)} \times \constructsignalspace^{(i)})$ satisfying: the marginals of $\signalscheme^{(i)}$ on $\valspace^{(i)}$ and $\constructsignalspace^{(i)}$ are $\modifiedpriorCDF^{(i)}$ and $\signalProfileDist^{(i)}$ respectively; And for every $\signalprofile \in \constructsignalspace^{(i)}$ with a unique highest signal $\signal_i$, 
\begin{align}
    \label{eq: winner expected value}
    \expect{\val_i \mid \signalprofile} = 
        \begin{cases} 
        \secmax(\signalprofile) - \displaystyle\frac{(\secmax(\signalprofile) -  \lowerSignal) \cdot \eps}{2\biddernum} & \text{if } \lowerSignal< \secmax(\signalprofile) < \displaystyle\frac{\signal_i + \lowerSignal}{2}, \\[10pt]
        \secmax(\signalprofile) + 
        \displaystyle \frac{(\signal_i - \secmax(\signalprofile))\cdot\eps}{\biddernum} & \text{if } 
        \displaystyle \frac{\signal_i + \lowerSignal}{2} \leq \secmax(\signalprofile) < \upperSignal.
        \end{cases}
        \end{align}
Fix a bidder $i \in [\biddernum]$ and a common signal space $\constructsignalspace$. Restricting attention to symmetric private private information structures, the signal space $\constructsignalspace$ uniquely determines $\signalProfileDist^{(i)}$, the conditional distribution of signal profiles given that bidder $i$ wins. 
For every signal profile $\signalprofile \in \constructsignalspace^{(i)}$ with a unique highest signal $\signal_i$, we define the function $\uniqueExpectedval(\signalprofile)$ as:
\begin{align}
  \label{eq: definition of y}\uniqueExpectedval(\signalprofile)  = 
        \begin{cases} 
        \secmax(\signalprofile) - \displaystyle\frac{(\secmax(\signalprofile) -  \lowerSignal) \cdot \eps}{2\biddernum} & \text{if } \lowerSignal< \secmax(\signalprofile) < \displaystyle\frac{\signal_i + \lowerSignal}{2}, \\[10pt]
        \secmax(\signalprofile) + 
        \displaystyle \frac{(\signal_i - \secmax(\signalprofile))\cdot\eps}{\biddernum} & \text{if } 
        \displaystyle \frac{\signal_i + \lowerSignal}{2} \leq \secmax(\signalprofile) < \upperSignal.
        \end{cases} 
\end{align}
Based on $\signalProfileDist^{(i)}$, the CDF of $\uniqueExpectedval(\signalprofile)$, denoted by $\winnerposteriorCDF$, is given by:
\begin{align*}
    \winnerposteriorCDF(y) & \triangleq  \int_{s:\ \lowerSignal< \secmax(\signalprofile) < \frac{\signal_i + \lowerSignal}{2}, \secmax(\signalprofile) - \frac{(\secmax(\signalprofile) -  \lowerSignal) \cdot \eps}{2\biddernum}\le y}{ \signalProfileDist^{(i)}(\signalprofile) \, \dd \signalprofile }\\
    & + \int_{s:\ \frac{\signal_i + \lowerSignal}{2} \leq \secmax(\signalprofile) < \upperSignal, \secmax(\signalprofile) + \frac{(\signal_i - \secmax(\signalprofile))\cdot\eps}{\biddernum} \le y}{\signalProfileDist^{(i)}(\signalprofile)\, \dd \signalprofile }
\end{align*}
Let  $\winnerExpectedvalProb$ denote the PDF of $\winnerposteriorCDF$.
Our goal is to construct a joint distribution over values and signals $\signalscheme^{(i)}$, conditional on bidder~$i$ winning, such that bidder~$i$'s posterior mean satisfies \eqref{eq: winner expected value} for every signal profile with a unique highest signal. It suffices to show that the induced distribution $\winnerExpectedvalProb$ is Bayes-plausible with respect to bidder $i$'s modified prior $\modifiedpriorCDF_i$ (the distribution of $\val_i$ conditional on $i$ winning under uniform tie-breaking). Let $\winnerpriorCDF$ denote the CDF of $\modifiedpriorCDF_i$. We establish Bayes-plausibility by choosing $(\lowerSignal, \upperSignal, \eps)$ such that $\winnerposteriorCDF$ is a MPC of $\winnerpriorCDF$.

Let $\winnerMean = \expect[\modifiedpriorCDF_{i}]{\val_i}$. We first show that there exist parameters $(\lowerSignal, \upperSignal, \eps) \in \mathbb{R}^+ \times \mathbb{R}^+ \times \mathbb{R}^+ $ such that 
$\expect[\winnerExpectedvalProb]{y} = \winnerMean$.
Let $\secmaxsignal(\signalprofile)$ denote the second-highest signal in $\signalprofile$. When bidder $i$~wins, the joint distribution of the winning signal $\signal_i$ and the second-highest signal $\secmaxsignal$ is given by:
\begin{align*}
    \jointdist^{(i)}(\signal_i, \secmaxsignal) = \frac{\biddernum(\biddernum-1)(\secmaxsignal - \lowerSignal)^{\biddernum-2}}{(\upperSignal - \lowerSignal)^\biddernum}.
\end{align*}
We note that by definition, $\uniqueExpectedval(\signalprofile)$ depends on $\signal_i$ and  $\secmaxsignal$, we use $\uniqueExpectedval(\signal_i, \secmaxsignal)$ to denote $\uniqueExpectedval(\signalprofile)$.
We show that there exists $(\lowerSignal, \upperSignal, \eps)$ such that $\expect[\winnerExpectedvalProb]{y} = \winnerMean$ under the constraint $\upperSignal - \lowerSignal \leq \eps$. We have 
\begin{align*}
    \expect[\winnerExpectedvalProb]{y} &= \int_{\lowerSignal}^{\upperSignal} \int_{\lowerSignal}^{\signal_i} \uniqueExpectedval(\signal_i, \secmaxsignal) \frac{\biddernum(\biddernum-1)(\secmaxsignal - \lowerSignal)^{\biddernum-2}}{(\upperSignal - \lowerSignal)^\biddernum} \, \dd \secmaxsignal \, \dd \signal_i.
\end{align*}
Since $\uniqueExpectedval(\signal_i, \secmaxsignal)$ is piecewise linear given $\signal_i$ and the density $\jointdist^{(i)}$ is continuous in $(\lowerSignal, \upperSignal)$, $\expect[\winnerExpectedvalProb]{y}$ is a continuous function of the parameters $(\lowerSignal, \upperSignal, \eps)$. 
Fix an arbitrary $\eps > 0$ and let the interval length $\Length = \upperSignal - \lowerSignal$ be a constant. We perform a change of variables with $\delta_i = \signal_i - \lowerSignal$ and $\delta_{\secmaxsignal} = \secmaxsignal - \lowerSignal$, where $\delta_i, \delta_{\secmaxsignal} \in (0, \Length)$. Based on \eqref{eq: definition of y}:
\begin{itemize}
    \item If $0 < \delta_{\secmaxsignal} < \frac{\delta_i}{2}$, then $\uniqueExpectedval = \lowerSignal + \delta_{\secmaxsignal} - \frac{\delta_{\secmaxsignal} \eps}{2\biddernum}$;
    \item If $\frac{\delta_i}{2} \leq \delta_{\secmaxsignal} < \Length$, then $\uniqueExpectedval = \lowerSignal + \delta_{\secmaxsignal} + \frac{(\delta_i-\delta_{\secmaxsignal})\eps}{\biddernum}.$
\end{itemize}
Consequently, $\uniqueExpectedval$ can be decomposed into $\lowerSignal + \Delta(\delta_i, \delta_{\secmaxsignal}, \eps)$, where the term $\Delta$ is bounded by $O(\Length + \eps \Length)$. The expectation then becomes:
\begin{align*}
    \expect[\winnerExpectedvalProb]{\uniqueExpectedval} &= \lowerSignal + \frac{\biddernum(\biddernum-1)}{\Length^\biddernum} \int_0^\Length \int_0^{\delta_i} \Delta(\delta_i, \delta_\secmaxsignal, \eps) \delta_\secmaxsignal^{\biddernum-2} \, \dd \delta_\secmaxsignal \, \dd \delta_i.
\end{align*}
We define
\begin{align*}
    K(\Length, \eps) = \frac{\biddernum(\biddernum-1)}{\Length^\biddernum} \int_0^\Length \int_0^{\delta_i} \Delta(\delta_i, \delta_\secmaxsignal, \eps) \delta_\secmaxsignal^{\biddernum-2} \, \dd \delta_\secmaxsignal \, \dd \delta_i~.
\end{align*}
This yields the linear form $\expect[\winnerExpectedvalProb]{\uniqueExpectedval} = \lowerSignal + K(\Length, \eps)$, where $0 < K(\Length, \eps) < \Length$ is a constant independent of $\lowerSignal$ and is positively correlated with $\Length$. As $\lowerSignal$ varies from $-\infty$ to $+\infty$, the continuous function $\expect[\winnerExpectedvalProb]{\uniqueExpectedval}$ spans the entire real line. By the Intermediate Value Theorem, there exists a $\lowerSignal$ such that $\expect[\winnerExpectedvalProb]{\uniqueExpectedval} = \winnerMean$.

Analogously, we extend this decomposition to the second-order stochastic dominance condition. Recall that the integral of the CDF can be expressed as the expected shortfall. By substituting the decomposition $\uniqueExpectedval = \lowerSignal + \Delta(\delta_i, \delta_{\secmaxsignal}, \eps)$, we obtain:
\begin{align*}
    \int_{\minVal}^{x} \winnerposteriorCDF(t) \,\dd t 
    &= \expect[\winnerExpectedvalProb]{\max\{0, x - \uniqueExpectedval\}} \\
    &= \frac{\biddernum(\biddernum-1)}{\Length^\biddernum} \int_0^\Length \int_0^{\delta_i} \max\bigl\{0, x - (\lowerSignal + \Delta(\delta_i, \delta_\secmaxsignal, \eps))\bigr\} \, \delta_\secmaxsignal^{\biddernum-2} \, \dd \delta_\secmaxsignal \, \dd \delta_i.
\end{align*}
Since the perturbation term satisfies $\Delta(\cdot) = O(\Length)$ (and $\lowerSignal$ is adjusted continuously to fix the mean), the integral expression above is continuous in $\Length$. Let us denote this function as $\Lambda(x; \Length)$.
Consider the limit as $\Length \to 0$. The term $\Delta(\cdot)$ vanishes and $\lowerSignal \to \winnerMean$. Thus, for any fixed $x$:
\begin{align*}
    \lim_{\Length \to 0} \Lambda(x; \Length) = \max\{0, x - \winnerMean\}.
\end{align*}
For the prior $\modifiedpriorCDF_i$, we have for all $x \in (\minVal, \winnerMean)$:
\begin{align*}
    \int_{\minVal}^{x} \winnerpriorCDF(t) \,\dd t > \max\{0, x - \winnerMean\}.
\end{align*}
Given any $x$,  $\int_{\minVal}^{x} \winnerpriorCDF(t) \,\dd t$ is a constant.
Since $\Lambda(x; \Length)$ is continuous function for $\Length$, and the strict inequality holds in the limit, there exists a sufficiently small $\Length > 0$ such that the inequality holds for all $x \in (\minVal, \maxVal)$. This confirms that $\winnerExpectedvalProb$ is a mean-preserving contraction of $\modifiedpriorCDF_i$.

\xhdr{Step 2: Construct $\signalscheme^{(i)}$} 
Fix any feasible $(\lowerSignal, \upperSignal, \eps)$. We are now ready to construct the local information structure $\signalscheme^{(i)}$. As the distribution $\winnerExpectedvalProb$ is Bayes-plausible under the modified prior $\modifiedpriorCDF_i$, there exists a martingale coupling between $\val_i\sim\modifiedpriorCDF_i$ and $\expectedVal_i \sim \winnerExpectedvalProb$:
equivalently, there are conditional distributions $\{\conditionalDist_{t}\}_{t\in\constructsignalspace} \in \Delta(\valspace)$ such that
\begin{align*}
\sum\nolimits_{t\in\constructsignalspace} \winnerExpectedvalProb(t)\,\conditionalDist_{t}(\val_i) = \modifiedpriorCDF_i(\val_i)
\quad\text{and}\quad
\expect[\conditionalDist_{t}]{\val_i} =t\ \ \text{for all }t\in\bar S.
\end{align*}

We construct the joint distribution $\signalscheme^{(i)}$ as follows:
We first sample $\signalprofile\sim \signalProfileDist^{(i)}$ and set $t=\uniqueExpectedval(\signalprofile)$, then we sample $\val_i\sim \conditionalDist_t$, and finally, we sample $\val_{-i}$ from the conditional distribution of $\modifiedpriorCDF^{(i)}(\cdot\mid \val_i)$ independently of $\signalprofile$ given $v_i$. 
That is, for any $\valprofile \in \valspace^{(i)}$ and $\signalprofile \in \constructsignalspace^{(i)}$:
\begin{align*}
    \signalscheme^{(i)}(\valprofile, \signalprofile) & = \prob{\valprofile \mid \signalprofile} \cdot \signalProfileDist^{(i)}(\signalprofile) = \prob{(\val_i, \val_{-i}) \mid \signalprofile} \cdot \signalProfileDist^{(i)}(\signalprofile) 
    = \prob{\val_i \mid \signalprofile} \cdot \prob{\val_{-i} \mid \signalprofile, \val_i}\cdot\signalProfileDist^{(i)}(\signalprofile)~,
\end{align*}
By the construction we have that $\val_{-i}$ is independent of $\signalprofile$ given $\val_i$. Thus,
\begin{align*}
    \signalscheme^{(i)}(\valprofile, \signalprofile) 
    & = \prob{\val_i \mid \signalprofile} \cdot \prob{\val_{-i} \mid \signalprofile, \val_i}\cdot\signalProfileDist^{(i)}(\signalprofile) \\
    & = \prob{\val_i \mid \signalprofile} \cdot \prob{\val_{-i} \mid \val_i}\cdot\signalProfileDist^{(i)}(\signalprofile)\\
    & = \conditionalDist_{\uniqueExpectedval(\signalprofile)}(\val_i) \cdot \modifiedpriorCDF^{(i)}(\val_{-i} \mid \val_i) \cdot \signalProfileDist^{(i)}(s)~.
\end{align*}
By construction, the $\signalprofile$ marginal is $\signalProfileDist^{(i)}$, the $\valprofile$ marginal is $\modifiedpriorCDF^{(i)}$, and for all $\signalprofile$, 
\begin{align*}
    \expect[\signalscheme^{(i)}]{\val_i \mid \signalprofile} &= \sum\nolimits_{\val_i} \val_i \cdot \frac{\sum\nolimits_{\val_{-i}} \signalscheme^{(i)}((\val_i, \val_{-i}), \signalprofile)}{\signalProfileDist^{(i)}(\signalprofile)}\\
    & = \sum\nolimits_{\val_i} \val_i \cdot \frac{\sum\nolimits_{\val_{-i}}\conditionalDist_{\uniqueExpectedval(\signalprofile)}(\val_i) \cdot \modifiedpriorCDF^{(i)}(\val_{-i} \mid \val_i) \cdot \signalProfileDist^{(i)}(s)}{\signalProfileDist^{(i)}(\signalprofile)}\\
    & = \sum\nolimits_{\val_i} \val_{i} \cdot  \conditionalDist_{\uniqueExpectedval(\signalprofile)}(\val_i)  = \expect[\conditionalDist_{\uniqueExpectedval(\signalprofile)}]{\val_i} =\uniqueExpectedval(\signalprofile)~.
\end{align*}
which satisfies condition Eqn.~\eqref{eq: winner expected value}. 
Moreover, for all $\signalprofile \in \strictsignalspace^{(i)}$, since $\valprofile \in \valspace^{(i)}$ and only bidder~$i$ has the highest signal and \eqref{eq: winner expected value} holds, we have that the signal profile $\signalprofile$ satisfies conditions in \Cref{lem:suff condi for approx full suplus extraction} and \ref{eq:construction rule}. Thus, according to \Cref{lem:global-assembly}, we can construct $\signalscheme$ based on $\signalscheme^{(i)}$ such that $\signalscheme$ satisfies conditions in \Cref{lem:suff condi for approx full suplus extraction} and \ref{eq:construction rule} for all $\signalprofile \in \strictsignalspace$.
\end{proof}

\section{Achievable Welfare Outcomes in General Information Structures}
\label{subsec:welfare outcome}

\newcommand{\inforstructureA}{\inforstructure_A}

In this section, we relax the signal independency condition imposed by the private private information and instead 
consider the full set of surplus pairs that can arise under general information structures.

\begin{theorem}
\label{thm:general full characterization}
For any pair of non-negative values $(R, B) \in \reals_+^2$ satisfying the feasibility constraints: $\worstwelfare \leq R + B \leq \optwelfare$,
there exists an information structure $\inforstructureGeneral$ and a corresponding equilibrium $\bidstrategyprofile$ such that the expected seller revenue is $\Rev(\inforstructureGeneral, \bidstrategyprofile) = R$ and the bidders' surplus is $\biddersurplus(\inforstructureGeneral, \bidstrategyprofile) = B$.

In other words, the set of payoff pairs implementable by general information structures in a second-price auction coincides exactly with the set of feasible payoffs satisfying $\worstwelfare \leq R + B \leq \optwelfare$.
\end{theorem}

As illustrated in \Cref{fig:feasible_region}, the feasible payoff set forms a trapezoid with vertices $A, B, C,$ and $D$. The upper boundary, represented by the segment $BA$, corresponds to the efficiency frontier where the sum of revenue and surplus equals the maximum welfare, $\optwelfare$. Conversely, the lower boundary (segment $CD$) represents the inefficiency frontier associated with the maximally inefficient welfare, $\worstwelfare$. 
\Cref{thm:general full characterization} shows that 
every point in the trapezoid is implementable as an equilibrium outcome under some information structure.

\wtedit{It is worth comparing \Cref{thm:general full characterization} with the welfare outcomes achievable in first price auctions. \citet{BBM-17} characterize the achievable welfare outcomes in first price auctions and show that point $A$ is typically not achievable.
Instead, there is a convex curve connecting points $C$ and $A$, and the achievable region is the convex set generated by this curve together with point $B$.
In contrast, \Cref{thm:general full characterization} shows that the second price auction is more flexible: every payoff pair that satisfies constraint $\worstwelfare \le R+B\le \optwelfare$ and gives nonnegative payoffs to both the bidders and the seller is implementable under some information structure.
Prior work on second-price auctions can be viewed as identifying particular outcomes within this trapezoid. 
For example, \citet{BBX-18} construct outcomes arbitrarily close to point $B$, \citet{CH-13} discuss outcomes with zero seller revenue, and \citet{BHMSW-22} characterize welfare outcomes when the value prior is a product measure and the information structure is required to be conditionally independent.
Complementing these implementability results, \cite{BH-04} show that, even under full information and under IPV setting with at least three bidders, there is a continuum of equilibria inducing distinct welfare and revenue outcomes.}

To establish this result, we first observe that the set of payoff pairs $(\Rev, \biddersurplus)$ implementable by general information structures is convex\footnote{\zsedit{Convexity holds because any convex combination of outcomes can be achieved by randomizing over the corresponding information structures.}}.
This convexity property simplifies our characterization substantially: it suffices to identify the four extreme points of the feasible set, namely the vertices $A, B, C,$ and $D$ in \Cref{fig:feasible_region}. 
Once the existence of information structures inducing these extreme points is established, the entire feasible region is achievable as the convex hull of these vertices.


\newcommand{\maxvalue}{\valUB}
\newcommand{\maxvalueLB}{\valUB\primed}
\newcommand{\maxvalueUB}{\valUB\doubleprimed}
\newcommand{\gaussnoise}{\zeta}
\newcommand{\noiseScale}{\tau}
\newcommand{\noisymax}{Y}
\newcommand{\postmaxmean}{\mu}
\newcommand{\winnerShift}{\Delta}
\newcommand{\lowBid}{x^{-}}
\newcommand{\highBid}{x^{+}}
\newcommand{\normalDist}{\mathcal{N}}
\newcommand{\Var}{\operatorname{Var}}

\xhdr{Full surplus extraction} 
We provide a simple correlated information structure that admits a BNE to achieve full surplus extraction.
\begin{proposition}[Maximum revenue with zero bidders' surplus]
\label{prop:general-full-extraction}
Suppose that $\maxvalue\equiv \max_{j\in[\biddernum]}\val_j$ is nondegenerate. There exists a general information structure $\inforstructure$ and an equilibrium $\bidstrategy$ such that the allocation is efficient and the seller extracts the full surplus:
\begin{align*}
    \Rev(\inforstructure,\bidstrategy)=\totalsurplus,
    \quad
    \biddersurplus(\inforstructure,\bidstrategy)=0~.
\end{align*}
\end{proposition}
The intuition behind the construction is as follows. We generate a common noisy statistic of the highest valuation. 
All non-selected bidders receive the same baseline signal, while a uniformly selected highest-valuation bidder receives an upward-shifted signal. The shift guarantees that the selected bidder submits the unique highest bid. At the same time, conditional on being selected as the winner, her expected valuation equals the baseline signal, which is exactly the second-highest bid. Thus, the winner receives zero expected surplus. A non-selected bidder, by contrast, must pay strictly more than her expected valuation in order to overtake the selected winner. These two properties make truthful bidding an equilibrium and allow the seller to extract the entire efficient welfare.

\begin{proof}[Proof of \Cref{prop:general-full-extraction}]
We prove the result by explicitly constructing the information structure and verifying the equilibrium properties. 
Let
\begin{align*}
    \maxvalueLB\equiv \min\supp(\maxvalue)~,
    \quad
    \maxvalueUB\equiv \max\supp(\maxvalue)~.
\end{align*}
Since $\maxvalue$ is nondegenerate, $\maxvalueLB<\maxvalueUB$. Fix arbitrary constants $\noiseScale>0$ and $\winnerShift>0$. Let $\gaussnoise\sim\normalDist(0,\noiseScale^2)$ be a noise independent of the valuation profile, and define $\noisymax\equiv \maxvalue+\gaussnoise$. 
Let $\postmaxmean(y)\equiv \expect{\maxvalue\condition \noisymax=y}$ denote the posterior mean of the maximum valuation after observing $\noisymax=y$.

We first record two useful properties of $\postmaxmean$. Since $\maxvalue$ is bounded and the Gaussian density is strictly positive, standard differentiation of the posterior mean gives
\begin{align*}
    \postmaxmean'(y)
    =
    \frac{\Var(\maxvalue\condition \noisymax=y)}{\noiseScale^2}
    >0~.
\end{align*}
The strict inequality follows from the nondegeneracy of $\maxvalue$ and the full support of the Gaussian noise. Hence $\postmaxmean$ is strictly increasing. Moreover, as $y\to -\infty$, the posterior of $\maxvalue$ concentrates on the lower end of its support, while as $y\to+\infty$, it concentrates on the upper end of its support. Therefore, $\postmaxmean(\R)=(\maxvalueLB,\maxvalueUB)$. We take the common signal space to be $\signalspace\equiv(\maxvalueLB,\maxvalueUB)$.

For every realized value profile $\valprofile$, select $\maxbidder\in\argmax_{j\in[\biddernum]}\val_j$ uniformly at random among the highest-valuation bidders. After drawing $\gaussnoise$, the signal profile is generated as follows:
\begin{align*}
    \signal_{\maxbidder}=\postmaxmean(\noisymax+\winnerShift),
    \quad
    \signal_j=\postmaxmean(\noisymax)
    \quad\text{for every }j\neq \maxbidder~.
\end{align*}
Let $\signalscheme(\cdot\condition \valprofile)$ denote the conditional law of the signal profile generated by this procedure, and let $\inforstructure$ be the resulting information structure. The signals are generally correlated because, conditional on $\noisymax$, all non-selected bidders receive the same signal.

We consider truthful bidding, $\bidstrategy_i(\signal_i)=\signal_i$ for every bidder $i$. We first show that $\maxbidder$ is independent of $\maxvalue$. For any measurable set $A\subseteq\supp(\maxvalue)$,
\begin{align*}
    \prob{\maxbidder=i,\maxvalue\in A}
    &=
    \expect{
    \indicator{\maxvalue\in A}
    \frac{\indicator{\val_i=\maxvalue}}{\setsize{\argmax_{j\in[\biddernum]}\val_j}}
    }~.
\end{align*}
By exchangeability, the right-hand side is the same for every bidder $i$. Summing over all bidders gives $\prob{\maxvalue\in A}$. Hence $\prob{\maxbidder=i,\maxvalue\in A}=\frac{1}{\biddernum}\prob{\maxvalue\in A}$ for every $i$ and every $A$, so $\maxbidder$ is independent of $\maxvalue$. Since $\gaussnoise$ is independent of $(\valprofile,\maxbidder)$, it follows that $\maxbidder$ is independent of $(\maxvalue,\noisymax)$.

We now verify that truthful bidding is a BNE. It suffices to consider pure bid deviations, since payoffs are linear in mixed deviations. Fix bidder $i$ and suppose she receives signal $\signal_i=x\in\signalspace$. Let $z\equiv \postmaxmean^{-1}(x)$, $\lowBid\equiv \postmaxmean(z-\winnerShift)$, and $\highBid\equiv \postmaxmean(z+\winnerShift)$. Since $\postmaxmean$ is strictly increasing, $\lowBid<x<\highBid$. There are two possible ways in which bidder $i$ can receive signal $x$.

First, suppose $i=\maxbidder$. Then $x=\postmaxmean(\noisymax+\winnerShift)$, so $\noisymax=z-\winnerShift$. Every competing bidder receives signal $\lowBid$. Since $i=\maxbidder$, bidder $i$ has value $\val_i=\maxvalue$. Using the independence of $\maxbidder$ and $(\maxvalue,\noisymax)$, we obtain
\begin{align*}
    \expect{\val_i\condition \signal_i=x,\maxbidder=i}
    =
    \expect{\maxvalue\condition \noisymax=z-\winnerShift,\maxbidder=i}
    =
    \expect{\maxvalue\condition \noisymax=z-\winnerShift}
    =
    \lowBid~.
\end{align*}
Under truthful bidding, bidder $i$ bids $x>\lowBid$, wins, and pays $\lowBid$, so her conditional expected payoff is zero. Any deviation to a bid above $\lowBid$ still wins and pays $\lowBid$; any deviation to $\lowBid$ ties at price $\lowBid$; and any deviation below $\lowBid$ loses. Thus, conditional on $i=\maxbidder$, no deviation yields a positive expected payoff.

Second, suppose $i\neq\maxbidder$. Then $x=\postmaxmean(\noisymax)$, so $\noisymax=z$. The selected winner receives signal $\highBid$, while all non-selected bidders receive signal $x$. Since $\val_i\leq \maxvalue$ state by state,
\begin{align*}
\expect{\val_i\condition \signal_i=x,\maxbidder\neq i}
\leq
\expect{\maxvalue\condition \noisymax=z,\maxbidder\neq i}
=
\expect{\maxvalue\condition \noisymax=z}
=
x,
\end{align*}
where the equality uses the independence of $\maxbidder$ and $(\maxvalue,\noisymax)$. Under truthful bidding, bidder $i$ bids $x<\highBid$, loses, and obtains payoff zero. Any deviation to a bid below $\highBid$ still loses. A deviation to $\highBid$ ties with the selected winner and, conditional on winning the tie, pays $\highBid$, yielding expected payoff at most $\frac{1}{2}(x-\highBid)<0$. A deviation above $\highBid$ wins and pays $\highBid$, yielding expected payoff at most $x-\highBid<0$. Thus, conditional on $i\neq\maxbidder$, no deviation is profitable.

Combining the two cases, truthful bidding gives bidder $i$ conditional expected payoff zero after every signal realization, while every deviation gives weakly lower conditional expected payoff. Since bidder $i$ was arbitrary, $\bidstrategy_i(\signal_i)=\signal_i$ is a BNE.


It remains to compute welfare, revenue, and bidder surplus. Since $\postmaxmean$ is strictly increasing,
\begin{align*}
    \signal_{\maxbidder}
    =
    \postmaxmean(\noisymax+\winnerShift)
    >
    \postmaxmean(\noisymax)
    =
    \signal_j
    \quad\text{for every }j\neq \maxbidder~.
\end{align*}
Thus $\maxbidder$ is the unique highest bidder. By construction, $\maxbidder$ is a highest-valuation bidder, so the allocation is efficient and 
$\totalsurplus(\inforstructure,\bidstrategy)
=
\expect{\maxvalue}
=
\totalsurplus$.
The second-highest bid is always $\postmaxmean(\noisymax)$. Thus,
\begin{align*}
    \Rev(\inforstructure,\bidstrategy)
    =
    \expect{\postmaxmean(\noisymax)}
    =
    \expect{\expect{\maxvalue\condition \noisymax}}
    =
    \expect{\maxvalue}
    =
    \totalsurplus~.
\end{align*}
Since the allocation is efficient, aggregate bidders' surplus is
$\biddersurplus(\inforstructure,\bidstrategy)
=
\totalsurplus(\inforstructure,\bidstrategy)
-
\Rev(\inforstructure,\bidstrategy)
=
0$.
This completes the proof.
\end{proof}

\xhdr{Maximal Bidder Surplus}
We first consider maximal bidder surplus that can be achieved by general information structures.

\begin{proposition}[Maximum bidders' surplus]
\label{prop:max bidders surplus}
There exists an information structure $\inforstructure$ and equilibrium $\bidstrategyprofile$ such that the bidders' surplus equals the maximum welfare, that is, $\biddersurplus(\inforstructure, \bidstrategyprofile) = \optwelfare$.
\end{proposition}

We prove \Cref{prop:max bidders surplus} by construction. 
The intuition behind our construction is as follows.
To implement an outcome that is efficient and bidder-optimal, the information structure must satisfy two seemingly conflicting requirements:
it must ensure that the item is allocated to the highest-valuation bidder (to achieve allocation efficiency) while simultaneously making the winner's payment to zero.
To achieve this, 
we construct an information structure with extreme ex post information asymmetry. 
In particular, we privately reveal the winning valuation only to the highest-valuation bidder and send a null (zero) signal to all other bidders, so their equilibrium bids are zero and the second price collapses to $0$.
In the second-price auction, this ensures that the winner bids truthfully (preserving efficiency) while the second-highest bid is fixed at $0$, allowing the winner to capture the entire surplus as an informational rent.

\begin{proof}[Proof of \Cref{prop:max bidders surplus}]
    We prove the proposition by explicitly constructing the information structure $\inforstructure$ and verifying the equilibrium properties.
    Consider the following information structure $\inforstructure$. For any realization of the valuation profile $\valprofile$, let $i^* \in \arg\max_{j} \val_j$ denote a bidder with the highest valuation (with ties broken uniformly at random). The signaling scheme is deterministic conditional on $\valprofile$:
    \begin{itemize}
        \item The winner $i^*$ receives a signal equal to their true valuation: $\signal_{i^*} = \val_{i^*} = \max(\valprofile)$.
        \item All other bidders $j \neq i^*$ receive a null signal: $\signal_j = 0$.
    \end{itemize}
    We now show that the truthful bidding $\bidstrategyprofile_i(\signal_i) = \signal_i$ is an equilibrium.
    Consider a bidder $i$ who receives the high signal $\signal_i = \val_i > 0$. This bidder infers that they have the highest valuation and that all other bidders $j \neq i$ have received signals $\signal_j = 0$. 
    Assuming all other bidders' bid truthfully where $\bid_j \equiv 0$ for all $j\neq i$. 
    By bidding $\bid_i = \signal_i$, bidder $i$ wins the item and pays the second-highest bid, which is $0$. 
    The payoff is $\val_i - 0 > 0$. 
    Deviating to a lower bid $\bid < 0$ is impossible (bids must be non-negative), and bidding $0 \le \bid < \signal_i$ gives the same outcome (winning at price $0$) unless $\bid=0$ and the tie-breaking rule will lead to a loss. 
    Bidding higher than $\signal_i$ does not change the outcome or payment. 
    Thus, bidding $\signal_i$ is a \wtedit{best response}.
    Consider a bidder $j$ who receives the signal $\signal_j = 0$. This bidder infers that there exists some other bidder $k$ with valuation $\val_k \geq \val_j$ who has received a signal $\signal_k = \val_k$ and will bid $\bid_k = \val_k$. 
    Under the truthful bidding, bidder $j$ bids $0$ and loses, obtaining a payoff of $0$. To win the item, bidder $j$ would need to bid $\bid > \val_k$. However, the resulting price would be $\val_k$. Since $\val_j \le \val_k$, the payoff from winning would be $\val_j - \val_k \le 0$. Thus, there is no profitable deviation from bidding $0$.

    Under this equilibrium $\bidstrategyprofile$, for any valuation profile $\valprofile$: 
    \begin{enumerate}
        \item Efficiency: The item is always allocated to bidder $i^*$ who holds the maximum valuation $\max(\valprofile)$. Therefore, the allocation is fully efficient, and the total generated welfare is $\optwelfare$.
        \item Revenue and Surplus: The winning bid is $\max(\valprofile)$ and the second-highest bid is always $0$. Consequently, the seller's revenue is $\Rev(\inforstructure, \bidstrategyprofile) = 0$ for every realized signal profile. 
        The bidders' surplus is the total welfare minus revenue:
        \[
        \biddersurplus(\inforstructure, \bidstrategyprofile) = \optwelfare - 0 = \optwelfare~.
        \]
    \end{enumerate}
    This completes the proof.
\end{proof}

\xhdr{Minimum Efficiency}
Thus far, our analysis has led us to equilibrium in which the allocation of the good was efficient, so that the welfare outcome lay on the northeast frontier of \Cref{fig:feasible_region}.
We next characterize two extreme points among all maximally inefficient outcomes.

\begin{proposition}[Minimum bidders' surplus with zero revenue]
\label{prop:min efficiency zero rev}
There exists an information structure $\inforstructure$ and equilibrium $\bidstrategyprofile$ such that the bidders' surplus equals the maximally inefficient welfare, that is, $\biddersurplus(\inforstructure, \bidstrategyprofile) = \worstwelfare$ and the seller's revenue is zero.
\end{proposition}

The intuition behind the constructed information structure is to create an information structure that forces the lowest-valuation bidder to win, so as to minimize welfare, while ensuring the second-highest bid is always $0$. 
We do this by giving the lowest-valuation bidder a uniquely high signal $\valUB$, and giving every other bidder a null signal $0$, so everyone can infer who is designated to win. 
Since $\valUB$ is chosen as the maximum valuation, no other bidder is willing to outbid $\valUB$: winning would require paying at least $\valUB$, which yields a non-positive payoff. 
Thus, truthful bidding is an equilibrium in which the lowest-valuation bidder always wins, and total welfare equals $\worstwelfare$,
\zsedit{and the winner's payment is always zero, leaving bidders with full surplus.}
\begin{proof}[Proof of \Cref{prop:min efficiency zero rev}]
    We prove the proposition by explicitly constructing the information structure $\inforstructure$ and verifying the equilibrium properties.
    Let \wtedit{$\valUB \triangleq \max \valspace$ be the maximum value among all possible valuations.} 
    For any realization of the valuation profile $\valprofile$, let $i_* \in \arg\min_{j} \val_j$ denote a bidder with the lowest valuation. 
    \wtedit{If multiple bidders share the same lowest valuation, we select $i_*$ uniformly at random among them.}
    The signaling scheme is then deterministic conditional on $(\valprofile, i_*)$:
    \begin{itemize}
        \item The lowest-valuation bidder $i_*$ receives a high signal: $\signal_{i_*} = \valUB$.
        \item All other bidders $j \neq i_*$ receive a null signal: $\signal_j = 0$.
    \end{itemize}
    We now show that the truthful bidding $\bidstrategyprofile_i(\signal_i) = \signal_i$ is an equilibrium.
    Consider a bidder $i$ who receives the high signal $\signal_i = \valUB > 0$. This bidder infers that they are the designated winner and that all other bidders $j \neq i$ have received signals $\signal_j = 0$. 
    Assuming all other bidders' bid truthfully where $\bid_j \equiv 0$ for all $j\neq i$.     
    By bidding $\valUB$, bidder $i$ wins the item at price $0$. The payoff is $\val_{i} - 0 \ge 0$. 
    Deviating to a lower bid $\bid < \valUB$ is only profitable if it lowers the price (which is impossible as price is already $0$) or changes the allocation (i.e., losing the item but yields payoff $0$). 
    Thus, bidding $\signal_i = \valUB$ is \wtedit{a best response}.
    Consider a bidder $j$ who receives $\signal_j = 0$. They infer that the lowest-valuation bidder $i_*$ has received signal $\valUB$ and will bid $\valUB$. 
    To win the item, bidder $j$ must bid $\bid = \valUB$ (note that \wtedit{we assume bidder are never bids an amount which is sure to be strictly greater than her value}). The resulting price would be still $\valUB$ (as the bidder $i_*$ would still bid $\valUB$). 
    Thus, the resulting payoff is at most $\sfrac{1}{2} \cdot (\val_j - \valUB) \le 0 $
    as $\valUB$ is an upper bound on values.
    Thus, there is no profitable deviation from bidding $0$.

    Under this equilibrium $\bidstrategyprofile$, for any valuation profile $\valprofile$:
    \begin{enumerate}
        \item Efficiency: The item is always allocated to a bidder  who holds the lowest valuation. 
        The expected total welfare is minimized at $\expect{\min_j \val_j} = \worstwelfare$.
        \item Revenue and Surplus: The winning bid is $\valUB$ and the second-highest bid is $0$. The revenue is $\Rev(\inforstructure, \bidstrategyprofile) = 0$. The bidders' surplus is $\worstwelfare - 0 = \worstwelfare$.
    \end{enumerate}
    This completes the proof.
\end{proof}

Finally, we can also show that there exists an information structure that achieves the most inefficient outcome with zero bidders' surplus.
\begin{proposition}[Minimum revenue with zero bidders' surplus]
\label{prop:min efficiency zero bidders surplus}
There exists an information structure $\inforstructure$ and equilibrium $\bidstrategyprofile$ such that the seller's revenue equals to the maximally inefficient welfare, that is, $\Rev(\inforstructure, \bidstrategyprofile) = \worstwelfare$ and the bidders' surplus is zero.
\end{proposition}
\begin{proof}[Proof of \Cref{prop:min efficiency zero bidders surplus}]
    We prove the proposition by explicitly constructing the information structure $\inforstructure$ and verifying the equilibrium properties.
    Recall that \wtedit{$\valUB \triangleq \max \valspace$ be the maximum value among all possible valuations.} 
    For any realization of the valuation profile $\valprofile$, let $i_* \in \arg\min_{j} \val_j$ denote a bidder with the lowest valuation, and 
    let $i^* \in \arg\max_j \val_j$ be a bidder with the highest valuation.
    \wtedit{If multiple bidders share the same lowest (resp.\ highest) valuation, we select $i_*$ (resp.\ $i^*$) uniformly at random among them.
    If $\min_j \val_j = \max_j \val_j$, then all bidders have the same valuation; in this case, we first select $i_*$ uniformly at random among $[\biddernum]$, and then select $i^*\in [\biddernum]\setminus\{i_*\}$ uniformly at random.}
    The signaling scheme is then deterministic conditional on $(\valprofile, i_*, i^*)$:
    \begin{itemize}
        \item The lowest-valuation bidder $i_*$ receives a signal: $\signal_{i_*} = (W, \valUB)$.
        \item The highest-valuation bidder $i^*$ 
        receives a signal equal to the lowest valuation: $\signal_{i^*} = (P, \val_{i_*})$.
        \item All other bidders $j$ receive a null signal: $\signal_j = 0$.
    \end{itemize}
    We now show that the following bidding strategy $\bidstrategy_i$ is an equilibrium:
    \begin{align}
        \label{eq:bidstrategy mineff}
        \bidstrategy_i(\signal_i)= 
        \begin{cases}
        \valUB, & \text { if } \signal_i = (W, \valUB), \\ 
        \val_{i_*}, & \text { if } \signal_i = (P, \val_{i_*})\\
        0, & \text { if } \signal_i = 0~.
        \end{cases}
    \end{align}
    Consider a bidder $i$ who receives the signal $\signal_i = (W, \valUB)$. 
    This bidder infers that they are the designated winner and that all other bidders $j \neq i$ either receive the signal $(P, \val_{i_*})$ or the signal $0$.
    Under $\bidstrategy_{-i}$, there is exactly one bidder (namely $i^*$) who bids $\val_{i_*}$, while all remaining bidders bid $0$.
    Thus, if bidder $i$ bids $\valUB$, they win the item and the price equals the second-highest bid, namely $\val_{i_*}$.
    Since bidder $i$ has the lowest valuation, the resulting payoff for this bidder is $\val_{i} - \val_{i_*} = \val_{i_*} - \val_{i_*} = 0$.
    If bidder $i$ deviates to any bid $\bid < \valUB$, then either $\bid < \val_{i_*}$, in which case they lose the item and obtain payoff $0$, or $\bid \in [\val_{i_*}, \valUB)$, in which case they still win and pay $\val_{i_*}$, yielding payoff $0$. Thus, no deviation yields a strictly higher payoff, and bidding $\valUB$ is a best response.

    Next, consider a bidder $j$ who receives the signal $(P, \val_{i_*})$. 
    This bidder infers that they are the designated highest-valuation bidder $i^*$ and that bidder $i_*$ will bid $\valUB$. Under $\bidstrategy$, bidding $\val_{i_*}$ ensures that bidder $j$ loses the item and obtains payoff $0$. If bidder $j$ deviates to any bid $\bid < \valUB$, they still lose and obtain payoff $0$. If bidder $j$ deviates to $\bid = \valUB$, then they tie with bidder $i_*$ for the highest bid; conditional on winning the tie, the price is $\valUB$, and the payoff is $\val_j  - \valUB \le 0$. Thus, the expected payoff from deviating to $\bid = \valUB$ is at most $0$, so there is no profitable deviation from bidding $\val_{i_*}$.

    Finally, consider a bidder $k$ who receives the null signal $\signal_k = 0$. Under $\bidstrategy$, bidding $0$ yields payoff $0$. If bidder $k$ deviates to any bid $\bid < \valUB$, they still lose to bidder $i_*$'s bid and obtain payoff $0$. If bidder $k$ deviates to $\bid = \valUB$, then they tie with bidder $i_*$; conditional on winning, the price is $\valUB$ and the payoff is $\val_k  -\valUB \le 0$. Thus, no deviation yields a strictly positive payoff, and bidding $0$ is a best response. Therefore, the bidding strategy defined in Eqn.~\eqref{eq:bidstrategy mineff} is indeed an equilibrium.

    Under this equilibrium $\bidstrategy$, for any valuation profile $\valprofile$:
    \begin{enumerate}
        \item Efficiency:The item is always allocated to a bidder who holds the lowest valuation. 
        The expected total welfare is minimized at $\expect{\min_j \val_j} = \worstwelfare$.
        \item Revenue and Surplus: The winning bidder always pays $\val_{i_*}$. The revenue is $\expect{\min_j \val_j} = \worstwelfare$. The bidder surplus is $\val_{i_*} - \val_{i_*} = 0$.
    \end{enumerate}
    This completes the proof.
\end{proof}


\bibliography{ref}

@article{CLW-26,
  title={Affiliated Multi-Unit All-Pay Auctions: Equilibrium and Full Surplus Extraction},
  author={Chen, Bo and Lu, Jingfeng and Wang, Zijia},
  journal={Available at SSRN 5918042},
  year={2026}
}

@article{DTWZ-25,
  title={Optimal Calibrated Signaling in Digital Auctions},
  author={Du, Zhicheng and Tang, Wei and Wang, Zihe and Zhang, Shuo},
  journal={arXiv preprint arXiv:2507.17187},
  year={2025}
}

@article{BH-04,
  title={All equilibria of the Vickrey auction},
  author={Blume, Andreas and Heidhues, Paul},
  journal={Journal of economic Theory},
  volume={114},
  number={1},
  pages={170--177},
  year={2004},
  publisher={Elsevier}
}

@inproceedings{EFGMM-23,
  title={Constant approximation for private interdependent valuations},
  author={Eden, Alon and Feldman, Michal and Goldner, Kira and Mauras, Simon and Mohan, Divyarthi},
  booktitle={2023 IEEE 64th Annual Symposium on Foundations of Computer Science (FOCS)},
  pages={148--163},
  year={2023},
  organization={IEEE}
}

@inproceedings{EFFG-18,
  title={Interdependent Values without Single-Crossing},
  author={Eden, Alon and Feldman, Michal and Fiat, Amos and Goldner, Kira},
  booktitle={Proceedings of the 19th ACM Conference on Economics and Computation (EC)},
  pages={369--369},
  year={2018}
}

@article{MP-15,
  title={Implementation with interdependent valuations},
  author={McLean, Richard P and Postlewaite, Andrew},
  journal={Theoretical Economics},
  volume={10},
  number={3},
  pages={923--952},
  year={2015},
  publisher={Wiley Online Library}
}

@article{MP-04,
  title={Informational size and efficient auctions},
  author={McLean, Richard and Postlewaite, Andrew},
  journal={The Review of Economic Studies},
  volume={71},
  number={3},
  pages={809--827},
  year={2004},
  publisher={Wiley-Blackwell}
}

@article{KM-86,
  title={On the strategic stability of equilibria},
  author={Kohlberg, Elon and Mertens, Jean-Francois},
  journal={Econometrica: Journal of the Econometric Society},
  pages={1003--1037},
  year={1986},
  publisher={JSTOR}
}

@incollection{B-88,
  title={Reexamination of the perfectness concept for equilibrium points in extensive games},
  author={Bielefeld, R Selten},
  booktitle={Models of strategic rationality},
  pages={1--31},
  year={1988},
  publisher={Springer}
}

@inproceedings{CH-13,
  title={Auctions with unique equilibria},
  author={Chawla, Shuchi and Hartline, Jason D},
  booktitle={Proceedings of the 14th ACM conference on Electronic commerce (EC)},
  pages={181--196},
  year={2013}
}

@inproceedings{GS-09,
  title={Bayes-Nash equilibria of the generalized second price auction},
  author={Gomes, Renato D and Sweeney, Kane S},
  booktitle={Proceedings of the 10th ACM conference on Electronic Commerce (EC)},
  pages={107--108},
  year={2009}
}

@inproceedings{B+26,
  title={Optimal Type-Dependent Liquid Welfare Guarantees for Autobidding Agents with Budgets},
  author={Baldeschi, Riccardo Colini and Klumper, Sophie and Kroll, Twan and Leonardi, Stefano and Sch{\"a}efer, Guido and Tsikiridis, Artem},
  booktitle={Proceedings of the 2026 Annual ACM-SIAM Symposium on Discrete Algorithms (SODA)},
  pages={1795--1823},
  year={2026},
  organization={SIAM}
}

@article{DMMZZ-24,
  title={Efficiency of the first-price auction in the autobidding world},
  author={Deng, Yuan and Mao, Jieming and Mirrokni, Vahab and Zhang, Hanrui and Zuo, Song},
  journal={Advances in Neural Information Processing Systems},
  volume={37},
  pages={139270--139293},
  year={2024}
}

@article{BS-03,
  title={Rationalizable bidding in first-price auctions},
  author={Battigalli, Pierpaolo and Siniscalchi, Marciano},
  journal={Games and Economic Behavior},
  volume={45},
  number={1},
  pages={38--72},
  year={2003},
  publisher={Elsevier}
}

@inproceedings{FLN-16,
  title={Correlated and coarse equilibria of single-item auctions},
  author={Feldman, Michal and Lucier, Brendan and Nisan, Noam},
  booktitle={International Conference on Web and Internet Economics},
  pages={131--144},
  year={2016},
  organization={Springer}
}

@inproceedings{AB-25,
  title={On the uniqueness of bayesian coarse correlated equilibria in standard first-price and all-pay auctions},
  author={Ahunbay, Mete {\c{S}}eref and Bichler, Martin},
  booktitle={Proceedings of the 2025 Annual ACM-SIAM Symposium on Discrete Algorithms (SODA)},
  pages={2491--2537},
  year={2025},
  organization={SIAM}
}

@article{ABST-21,
  title={Feasible joint posterior beliefs},
  author={Arieli, Itai and Babichenko, Yakov and Sandomirskiy, Fedor and Tamuz, Omer},
  journal={Journal of Political Economy},
  volume={129},
  number={9},
  pages={2546--2594},
  year={2021},
  publisher={The University of Chicago Press Chicago, IL}
}

@article{I-20,
  title={Online privacy and information disclosure by consumers},
  author={Ichihashi, Shota},
  journal={American Economic Review},
  volume={110},
  number={2},
  pages={569--595},
  year={2020},
  publisher={American Economic Association 2014 Broadway, Suite 305, Nashville, TN 37203}
}

@article{SY-25,
  title={Non-Discriminatory Personalized Pricing},
  author={Strack, Philipp and Yang, Kai Hao},
  journal={arXiv preprint arXiv:2506.20925},
  year={2025}
}

@article{EEM-21,
  title={Bayesian privacy},
  author={Eilat, Ran and Eliaz, Kfir and Mu, Xiaosheng},
  journal={Theoretical Economics},
  volume={16},
  number={4},
  pages={1557--1603},
  year={2021},
  publisher={Wiley Online Library}
}

@inproceedings{SY-22,
  title={Information design for differential privacy},
  author={Schmutte, Ian M and Yoder, Nathan},
  booktitle={Proceedings of the 23rd ACM Conference on Economics and Computation (EC)},
  pages={1142--1143},
  year={2022}
}

@article{AMMO-22,
  title={Too much data: Prices and inefficiencies in data markets},
  author={Acemoglu, Daron and Makhdoumi, Ali and Malekian, Azarakhsh and Ozdaglar, Asu},
  journal={American Economic Journal: Microeconomics},
  volume={14},
  number={4},
  pages={218--256},
  year={2022},
  publisher={American Economic Association 2014 Broadway, Suite 305, Nashville, TN 37203-2425}
}

@inproceedings{LM-20,
  title={Data and incentives},
  author={Liang, Annie and Madsen, Erik},
  booktitle={Proceedings of the 21st ACM Conference on Economics and Computation (EC)},
  pages={41--42},
  year={2020}
}

@inproceedings{DMNS-06,
  title={Calibrating noise to sensitivity in private data analysis},
  author={Dwork, Cynthia and McSherry, Frank and Nissim, Kobbi and Smith, Adam},
  booktitle={Theory of cryptography conference},
  pages={265--284},
  year={2006},
  organization={Springer}
}

@inproceedings{D-06,
  title={Differential privacy},
  author={Dwork, Cynthia},
  booktitle={International colloquium on automata, languages, and programming},
  pages={1--12},
  year={2006},
  organization={Springer}
}

@article{DR-14,
  title={The Algorithmic Foundations of Differential Privacy},
  author={Dwork, Cynthia and Roth, Aaron},
  journal={Foundations and Trends{\textregistered} in Theoretical Computer Science},
  volume={9},
  number={3-4},
  pages={211--407},
  year={2014},
  publisher={Now Publishers Boston—Delft}
}

@article{BHM-26,
  title={Screening with persuasion},
  author={Bergemann, Dirk and Heumann, Tibor and Morris, Stephen},
  journal={Journal of Political Economy},
  volume={134},
  number={2},
  pages={000--000},
  year={2026},
  publisher={The University of Chicago Press Chicago, IL}
}

@article{M-81,
  title={Optimal auction design},
  author={Myerson, Roger B},
  journal={Mathematics of operations research},
  volume={6},
  number={1},
  pages={58--73},
  year={1981},
  publisher={INFORMS}
}

@article{CY-23,
  title={Information design in optimal auctions},
  author={Chen, Yi-Chun and Yang, Xiangqian},
  journal={Journal of Economic Theory},
  volume={212},
  pages={105710},
  year={2023},
  publisher={Elsevier}
}

@article{K-20,
  title={Information disclosure and full surplus extraction in mechanism design},
  author={Kr{\"a}hmer, Daniel},
  journal={Journal of Economic Theory},
  volume={187},
  pages={105020},
  year={2020},
  publisher={Elsevier}
}

@article{SY-24,
  title={Privacy-Preserving Signals},
  author={Strack, Philipp and Yang, Kai Hao},
  journal={Econometrica},
  volume={92},
  number={6},
  pages={1907--1938},
  year={2024},
  publisher={Wiley Online Library}
}

@article{BP-07,
  title={Information structures in optimal auctions},
  author={Bergemann, Dirk and Pesendorfer, Martin},
  journal={Journal of economic theory},
  volume={137},
  number={1},
  pages={580--609},
  year={2007},
  publisher={Elsevier}
}

@inproceedings{BBDPXZ-25,
  title={Data-Driven Mechanism Design: Jointly Eliciting Preferences and Information},
  author={Bergemann, Dirk and Bojko, Marek and Duetting, Paul and Paes Leme, Renato and Xu, Haifeng and Zuo, Song},
  booktitle={Proceedings of the 26th ACM Conference on Economics and Computation (EC)},
  pages={507--507},
  year={2025}
}

@article{EFGPT-14,
  title={Signaling schemes for revenue maximization},
  author={Emek, Yuval and Feldman, Michal and Gamzu, Iftah and PaesLeme, Renato and Tennenholtz, Moshe},
  journal={ACM Transactions on Economics and Computation (TEAC)},
  volume={2},
  number={2},
  pages={1--19},
  year={2014},
  publisher={ACM New York, NY, USA}
}

@inproceedings{BS-12,
  title={Send mixed signals: earn more, work less},
  author={Bro Miltersen, Peter and Sheffet, Or},
  booktitle={Proceedings of the 13th ACM conference on electronic commerce (EC)},
  pages={234--247},
  year={2012}
}

@inproceedings{BDPZ-22,
  title={Calibrated click-through auctions},
  author={Bergemann, Dirk and Duetting, Paul and Paes Leme, Renato and Zuo, Song},
  booktitle={Proceedings of the ACM Web Conference 2022},
  pages={47--57},
  year={2022}
}

@inproceedings{BBX-18,
  title={Targeting and signaling in ad auctions},
  author={Badanidiyuru, Ashwinkumar and Bhawalkar, Kshipra and Xu, Haifeng},
  booktitle={Proceedings of the twenty-ninth annual ACM-SIAM symposium on discrete algorithms (SODA)},
  pages={2545--2563},
  year={2018},
  organization={SIAM}
}

@article{MW-82,
  title={A theory of auctions and competitive bidding},
  author={Milgrom, Paul R and Weber, Robert J},
  journal={Econometrica: Journal of the Econometric Society},
  pages={1089--1122},
  year={1982},
  publisher={JSTOR}
}

@article{MR-92,
  title={Correlated information and mecanism design},
  author={McAfee, R Preston and Reny, Philip J},
  journal={Econometrica: Journal of the Econometric Society},
  pages={395--421},
  year={1992},
  publisher={JSTOR}
}

@article{CM-88,
  title={Full extraction of the surplus in Bayesian and dominant strategy auctions},
  author={Cr{\'e}mer, Jacques and McLean, Richard P},
  journal={Econometrica: Journal of the Econometric Society},
  pages={1247--1257},
  year={1988},
  publisher={JSTOR}
}

@article{BBM-17,
  title={First-price auctions with general information structures: Implications for bidding and revenue},
  author={Bergemann, Dirk and Brooks, Benjamin and Morris, Stephen},
  journal={Econometrica},
  volume={85},
  number={1},
  pages={107--143},
  year={2017},
  publisher={Wiley Online Library}
}

@article{KFO-25,
  title={Private private information},
  author={He, Kevin and Sandomirskiy, Fedor and Tamuz, Omer},
  journal={Journal of Political Economy},
  volume={134},
  number={5},
  pages={1561--1606},
  year={2026},
  publisher={The University of Chicago Press Chicago, IL}
}

@article{BHMSW-22,
  title={Optimal information disclosure in classic auctions},
  author={Bergemann, Dirk and Heumann, Tibor and Morris, Stephen and Sorokin, Constantine and Winter, Eyal},
  journal={American Economic Review: Insights},
  volume={4},
  number={3},
  pages={371--388},
  year={2022},
  publisher={American Economic Association 2014 Broadway, Suite 305, Nashville, TN 37203}
}
\end{document}